\documentclass[a4paper,onecolumn,unpublished,11pt]{quantumarticle}
\pdfoutput=1

\usepackage{aux_tikz}

\begin{document}

\renewcommand{\ref}[1]{{\color{red}[use \texttt{cref} not \texttt{ref}]}}

\title{Emergent causal order and time direction: bridging causal models and tensor networks}

\author{Carla Ferradini$^1$}
\author{Giulia Mazzola$^1$}
\author{V.\ Vilasini$^{1,2}$}
\affiliation{$^1$Institute for Theoretical Physics, ETH Zurich, 8093 Zürich, Switzerland\\
$^2$Université Grenoble Alpes, Inria, 38000 Grenoble, France}
\maketitle

\begin{abstract}
Can the direction of time and the causal structure of space-time be inferred from operational principles? Causal models and tensor networks offer complementary perspectives: the former encodes cause-effect relations via directed graphs, with intrinsic ordering; the latter describes multipartite systems on undirected graphs, without presupposing directionality. We construct two-way mappings between these two frameworks, linking direction agnostic correlation functions and operational notions of signalling. This clarifies the operational meaning of causal influence in tensor networks and introduces discrete ``space-time rotations'' of causal models which preserve signalling relations. Applying our framework to holographic tensor networks, we use tools from causal inference, like graph-separation, to analyse emergent causal structures. By permitting cyclic and indefinite causal structures, our results enable transfer of techniques across tensor networks and a range of causality frameworks.
\end{abstract}
\newpage
\tableofcontents

\newpage
\section{Introduction}

A central question in the foundations of physics is whether the direction of time --- and, more generally, properties of space-time --- can be inferred rather than assumed a priori. In standard causal reasoning, temporal directionality is built into the formalism: causes precede effects, interventions propagate forward, and explanatory structure reflects this asymmetry. Causal models~\cite{pearl_2009,Barrett_2019,Henson_2014,Barrett_2021,Quantum_paper,forre_2018,VilasiniRennerPRA,VilasiniRennerPRL,Costa_2016} (both classical and quantum) exemplify this viewpoint. They describe relations between processes using directed graphs whose arrows encode causal ordering, and whose building blocks possess a predefined input-output structure. Even in extensions allowing cycles~\cite{Bongers_2021, Barrett_2021,Quantum_paper}, a local orientation remains embedded in the components themselves, so that directionality is part of the description from the outset. Although causality in such models can be formalised from an information-theoretic perspective, without assuming a background space-time, this is done through the connectivity of channels and systems, whose distinction between inputs and outputs induces directionality.

In many areas of physics, however, causal structure plays a different role. In relativity and quantum field theory, causal relations are induced by the geometry of space-time, with light cones determining which events may influence one another. In candidate theories of quantum gravity, space-time itself is expected to be dynamical and possibly subject to quantum superposition, motivating the question of whether causal and geometric structure might emerge from more primitive quantum correlations. Tensor networks~\cite{Cotler_2018,Vidal_2003,Verstraete_2004,Levin_2007,Qi_2018}, since no causal structure is assumed as a primitive, provide a natural framework in which to explore this possibility. As graphical representations of many-body quantum states and processes, tensor networks are defined on undirected graphs and treat all tensor legs symmetrically, without presupposing any global or local time direction. They encode correlations and compositional structure while remaining agnostic about causal order. In this sense, tensor networks offer a setting in which space-time or causal structure, if present, must emerge rather than be imposed.

Each framework, however, faces certain challenges in addressing the operational emergence of time. In tensor networks, a notion of quantum causal influence based on correlation functions has been proposed~\cite{Cotler_2019}. While this is an intriguing proposal, its relation to established operational notions of causality and causal reasoning remains unclear. In particular, is not evident whether such definitions capture genuine causal structure or rather reflect patterns of correlation, nor whether they support the kind of operational reasoning characteristic of causal models. Causal models, by contrast, encode directionality at the structural level. While this feature underlies their explanatory power, it is sometimes regarded as conceptually delicate when applied to fundamental physics, where many underlying laws are time-reversal symmetric. From this perspective, the asymmetry of directed graphs may appear to reflect modelling choices --- such as the role of interventions or agents --- rather than fundamental physical structure. How to understand or adapt such directionality in scenarios with dynamical or indefinite causal order remains an open question.

Addressing these challenges calls for a framework that relates directionally agnostic descriptions of correlations to operational notions of causality. In this work, we establish a precise connection between tensor networks and quantum causal models by constructing explicit mappings between the two. These mappings translate correlation functions in tensor networks into probabilistic quantities in, possibly cyclic, causal models and relate causal influence to signalling relations. By linking a structure without predefined orientation to one in which causal direction is explicit, we provide a setting in which questions about the emergence, interpretation, and operational meaning of time direction can be formulated sharply and analysed systematically, while being able to apply tools from both frameworks consistently. In doing so, we aim to build a bridge between approaches to causal inference and questions on emergent space-time, and to clarify how operational notions of causality can arise within directionally agnostic quantum frameworks.

\begin{figure}[h!]
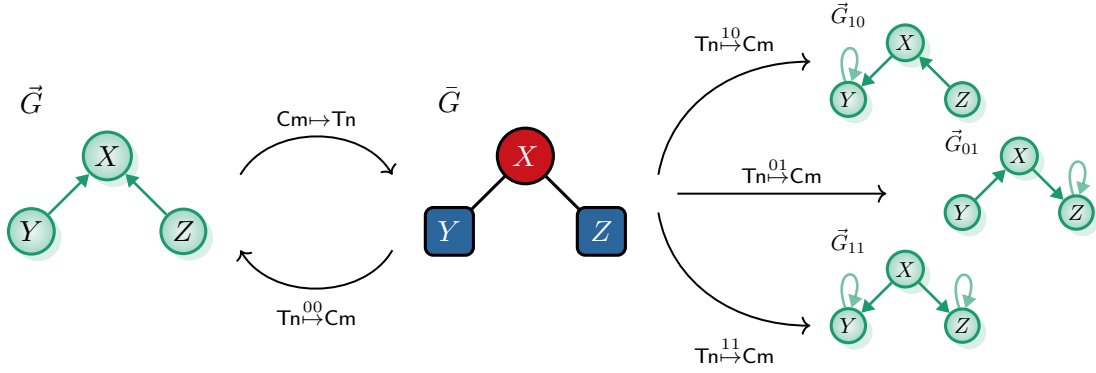

\centering
 $\centertikz{
        \node at (-6.5,0){};
        \begin{scope}[shift = {(-5.5,0)},transform shape]
            \node at (-1,0.75) {$\graphnamedir$};
            \node[unode] (x) at (0,0) {$X$};
            \node[unode] (y) at (-1,-1) {$Y$};
            \node[unode] (z) at (1,-1) {$Z$};
            \draw[qleg] (y) --node[left] {} (x);
            \draw[qleg] (z) --node[right] {} (x);
        \end{scope}  
        \draw[line width =0.75pt,->,color=black] (-3.75,-.25) to [in =120, out=60] node[above] {\indexstyle{\cm\mapsto\tn}}(-1.75,-0.25);
        \draw[line width =0.75pt,<-,color=black] (-3.75,-1.25) to [in=240, out=300] node[below] {\indexstyle{\tn\stackrel{00}{\mapsto}\cm}}(-1.75,-1.25);
        \begin{scope}[shift = {(0,0)},transform shape]  
            \node at (-1,0.75) {$\graphnameud$};
            \draw[mybaseline] (0,0) -- (-1,-1);
            \draw[mybaseline] (0,0) -- (1,-1);
            \node[tnode, fill=red, minimum size=21pt] (x) at (0,0){\color{white}{$X$}};
            \node[tnrect, fill=blue, minimum size=18pt] (y) at (-1,-1){\color{white}{$Y$}};
            \node[tnrect, fill=blue, minimum size=18pt] (z) at (1,-1){\color{white}{$Z$}};                  
        \end{scope}
        \draw[line width =0.75pt,->,color=black] (1.75,-.25) to [in =180, out=80] (3.75,1.25) node[above,xshift=-1cm] {\indexstyle{\tn\stackrel{10}{\mapsto}\cm}};
        
        \draw[line width =0.75pt,->,color=black] (2,-0.5) --node[above] {\indexstyle{\tn\stackrel{01}{\mapsto}\cm}} (4.75,-0.5);
        \draw[line width =0.75pt,->,color=black] (1.75,-.75) to [in =180, out=280] (3.75,-2.25)node[below,xshift=-1cm] {\indexstyle{\tn\stackrel{11}{\mapsto}\cm}};
        
        \begin{scope}[shift={(6.5,0)}]
            \begin{scope}[scale = 0.75, transform shape,shift={(-2,2)}]
            \node at (-1,0.5) {$\graphnamedir_{10}$};
                \node[unode] (x) at (0,0) {$X$};
                \node[unode] (y) at (-1,-1) {$Y$};
                \node[unode] (z) at (1,-1) {$Z$};
                \draw[qleg] (x) -- (y);
                \draw[qleg] (z) -- (x);
                \draw[qleg,opacity=0.6] (y.100) .. controls ($(y.100)+(-0.25,0.75)$) and ($(y.80)+(0.25,0.75)$) .. (y.80);
            \end{scope}
            \begin{scope}[scale = 0.75, transform shape]
                \node at (-1,0.25) {$\graphnamedir_{01}$};
                \node[unode] (x) at (0,0) {$X$};
                \node[unode] (y) at (-1,-1) {$Y$};
                \node[unode] (z) at (1,-1) {$Z$};
                \draw[qleg] (y) -- (x);
                \draw[qleg] (x) -- (z);
                \draw[qleg,opacity=0.6] (z.100) .. controls ($(z.100)+(-0.25,0.75)$) and ($(z.80)+(0.25,0.75)$) .. (z.80);
            \end{scope}
            \begin{scope}[scale = 0.75, transform shape,shift={(-2,-2)}]
            \node at (-1,0.5) {$\graphnamedir_{11}$};
                \node[unode] (x) at (0,0) {$X$};
                \node[unode] (y) at (-1,-1) {$Y$};
                \node[unode] (z) at (1,-1) {$Z$};
                \draw[qleg] (x) -- (y);
                \draw[qleg] (x) -- (z);
                \draw[qleg,opacity=0.6] (y.100) .. controls ($(y.100)+(-0.25,0.75)$) and ($(y.80)+(0.25,0.75)$) .. (y.80);
                \draw[qleg,opacity=0.6] (z.100) .. controls ($(z.100)+(-0.25,0.75)$) and ($(z.80)+(0.25,0.75)$) .. (z.80);
            \end{scope}  
        \end{scope}
    }$  
    \caption{{\bf Visual intuition behind main results}\quad Our work defines operationally motivated mappings between (possibly cyclic) quantum causal models and tensor networks, going in both directions. Through this, a causal model on the directed graph $\graphnamedir$ can be uniquely mapped to a tensor network on the undirected graph $\graphnameud$. From the tensor network on $\graphnameud$, we can obtain four different causal models, one for each choice of directions of edges. The choice associated with the directions of $\graphnamedir$ will lead again to the original causal model. While, other choices, which lead to causal models on the directed graphs on the right, might require self-loops to give a valid model. Such possibility can be easily checked through evaluating a partial trace. We show that the mappings preserve certain useful properties, allow to interpret alternative causal models obtained through the mapping as ``rotations'' of the original model, and to apply graph-separation theorems from causal models to analyse tensor networks. \vspace{-0.25cm}}
    \label{fig:one}
\end{figure}

\subsection{Summary of contributions and structure of the paper}
\label{sec: summary_contrib}
In this work, we present mappings from causal models to tensor networks and vice versa which link signalling relations in the former to causal influence in the latter. In~\Cref{sec: TNrev}, we provide a review on tensor network and the definition of causal influence proposed in~\cite{Cotler_2019}, collecting clarifications and open questions regarding it. In~\Cref{sec: CMrev}, we review a cyclic causal modelling framework presented in~\cite{Quantum_paper,Ferradini_2025C}, together with definitions of interventions and signalling. The main contributions of the paper are the following:
\begin{itemize}
    \item \textbf{Mappings between tensor networks and causal models:} In~\Cref{sec: mappings}, we propose concrete mappings which tightly link properties of the two frameworks (see \cref{fig:one} for an illustrated example). The features of these mappings naturally portray the absence of a predefined direction of inputs and outputs in tensor networks:  a given causal model can always be uniquely mapped to a tensor network on the same graph deprived of directions (\Cref{sec:CMtoTNmap}); while for tensor networks to causal models we present two mappings, the second being a strict generalisation of the first. The first mapping (\Cref{sec:TNtoCMmap}) preserves the graph structure of the tensor network up to augmenting it with direction, but, depending on the tensors of the network can lead to one, multiple or no causal models. The generalised mapping~(\Cref{sec:TNtoCMmapgen}), on the other hand, allows to construct a valid causal model for any choice of directions up to introducing cyclicity.
    Finally, in \Cref{sec:prop_maps} we investigate properties of these mappings.
    \item \textbf{Operational Causal Influence:} In~\Cref{sec: new ci}, we use insights from the mappings to analyse the definition of causal influence in a tensor network proposed in~\cite{Cotler_2019}, and propose a slightly modified version of it in~\Cref{sec:newCIdef}, called \textit{operational causal influence}, with a clear operational interpretation in terms of conditional probabilities.
    \item \textbf{Connections between signalling and operational causal influence:} In~\Cref{sec:signalling and CI}, we show that the operational causal influence in a given tensor network is preserved as signalling in the image causal model given by our mappings, and vice versa. Thus, we build a concrete connection between these two properties that allows the study of one using tools known for the other. Moreover, this shows that the causal influence definition of \cite{Cotler_2019}, up to the modification we propose, does in fact capture causal structure rather than mere correlations, in an operationally meaningful way.
     
    \item \textbf{Discrete rotations of causal models:} An interesting consequence of the mappings is explored in~\Cref{sec:rotCM}. The mappings show that from a given tensor network on a undirected graph one can define a set of genuinely distinct causal models by choosing different directions of edges in the undirected graph. While the connectivity of the directed graphs underlying such causal models differs from one to the other, they all share the same no-signalling relations. Thus, one can think of such class of causal models as space-time rotations of each other, in similar flavour to the concept of so-called \textit{dual unitaries}~\cite{Piroli_2020}.
    \item \textbf{Application to the holographic tensor network:} In~\Cref{sec:holography}, we apply the mappings to a specific tensor network modelling the holographic principle~\cite{Pastawski_2015}. We show how known results involving properties of directed graphs, such as the $d$-separation theorem~\cite{pearl_2009,Henson_2014}, can be applied to an image causal model of the holographic tensor network to infer causal influence. 
\end{itemize}

Finally,~\Cref{sec:conclusion} discusses future avenues of this work.

\subsection{Notation}
\label{sec:notation}
In the following, we denote with $\hilmap$ a finite-dimensional Hilbert space and with $\linops(\hilmap,\hilmap')$ the complex vector space of linear operators from $\hilmap$ to $\hilmap'$, and denote $\linops(\hilmap):=\linops(\hilmap,\hilmap)$. When a Hilbert space is associated to a system $S$, we denote it as $\hilmaparg{S}$.
The normalised and un-normalised maximally entangled state on $\hilmaparg{A}\cong\hilmaparg{A'}$ are respectively denoted as 
\begin{equation}
    \ket{\bellstate}_{AA'}=\frac{1}{\sqrt{d_A}}\sum_{i}\ket{i}_A\otimes\ket{i}_{A'} \quad \ket{\ubellstate}_{AA'}=\sum_{i}\ket{i}_A\otimes\ket{i}_{A'}.
\end{equation}
 We denote a superoperator $\chanmap \in \linops(\linops(\hilmaparg{A}),\linops(\hilmaparg{B}))$, as $\chanmap_{B|A}$. We denote sequential composition between superoperators $\mathcal{M}_{B|A}$ and $\mathcal{N}_{C|B}$ as $\mathcal{N}_{C|B}\circ\mathcal{M}_{B|A} \equiv \mathcal{T}_{C|A}$ and keep factors of identity implicit, i.e., $\mathcal{E}_{BD|A}$ and $\mathcal{F}_{C|BE}$ can be composed as $\mathcal{F}_{C|BE}\circ\mathcal{E}_{BD|A}\equiv \mathcal{T}'_{CD|AE}$ (where an identity on the $E$ system is implicit).
We use the acronyms: CP for completely positive, CPTP for completely positive and trace preserving, and POVM for positive operator valued measurement.

\paragraph{Graph notation}We denote with $\graphnamedir = \graphexpldir$ a directed graph, denoting directed edges as ordered pairs, i.e., $\edgename=(v,w)\in\edgesetdir$.
The incoming and outgoing edges to a vertex $\vertname\in\vertsetdir$ are denoted as
\begin{equation}
\begin{split}
    \inedges{\vertname} &= \biglset
    { e\in\edgesetdir }{ \exists \vertname'\in\vertsetdir \st e = \edgearg{\vertname'}{\vertname} }, \\
    \outedges{\vertname} &= \biglset
    { e\in\edgesetdir }{ \exists \vertname'\in\vertsetdir \st e = \edgearg{\vertname}{\vertname'} },
\end{split}
\end{equation}
while the parents and children of a vertex $\vertname\in\vertsetdir$ are denoted
\begin{equation}
\begin{split}
    \parnodes{\vertname} &= \biglset{ \vertname' \in \vertsetdir }{ (\vertname',\vertname) \in \edgesetdir}, \\
    \childnodes{\vertname} &= \biglset{ \vertname' \in \vertsetdir }{ (\vertname,\vertname') \in \edgesetdir }.
\end{split}
\end{equation}
We denote with $\graphnameud = \graphexplud$ an undirected graph, and undirected edges as two-element sets, i.e., $\edgename=\{v,w\}\in\edgesetud$, highlighting the absence of any ordering between elements of such pairs. The incident (or adjacent) edges to a vertex $v\in\vertsetud$ are denoted
\begin{equation}
    \incedges{\vertname}=\biglset{ \edgename\in\edgesetud}{\vertname\in \edgename }.
\end{equation}

We consider graphs, both directed and undirected, with no open edges, i.e., edges are always connecting two vertices.
\paragraph{Choi-Jamiołkowski isomorphism} The Choi-Jamiołkowski isomorphism~\cite{Choi1975,Jamiolkowski1972} connects linear superoperators in $\linops(\linops(\hilmaparg{A}),\linops(\hilmaparg{B}))$ to linear operators in $\linops(\hilmaparg{A'}\otimes\hilmaparg{B})$ where $\hilmaparg{A}\cong\hilmaparg{A'}$.

We denote the direct Choi-Jamiołkowski mapping as $\CJ_{B|A}: \linops\left(\linops(\hilmaparg{A}),\linops(\hilmaparg{B})\right)\mapsto \linops(\hilmaparg{A'}\otimes\hilmaparg{B})$, and define it as
\begin{equation}
\label{eq:choidirect}
    \CJ_{B|A}(\chanmap_{B|A})=(\chanmap_{B|A}\otimes \mathcal{I}_{A'})\left(\ketbra{\bellstate}_{AA'}\right).
\end{equation}
We denote the inverse as $\CJin_{B|A}:\linops(\hilmaparg{A'}\otimes\hilmaparg{B})\mapsto \linops(\linops(\hilmaparg{A}),\linops(\hilmaparg{B}))$. One can prove that the inverse of~\cref{eq:choidirect} acts on $S_A\in\linops(\hilmaparg{A})$ as
\begin{equation}
\CJin_{B|A}(M_{A'B})(S_{A})= d_{A} \Tr_{A'}\left(\mathcal{T}_{A'|A}(S_A)M_{A'B}\right)
\end{equation}
where \begin{equation}
    \mathcal{T}_{A'|A}(S_A) = \sum_{i,j} \ketbraa{i}{j}_{A'} \bra{j}S_A\ket{i}_A.
\end{equation}

We recall the following result characterising CPTP maps.
\begin{lemma}
    A superoperator $M_{B|A}\in\linops\left(\linops\left(\hilmaparg{A}\right),\linops\left(\hilmaparg{B}\right)\right)$ is \textup{completely positive} if and only if $\CJ_{B|A}(M_{B|A})$ is positive and it is \textup{trace preserving} if and only if $\Tr_{B}\left[\CJ_{B|A}(M_{A|B})\right] = \frac{\id_A}{d_{A}}$.
\end{lemma}
\begin{proof}
    See Section 5.3 of~\cite{Renes_2022}.
\end{proof}

Further, we introduce the pure inverse Choi-Jamiołkowski isomorphism, denoted as $\CJp_{B|A}:\hilmaparg{A'}\otimes\hilmaparg{B} \mapsto \linops(\hilmaparg{A},\hilmaparg{B})$, defined as
\begin{equation}
      \CJp_{B|A} (\ket{\Psi}_{A'B})(\ket{\phi}_A)=\bra{\Phi^+}_{AA'}(\ket{\psi}_A\otimes \ket{\Psi}_{A'B})\in \hilmap_B.
\end{equation}
Notice that $ \CJp_{B|A} (\ket{\Psi}_{A'B})=\bra{\Phi^+}_{AA'} \ket{\Psi}_{A'B}$ is then a linear transformation from $\hilmap_A$ to $\hilmap_B$. The following lemma connects the inverse to the pure inverse Choi-Jamiołkowski.

\begin{restatable}[]{lemma}{invpurechoi}
\label{lemma:inv_choi}
    Suppose that $M_{A'B}=\sum_k p_k \ket{\Psi_k}\bra{\Psi_k}_{A'B}$  where each $\ket{\Psi_k}$ is normalised. Then,
   \begin{equation}
          \CJin_{B|A}(M_{A'B})=d_A^2\sum_k p_k   \CJp_{B|A} (\ket{\Psi_k}_{A'B})\otimes \Big(  \CJp_{B|A} (\ket{\Psi_k}_{A'B})\Big)^\dagger.
    \end{equation}
\end{restatable}
\begin{proof}
    See~\Cref{app:proofs}.
\end{proof}

\section{Tensor networks and causal influence}
\label{sec: TNrev}
Tensor networks are a versatile representation of tensors and sums of thereof which have found multiple applications in physics, e.g., in condensed matter, quantum information theory and quantum error correction, but also in data science and machine learning. Tensors, i.e., collections of complex numbers labelled by indices $T=\{T_{i_1,i_2,\dots,i_n}\}_{i_1,\dots,i_n}$, are represented through shapes having as many legs as the tensor indices
\begin{equation}
    T_{i_1,i_2,\dots,i_n} = \centertikz{
        \foreach \n in {1,...,8}{ 
            \begin{scope}[rotate =40*\n]
                 \node[] at (0,0.75) {\indexstyle{i_{\n}}};
            \end{scope}
            \draw[color=black,mybaseline,rotate=40*\n] (0,0) --(0,0.5);
        }
        \node at (0,0.5) {\indexstyle{\dots}};
        \node[tnode, fill=blue, minimum size=13pt] at (0,0) {};
    },
\end{equation}
summation between indices of tensors can be represented through contracting legs, for instance contracting
\begin{equation}
     c_i = \centertikz{
        \node[tnode, fill=red, minimum size=13pt] (a){};
        \node (i) [right=\chanvspace*0.5 of a] {\indexstyle{i}};
        \draw[mybaseline] (a.east) -- (i);
    } \textup{ and}  \quad M_{ij}=\centertikz{
        \node[tnode, fill= green, minimum size=13pt] (a){};
        \node (i) [left=\chanvspace*0.5 of a] {\indexstyle{i}};
        \draw[mybaseline] (a.west) -- (i);
        \node (o) [right=\chanvspace*0.5 of a] {\indexstyle{j}};
        \draw[mybaseline] (a.east) -- (o);
    } \textup{ leads to }  \centertikz{
        \node[tnode, fill = red, minimum size=13pt] (a) {};
        \node[tnode, fill = green, minimum size=13pt] (b) [right =\chanvspace*0.5 of a] {};
        \node (i) [right =\chanvspace*0.5 of b] {\indexstyle{j}};
        \draw[mybaseline] (b.east) -- (i);
        \draw[mybaseline] (a.east) -- (b.west);
    } = \sum_i^d c_i M_{ij} \equiv d_j.
\end{equation}

In principle, any system or dynamics which can be described in terms of summation of tensors can be represented through a tensor network. For instance, quantum mechanical states and channels can be easily embedded in a tensor network up to choosing a preferred basis. As the notation portrays, all indices of tensors in a tensor network are of the same type, meaning that there is no such thing as ``inputs'' or ``output'' of a tensor. Therefore, tensor networks are agnostic of any direction, for instance of time, that could underlie the dynamics that they describe. This feature is crucial for studying the emergence of space-time from quantum correlations. Indeed, using tensor networks we can describe a quantum mechanical system and dynamics without assuming any predefined notion of time nor of cause and effect, and infer properties from the connectivity of the network only.

\subsection{Definition of tensor network}
We use a compact way of defining a tensor network used, for instance, in~\cite{Cotler_2019}.

\begin{definition}[Tensor network]
\label{def:TN}
A tensor network is specified by
\begin{myitem}
    \item an undirected graph $\graphnameud=(\vertsetud,\edgesetud)$;
    \item a finite dimensional Hilbert space $\hilmaparg{\edgename}$ associated to each edge $\edgename\in\edgesetud$;\footnote{We use a notation where
            $\hilmaparg{\edgeset'} = \bigotimes_{\edgename \in \edgeset'} \hilmaparg{\edgename}$ for $\edgeset'\subseteq\edgesetud$.}
    \item a density operator $\projarg{\vertname}\in\linops\left(\hilmaparg{\incedges{\vertname}}\right)$ associated to each vertex $\vertname\in\vertsetud$.
\end{myitem}
We denote the triple defining a tensor network as $\tn=\left(\graphnameud,\{\hilmaparg{\edgename}\}_{\edgename\in\edgesetud}, \{\projarg{\vertname}\}_{\vertname\in\vertsetud}\right)$.
\end{definition}
For our results, it is sufficient to allow the operators specified in $3.$ to be un-normalised positive semi-definite operators; in~\cref{rem:densityop}, we clarify how any tensor network defined in terms of coefficients as above can be written in this form. Notice that tensors with open legs are not allowed by~\cref{def:TN}.

We define the \textit{link state} and \textit{total state} of a given tensor network, which allow us to compute correlation functions of the network.
\begin{definition}[Link state and total state]
\label{def:linktotalstate}
    Let $\tn=(\graphnameud,\{\hilmaparg{\edgename}\}_{\edgename\in\edgesetud},\{\projarg{\vertname}\}_{\vertname\in\vertsetud})$ be a tensor network. For each $\edgename=\{v,w\}\in\edgesetud$, let us introduce two Hilbert spaces denoted as $\hilmaparg{(\edgename,v)}$ and $\hilmaparg{(\edgename,w)}$, where $\hilmaparg{(\edgename,v)}\cong\hilmaparg{(\edgename,w)}\cong\hilmaparg{\edgename}$ and denote $\Htot=\bigotimes_{\vertname\in\vertsetud}\bigotimes_{\edgename\in\incedges{\vertname}}\hilmap_{(\edgename,\vertname)}\cong \hilmaparg{\edgesetud}^{\otimes 2}$. We define
    \begin{itemize}
        \item the link state $\linkstate\in \Htot$ as
    \begin{equation}
        \linkstate = \bigotimes_{\edgename\in\edgesetud}\ket{\ubellstate}_{\edgename}
    \end{equation}
    where for $\edgename=\{v,w\}$, the state $\ket{\ubellstate}_{\edgename}\in\hilmaparg{(\edgename,v)}\otimes \hilmaparg{(\edgename,w)}$ is the un-normalised maximally entangled state;
    \item the total state $\rhop\in\linops\left(\Htot\right)$
    \begin{equation}
        \rhop = \bigotimes_{\vertname\in\vertsetud} \projarg{\vertname}
    \end{equation}
    and $\projarg{\vertname}\in \linops\left(\otimes_{\edgename\in\incedges{\vertname}}\hilmaparg{(\edgename,v)}\right)$.
    \end{itemize}
\end{definition}

Notice that a tensor network is originally defined through associating a Hilbert space to each edge, i.e., $\hilmaparg{e}$ to $e\in\edgesetud$. However, in the former definition each edge is then associated with two Hilbert spaces isomorphic to $\hilmaparg{e}$ so that one can define a maximally entangled state over them. This doubling ensures that no direction of the edge is implicitly induced in defining the tensor contraction below.

We now define the tensor contraction which allows to compute correlation functions of a tensor network.

\begin{definition}[Tensor contraction]
\label{def: tn contraction}
   Let $\tn=(\graphnameud,\{\hilmaparg{\edgename}\}_{\edgename\in\edgesetud},\{\projarg{\vertname}\}_{\vertname\in\vertsetud})$ be a tensor network. The tensor contraction is given by
   \begin{equation}
       \bra{L}\rhop\ket{L}\in\mathbb{R}.
   \end{equation}
   For any two linear operators $Q_1,Q_2\in\linops\left(\Htot\right)$, we define the corresponding correlation function as
   \begin{equation}
       \bra{L}Q_1\rhop Q_2\ket{L}\in \mathbb{C}.
   \end{equation}
\end{definition}
An explicit computation of the tensor contraction is provided in the example below.
\begin{remark}
\label{rem:densityop}
     While tensors are generally defined as a collection of numbers, in~\cref{def:TN} a tensor is instead specified by a density operator. We now clarify how one can embed a tensor into a possibly un-normalised density operator and argue that the normalization can be appropriately neglected when studying observable properties of the network. Consider a tensor $T_{i_1,\dots, i_n}$ for $i_k=1,\dots, d_k$, and define the vector
\begin{equation}
\label{eq:pure_enc}
    \ket{\psi_T}= \sum_{i_1,\dots,i_n}^{d_1,\dots,d_n} T_{i_1,\dots, i_n} \ket{i_1}\otimes \dots \otimes \ket{i_n}
\end{equation}
in a Hilbert space $\bigotimes_{j=1}^{n} \hilmaparg{j}$ where $\dim(\hilmaparg{j}) = d_j$. Then, the operator associated to the tensor $T$ is $\rho_T=\ketbra{\psi_T}$, which, up to a real and positive constant representing the normalisation of the state, is a density operator. 
Hence, to an arbitrary collection of numbers, we can always associate a density operator up to a real and positive constant. Such a constant only acts as a global scaling of the tensor, which affects all correlation functions equally. Thus, for the study of correlation functions, one can without loss of generality set such a constant to be one. In our case, the mappings formulated in~\Cref{sec: mappings} also apply to un-normalised density matrices without fixing this constant to one, and the definitions of causal influence remain well defined, since they involve a renormalisation (see~\cref{def: causal influence old,def: causal influence}).

\paragraph{Example} Consider the following example, where
\begin{equation}
\label{eq:example_TN}
\{c_i\}_{i=1}^{d_1}=\centertikz{
\node[tnode, fill=red, minimum size=13pt] (a) at (0,0){};
\draw[mybaseline] (a) --++(0,.5);
},
\quad
\{M_{ij}\}_{i,j=1}^{d_1,d_2}=\centertikz{
\node[tnode, fill=green, minimum size=13pt] (a) at (0,0){};
\node[] at (0,0.75){\indexstyle{}};
\node[] at (0.75,0){\indexstyle{}};
\draw[mybaseline] (a) --++(0,.5);
\draw[mybaseline] (a) --++(0.5,0);
},
\quad
    \centertikz{
        \node[tnode, fill=red, minimum size=13pt] (a) at (0,0){};
        \node[tnode, fill=green, minimum size=13pt] (b) at (-1,0){};
        \draw[mybaseline] (a) -- (b);
        \draw[mybaseline] (b) --++ (0,0.5);
    } = \left\{\sum_{i,j} c_i M_{ij}\right\}_{j} \equiv \tilde{c}.
\end{equation}
with associated states
\begin{equation}
    \ket{\psi_c}_1=\sum_{i}c_i\ket{i}_1, \quad \ket{\psi_M}_{2,3}=\sum_{i,j}M_{ij}\ket{i}_2\ket{j}_3,
\end{equation}
where $1,2,3$ are system labels.
Through the inner product with a maximally entangled state, we obtain
\begin{equation}
    \bra{\ubellstate}_{12}\ket{\psi_c}_{1}\otimes\ket{\psi_M}_{2,3} = \sum_{i,j} c_i M_{ij}\ket{j}_3 \equiv \ket{\psi_{\tilde{c}}}_3.
\end{equation}
Further, if we consider another tensor
\begin{equation}
\label{eq:contracted_TN}
\{b_j\}_{j=1}^{d_2}=\centertikz{
\node[tnode, fill=blue, minimum size=13pt] (a) at (0,0){};
\node[] at (0,0.75){\indexstyle{}};
\draw[mybaseline] (a) --++(0,.5);
}\mapsto \ket{\psi_b}_4 = \sum_kb_k\ket{k}_4
\textup{ and }
\quad
    \centertikz{
        \node[tnode, fill=red, minimum size=13pt] (a) at (0,0){};
        \node[tnode, fill=green, minimum size=13pt] (b) at (-0.75,1.){};
        \node[tnode, fill=blue, minimum size=13pt] (c) at (0.75,1.){};
        \draw[mybaseline] (a) -- (b);
        \draw[mybaseline] (c) -- (b);
    } = \sum_{i,j} c_i M_{ij} b_j.
\end{equation}
We obtain
\begin{equation}
    \label{eq:example_finalsum}\bra{\ubellstate}_{3,4} \ket{\psi_{\tilde{c}}}_3\otimes\ket{\psi_b}_4 = \sum_{i,j} c_iM_{ij}b_j
\end{equation}
corresponding to the tensor summation.

As a consequence of this way of embedding tensors into pure states, we remark that the tensor network contraction of~\cref{def: tn contraction} equals the square of the tensor summation, a quantity which is used to compute correlation functions. For more details we refer to~\cite{Cotler_2019}.
\end{remark}

Using the compact notation of~\cref{def:TN}, the example above is equivalently defined by:
\begin{myitem}
    \item the undirected graph $\centertikz{
        \node[circle,draw,line width = 0.75pt] (a) at (-1.,0){\scriptsize$v_c$};
        \node[circle,draw,line width = 0.75pt] (b) at (0,0){\scriptsize$v_M$};
        \node[circle,draw,line width = 0.75pt] (c) at (1.,0){\scriptsize$v_b$};
        \draw[mybaseline] (a) -- (b);
        \draw[mybaseline] (c) -- (b);
    }$, with vertex set $\vertsetud=\{v_c,v_M,v_b\}$ and edge set $\edgesetud=\{\{v_M,v_c\}\equiv e_1, \{v_M,v_b\}\equiv e_2 \}$;
    \item Hilbert spaces associated to $\edgename_1$ and $\edgename_2$ respectively with dimensions $d_1$ and $d_2$, i.e., $\hilmaparg{1}=\textup{span}\{\ket{i}\}_{i=1}^{d_1}$ and $\hilmaparg{2}=\textup{span}\{\ket{j}\}_{j=1}^{d_2}$;
    \item the following linear operators
    \begin{equation}
        \projarg{c} =\ketbra{\psi_c}_{1}, \quad \projarg{M} =\ketbra{\psi_M}_{1,2}, \quad \projarg{b} =\ketbra{\psi_b}_{2}.
    \end{equation}
\end{myitem}
Notice that the undirected graph used to define a tensor network is the same underlying the network itself in the rightmost diagram in~\cref{eq:contracted_TN}.

To define the link state and total state let us define isomorphic Hilbert spaces $\hilmaparg{1_c}\cong \hilmaparg{1_M}\cong\hilmaparg{1}$ and $\hilmaparg{2_b}\cong \hilmaparg{2_M}\cong\hilmaparg{2}$. Then we define
\begin{equation}
    \rhop = \underbrace{\ketbra{\psi_c}_{1_c}}_{\projarg{c}}\otimes \underbrace{\ketbra{\psi_b}_{2_b}}_{\projarg{b}} \otimes \underbrace{\ketbra{\psi_M}_{1_M,2_M}}_{\projarg{M}}
\end{equation}
and 
\begin{equation}
    \linkstate = \ket{\ubellstate}_{1_c,1_M}\otimes \ket{\ubellstate}_{2_b,2_M} = \sum_{i=1}^{d_1} \ket{i}_{1_c}\otimes\ket{i}_{1_M} \otimes \sum_{j=1}^{b_2} \ket{j}_{2_b}\otimes\ket{j}_{2_M}.
\end{equation}
Then we evaluate the tensor contraction as:
\begin{equation}
\begin{split}
    &\bra{L}\rhop\ket{L}\\
    &=\bra{\ubellstate}_{1_c,1_M}\otimes \bra{\ubellstate}_{2_b,2_M}\left(\ketbra{\psi_c}_{1_c}\otimes\ketbra{\psi_b}_{2_b}\otimes \ketbra{\psi_M}_{1_M,2_M}\right)\ket{\ubellstate}_{1_c,1_M}\otimes \ket{\ubellstate}_{2_b,2_M}\\
    &=\left|\sum_{ij=1}^{d_1,d_2}c_iM_{ij}b_j\right|^2,
\end{split}
\end{equation}
thus showing that we obtain the square of the tensor summation in~\cref{eq:example_finalsum}.

\subsection{Causal influence as defined in~\cite{Cotler_2019}}
\label{sec:tn causal influence intro}
In~\cite{Cotler_2019}, a definition of causal influence in tensor networks is proposed. Given two disjoint subsets of edges of the network, $\region A$ and $\region B$, one says that there is causal influence from $\region A$ to $\region B$ if correlations between a Hermitian observable inserted in $\region B$ and a unitary operation applied in $\region A$ depend on which unitary is chosen in $\region A$. In other words, causal relations are inferred by probing how local changes in one region of the network affect correlation functions in another.

\begin{restatable}[Quantum causal influence --- \cite{Cotler_2019}]{definition}{oldci}
\label{def: causal influence old}
    Let $\tn = \left(\graphnameud,\{\hilmaparg{\edgename}\}_{\edgename\in\edgesetud},\{\projarg{\vertname}\}_{\vertname\in\vertsetud}\right)$ be a tensor network and consider an ordered pair of disjoint subsets of edges, i.e., $(\region{A},\region{B})$ such that $\region{A},\region{B}\subset\edgesetud$ and $\region{A}\cap\region{B}=\emptyset$, which we call \textup{regions}. Define the correlation function
    \begin{equation}
        M_{\textup{CHQY}}(U_\region{A}:O_\region{B}) = \frac{\bra{L}(U_\region{A}\otimes O_\region{B})\rhop (U_\region{A}^\dagger\otimes O_\region{B}^\dagger)\ket{L}}{\bra{L}(U_\region{A}\otimes \id_\region{B})\rhop (U_\region{A}^\dagger\otimes \id_\region{B})\ket{L}},
    \end{equation}
    where $U_\region{A}\in\linops(\hilmaparg{\region{A}})$ is a unitary and $O_\region{B}\in\linops(\hilmaparg{\region{B}})$ an Hermitian operator such that $O_\region{B}^\dagger O_\region{B}\leq \id_{\region{B}}$.
    The quantum causal influence from $\region{A}$ to $\region{B}$ is zero if and only if 
    \begin{equation}
        M_{\textup{CHQY}}(U_\region{A}:O_\region{B}) = M_{\textup{CHQY}}(U'_\region{A}:O_\region{B})
    \end{equation}
    for all $U_\region{A}$ and $U_\region{A}'$ unitaries and for all $O_\region{B}$ Hermitian operators.
\end{restatable}

As discussed in the remark of the previous section, notice that eventual normalization factors of the operators $\projarg{\vertname}$ do not play any role in~\cref{def: causal influence old} because of the denominator.
\paragraph{Example}
Let us study causal influence in the simple tensor network that we used as an example in the previous section (\cref{eq:contracted_TN}):
\begin{equation}
    \centertikz{
        \node[tnode, fill=red, minimum size=13pt] (a) at (0,0){};
        \node[tnode, fill=green, minimum size=13pt] (b) at (-0.75,1.){};
        \node[tnode, fill=blue, minimum size=13pt] (c) at (0.75,1.){};
        \node[above of = a,node distance=13pt] {$c$};
        \node[above of = c,node distance=13pt] {$b$};
        \node[above of = b,node distance=13pt] {$M$};
        \draw[mybaseline] (a) -- (b);
        \draw[mybaseline] (c) -- (b);
    } = \sum_{i,j} c_i M_{ij} b_j,
\end{equation}
and let $\region{A}=\{e_1=\{v_M,v_c\}\}$ and $\region{B}=\{e_2=\{v_M,v_b\}\}$. Then, the causal influence involves the study of correlation functions
\begin{equation}
\begin{split}
     &\bra{L}(U_\region{A}\otimes O_\region{B})\rhop (U_\region{A}^\dagger\otimes O_\region{B}^\dagger)\ket{L}\\
     &=\left(\sum_{i,j=1}^{d_1,d_2}\bra{i}_{1_c}\otimes\bra{i}_{1_M}\otimes\bra{j}_{2_b}\otimes\bra{j}_{2_M}\right) \left(U_{1_c}\otimes O_{2_b}\otimes \id_{1_M,2_M}\right)\\
    &\quad\left(\ketbra{\psi_c}_{1_c}\otimes\ketbra{\psi_b}_{2_b}\otimes \ketbra{\psi_M}_{1_M,2_M}\right)\\
    &\quad\quad\left(U^\dagger_{1_c}\otimes O^\dagger_{2_b}\otimes \id_{1_M,2_M}\right)\left(\sum_{k,l=1}^{d_1,d_2}\ket{k}_{1_c}\otimes\ket{k}_{1_M}\otimes\ket{l}_{2_b}\otimes\ket{l}_{2_M}\right).
\end{split}
\end{equation}
The Hilbert spaces associated to each edge are doubled in defining the total state and link state. This creates an ambiguity in~\cref{def: causal influence old} for the spaces on which the unitary and Hermitian operators act on the doubled Hilbert space $\Htot$. In the next section, we argue that for determining whether there is causal influence or not, one can choose arbitrarily on which space to apply the operation because of the structure of the link state.

\subsection{Remarks and open questions regarding this notion causal influence}
\label{sec:open questions ci}
The notion of causal influence of~\cref{def: causal influence old} is extensively motivated in~\cite{Cotler_2019}. 
Here we summarize the main arguments clarifying the definition and we collect open questions regarding it. 
In~\Cref{sec: new ci}, we will continue the discussion on causal influence and motivate the introduction of a modified version of it.

\paragraph{On which Hilbert spaces are the operators $U_{\region{A}}$ and $O_{\region{B}}$ acting?} The doubling of Hilbert spaces, necessary to define the total state and link state, leads to an ambiguity on where the operators $U_{\region{A}}$ and $O_{\region{B}}$ act. Indeed, the total and link states are defined on the Hilbert space $\Htot=\bigotimes_{\vertname\in\vertsetud}\bigotimes_{\edgename\in\incedges{\vertname}}\hilmap_{(\edgename,\vertname)}\cong \hilmaparg{\edgesetud}^{\otimes 2}$ thus there are two subspaces isomorphic to $\hilmaparg{\region{A}}$ and two subspaces isomorphic to $\hilmaparg{\region{B}}$. \Cref{def: causal influence old} does not specify on which of these spaces the unitary and Hermitian operator act on. However, for determining whether there is causal influence or not one can arbitrarily choose on which space to apply the unitary as long as the choice is consistent in all evaluations of the quantity $M$ for different unitaries and Hermitian operators. More details, including a proof of this statement, can be found in~\Cref{app:clarification on causal influence}.

\paragraph{Why are unitaries and Hermitian operators introduced?}
Tensor networks do not have a predefined direction, in particular no predefined notion of input and output. This property should be preserved when modifying the network to test information flow. In particular, the regions in which we introduce modifications correspond to undirected edges. Therefore, any operation inserted in these regions must be implemented in a way that does not impose an artificial direction on those edges. For this reason, the authors propose modifications of the network that are unitary operators, which remain unitary if we reverse input and output, and Hermitian operators which are invariant up to reversing inputs and outputs.

\paragraph{What is the role of the denominator?}

We aim to quantify the influence from $\region{A}$ to $\region{B}$ in the tensor network. This amounts to modifying locally the network in $\region{A}$ and $\region{B}$ and evaluate correlation functions to quantify how modifications in $\region{A}$ affect the network in $\region{B}$. However, it might be that a modification in $\region{A}$, generally affects the network contraction but not necessarily in $\region{B}$. Explicitly, it might be that for two unitaries $U_\region{A}$ and $V_\region{A}$
    \begin{equation}
       \bra{L}(U_\region{A}\otimes \id_\region{B})\rhop (U_\region{A}^\dagger\otimes \id_\region{B})\ket{L}\neq\bra{L}(V_\region{A}\otimes \id_\region{B})\rhop (V_\region{A}^\dagger\otimes \id_\region{B})\ket{L}.
    \end{equation}
    The role of the denominator is to normalize the network contraction so that the measure $M_{\textup{CHQY}}$ solely quantifies the influence in $\region{B}$ and not in the global network contraction. In addition, the denominator ensures that correlations functions are not affected by possible normalisation factors that one introduces to describe a given tensor as a density operator.

\paragraph{Open questions}
\Cref{def: causal influence old} provides a notion of causality in a tensor network based on correlation functions. This naturally raises questions concerning both its interpretation and its relation to existing approaches to causality.

A first question is whether and how this notion can be connected to causal modelling frameworks such as~\cite{Quantum_paper,Barrett_2021,Henson_2014}. These frameworks typically assume a predefined distinction of inputs and outputs of their building blocks, and within them several techniques for analysing causal structure are known. Establishing such a connection could therefore allow techniques developed in one formalism to be transferred to the other.

A second, closely related question concerns the interpretation of defining causal influence purely in terms of correlation functions: it is not immediately clear to what extent \cref{def: causal influence old} captures genuine causation rather than mere correlations. This issue is fundamental, as distinguishing causation from correlation is one of the main motivations behind causal modelling frameworks, where this distinction is addressed in a formal and rigorous way, typically through the notion of interventions.

Moreover, although the role of the denominator in \cref{def: causal influence old} is motivated as described above, it is not evident that its particular form fully accomplishes this purpose. In the absence of a clear operational interpretation, it remains ambiguous whether the normalization indeed isolates causal influence rather than simply rescaling correlations. In~\Cref{sec:newCIdef}, we propose a modified definition which replaces the denominator with a term that admits a direct operational meaning in terms of conditional probabilities of outcomes under freely chosen interventions. This ensures that the resulting quantity captures genuine causal effects rather than mere statistical correlations.

Finally, these considerations point to the need for a clearer operational interpretation of causal influence in tensor networks. In particular, one would like to understand why the proposed definition should be regarded as describing causality and how it relates to operational procedures.
As we will see in~\Cref{sec: new ci,sec:newCIdef}, these questions lead to a modified version of~\cref{def: causal influence old}, which can be connected to known notions of causality and, through this connection, provides tensor networks with a clearer operational meaning.

\section{Cyclic causal models and signalling}
\label{sec: CMrev}
In this section, we briefly review the causal modelling framework presented in~\cite{Quantum_paper}. This framework allows modelling of finite dimensional causal models on cyclic graphs, through formulating a probability rule which is well-defined in all but a handful of pathological models. 
Different quantum causal modelling frameworks have been proposed~\cite{Barrett_2019,Henson_2014}, however, these do not allow for arbitrary models on generally cyclic graphs. As we will see in the following sections, tensor networks, which do not have a predefined direction of time or cause-and-effect, can generally lead to cyclic causal explanation. Thus, a generally cyclic framework is  necessary to accommodate this possibility.

In what follows, we first introduce the notation and definition of a causal model, and present the probability rule. We refer to~\cite{Quantum_paper} for all the proofs and additional results related to the framework. Then, we introduce the notion of interventions and define signalling for these causal models. Among others, the topics of this section are the content of a follow-up work focussed on intervention in cyclic scenarios~\cite{Interventions_paper}.

\subsection{Causal modelling framework}
Causal models rely on a graphical representation of cause and effect relations through so-called \textit{causal graphs}. 

\begin{definition}[Causal graph]
\label{def: causal graph}
    A causal graph is a directed graph $\graphnamedir = (\vertsetdir,\edgesetdir)$ such that $\vertsetdir$ can be partitioned into $\vertsetdir = \overtset \cup \uvertset$, where $\vertname \in \overtset$ are called the \textup{observed} vertices, denoted as $\centertikz{\node[onode] {$\vertname$};}$, and $\vertname \in \uvertset$ are called the \textup{unobserved} vertices and are denoted as $\centertikz{\node[unode] {$\vertname$};}$.
    
\end{definition}

A causal model assigns to a causal graph specific causal mechanisms for each vertex.
These causal mechanisms have to match the type of the vertex, e.g., if the vertex is observed, the causal mechanism has to specify how the observed outcome is obtained.

\begin{definition}[Causal model on causal graph]
\label{def:causal model}
    A causal model on a causal graph is specified by the following items:
    \begin{myitem}
        \item a causal graph $\graphnamedir=(\vertsetdir,\edgesetdir)$;
        \item a finite-dimensional Hilbert space $\hilmaparg{\edgename}$ is associated to each edge $\edgename\in\edgesetdir$;\footnote{We use a notation where
            $\hilmaparg{\edgeset'} = \bigotimes_{\edgename \in \edgeset'} \hilmaparg{\edgename}$ for $\edgeset'\subseteq\edgesetdir$.}
        \item a quantum operation $\mathcal{T}^\vertname\subseteq\linops(\linops(\hilmaparg{\inedges{\vertname}}),\linops(\hilmaparg{\outedges{\vertname}}))$ associated to each vertex $\vertname\in\vertsetdir$ as follows:
        \begin{itemize}
            \item for unobserved $\vertname\in\uvertset$ a single CPTP map $\mathcal{T}^{\vertname}=\{\chanmaparg{\vertname}\}$,
          
            \item for observed $\vertname\in\overtset$ an instrument $\mathcal{T}^\vertname=\{\measmaparg{\outcome}{\vertname}\}_{\outcome\in\outcomemaparg{\vertname}}$, i.e., a collection of CP maps, each acting on $\rho\in\linops(\hilmaparg{\inedges{\vertname}})$ as
            \begin{equation}
                \measmaparg{\outcome}{\vertname}(\rho) = \Tr[\povmelarg{\outcome}{\vertname} \rho] \bigotimes_{\edgename\in\outedges{\vertname}} \ketbra{\outcome}_{\edgename}
            \end{equation} where $\left\{\povmelarg{\outcome}{\vertname}\in\linops(\inedges{\vertname})\right\}_{\outcome\in\outcomemaparg{\vertname}}$ is a POVM associated with a random variable $X_\vertname$ taking values from a finite set $x_\vertname\in\outcomemaparg{\vertname}$. We denote the CPTP map obtained through marginalisation as $\chanmaparg{\vertname}=\sum_x\measmaparg{\outcome}{\vertname}$
        \end{itemize}
    \end{myitem}
    We denote the triple defining a causal model as $\cm=(\graphnamedir,\{\hilmaparg{\edgename}\}_{\edgename\in\edgesetdir},\{\mathcal{T}^\vertname\}_{\vertname\in\vertsetdir})$.
\end{definition}
The definition above is a simplified version of definition 2 in~\cite{Quantum_paper} phrased to highlight the analogy with~\cref{def:TN} of a tensor network. \Cref{def:causal model} retains all the defining features of a causal model which are needed here, while neglecting features which do not play any role in what follows. For instance, the dimensions of Hilbert spaces associated to outgoing edges of observed vertices have to match the cardinality of the finite set associated to the vertex. Notice that for an exogenous vertex $\vertname$, the associated maps are elements of $\linops(\linops(\mathbb{C}),\linops(\hilmaparg{\outedges{\vertname}}))\cong\linops(\mathbb{C},\linops(\hilmaparg{\outedges{\vertname}}))\cong \linops(\hilmaparg{\outedges{\vertname}})$. We refer to definition 2 from~\cite{Quantum_paper} for more details. 

We also note that there exist several distinct frameworks of quantum causal models in the literature, differing in their underlying assumptions, mathematical frameworks, and regime of validity. The relation between the present framework and alternative approaches is discussed in detail in~\cite{Quantum_paper}. Conceptually, the definition adopted here can be viewed as a cyclic extension of the causal models introduced in~\cite{Henson_2014}, restricted to quantum theory.\footnote{\cite{Henson_2014} defined a causal modelling approach for acyclic graphs, but applicable to a general class of operational probabilistic theories, including quantum theory.}

\subsection{Probability rule}
Causal modelling formalisms include a prescription of how to evaluate probabilities over the outcomes associated to observed vertices of a causal graph, using the specified causal mechanisms (here, CP/CPTP maps, POVMs etc).
Within the framework of~\cite{Quantum_paper}, a probability rule valid for arbitrary finite-dimensional causal models, i.e., those given by~\cref{def:causal model}, is provided. For a complete analysis of the probability rule we refer to section 4 of~\cite{Quantum_paper}.

Here, we present a formulation of the probability rule in~\cite{Quantum_paper} which relies on a special type of channel composition called \textit{self-cycle composition}. This formulation is equivalent to the probability rule of~\cite{Quantum_paper} (see~\Cref{app:cyctoacyc}).

Firstly, we introduce the \textit{self-cycle composition} for quantum channels~\cite{VilasiniRennerPRA}, which, in the finite-dimensional setting considered here, is equivalent to loop composition presented in~\cite{Portmann_2017}.

\begin{definition}[Self-cycle composition]
\label{def:selfcycle_v3}
    Let $\hilmaparg A \cong \hilmaparg C$ with $d = \dim(\hilmap_A)=\dim(\hilmap_C)$. Further, let $\{\ket{k}_A\}_{k=1}^{d}$ and $\{\ket{l}_A\}_{l=1}^{d}$ be any orthonormal bases of $\hilmaparg A$ and $\{\ket{k}_C\}_{k=1}^{d}$ and $\{\ket{l}_C\}_{l=1}^{d}$ be the corresponding bases of $\hilmaparg{C}$, i.e., $\ket{k}_A\cong \ket{k}_C$ and $\ket{l}_A\cong \ket{l}_C$ for all $k,l=1,\dots,d$.
    Then for any linear map $\mathcal{M}_{AB|CD} : \linops(\hilmaparg C\otimes \hilmaparg D) \mapsto \linops(\hilmaparg A\otimes \hilmaparg B)$, we define the self-cycle composition $\selfcycle_{A,C}(\mathcal{M}_{AB|CD}):\linops( \hilmaparg D) \mapsto \linops(\hilmaparg B)$ as follows
    \begin{align}
        \selfcycle_{A,C}(\mathcal{M}_{AB|CD}) = \sum_{k,l=1}^{d} \bra{k}_A \mathcal M_{AB|CD}(\ketbraa{k}{l}_C) \ket{l}_A.
    \end{align}
\end{definition}

It is worth noting that completely positive maps are closed under loop composition, whereas completely positive trace-preserving maps are not: composing a CPTP map with a loop may yield a map that remains CP but is no longer trace-preserving.

Notice that if the spaces $\hilmaparg{B}$ and $\hilmaparg{D}$ are trivial, $\selfcycle_{A,C}(\mathcal{M}_{A|C})\in\mathbb{C}$. In this case, we will omit the subscript and denote the composition as $\selfcycle(\mathcal{M}_{A|C})$. In addition, we might sometimes use the same label for the isomorphic systems $A$ and $C$, and denote the composition as $\selfcycle_A$.

Let us define a collection of CP maps which compactly contains the maps associated to vertices of a given causal model.

\begin{definition}[Total maps of a causal model]
    \label{eq:def etot}
    Let $\cm=(\graphnamedir,\{\hilmaparg{\edgename}\}_{\edgename\in\edgesetdir},\{\mathcal{T}^\vertname\}_{\vertname\in\vertsetdir})$ be a causal model, whose causal graph has sets $\overtset$ and $\uvertset$ of observed and unobserved vertices. 
    We denote a joint observed event associated with $\overtset$ as $\outcome := \{\outcome_\vertname \in \outcomemaparg{\vertname}\}_{\vertname\in\overtset}$, and, introducing $\hilmap_{\edgesetdir_{\textup{in}}}\cong\hilmap_{\edgesetdir_\textup{out}}\cong \hilmaparg{\edgesetdir}$, define the collection of CP maps, called total maps,
    \begin{align}
        \bar{\etot}=\Big\{\etot^\outcome : \textstyle\linops\left(\hilmap_{\edgesetdir_{\textup{in}}}\right) \mapsto \linops\left(\hilmap_{\edgesetdir_{\textup{out}}}\right) \Big\}_{\outcome}
    \end{align} 
    through
    \begin{align}
        \etot^\outcome = \bigotimes_{\vertname\in\overtset} \measmaparg{\outcome_\vertname}{\vertname}\bigotimes_{\vertname'\in \uvertset}\chanmaparg{\vertname'},
    \end{align}
    where $\measmaparg{\outcome_\vertname}{\vertname}$ and $\chanmaparg{\vertname'}$ refer to the maps of $\cm$ (see \cref{def:causal model}). Notice that by~\cref{def:causal model}, $ \etot^\outcome$ has input in $\hilmaparg{\edgesetdir_{\textup{in}}}=\cup_{\vertname\in \vertset} \inedges{\vertname}$ and output in $\hilmaparg{\edgesetdir_{\textup{out}}}=\cup_{\vertname\in \vertset} \outedges{\vertname}$, and both spaces are isomorphic to the space $\hilmaparg{\edgesetdir}$. Furthermore, we will refer to $\etot:=\sum_x\etot^x$ as the marginalised total map of $\cm$.
\end{definition}

Now, we define the probability rule through self-cycle composition of the total map of a causal model as follows.

\begin{definition}[Probability distribution of a causal model]
    \label{prop:probs as self cycles}
    Consider a causal model $\cm=(\graphnamedir,\{\hilmaparg{\edgename}\}_{\edgename\in\edgesetdir},\{\mathcal{T}^\vertname\}_{\vertname\in\vertsetdir})$ and let a joint observed event associated with $\overtset$ be denoted as $\outcome := \{\outcome_\vertname \in \outcomemaparg{\vertname}\}_{\vertname\in\overtset}$.
    If $\sum_{x} \selfcycle(\etot^x) \neq 0$, the probability $\prob(x)_{\graphnamedir}$ is defined as
    \begin{align}
    \label{eq:cyclic probs_v2}
        \prob(x)_{\graphnamedir} = \frac{\selfcycle(\etot^x)}{\sum_x \selfcycle(\etot^x)}.
    \end{align}
    If $\sum_{x} \selfcycle(\etot^x) = 0$, we say that the causal model is inconsistent and the probabilities $ \prob\left(\outcome\right )_{\graphnamedir}$ are undefined.
\end{definition}
In~\Cref{app:cyctoacyc}, we recall results from~\cite{Quantum_paper} showing that the probability rule of~\cref{prop:probs as self cycles} can be equivalently obtained using the acyclic probability rule and post-selection. The denominator appearing in~\cref{prop:probs as self cycles} is related to the success probability of the post-sections. 
As shown in~\cite{Quantum_paper}, in case the causal graph is acyclic, the denominator is $1$ and the numerator is the usual sequential composition of channels which leads to the Born rule. For (possibly cyclic) causal models associated to valid process matrices~\cite{Oreshkov_2012, Barrett_2021}, the denominator is also $1$ and the numerator gives the generalised Born rule defined for processes \cite{Quantum_paper}. Thus,~\cref{prop:probs as self cycles} recovers known probability rules in existing frameworks. 
For more details we refer to~sections $4$, $6$ and $7$ of~\cite{Quantum_paper}. 

For connecting tensor networks to causal models we need the framework in~\cite{Quantum_paper}, as the absence of a predefined causal direction generally leads to cycles. Hence, building up the connection within our framework ensures that such possibility is not a priori ruled out.

\paragraph{Example}Consider the following directed graph as an example
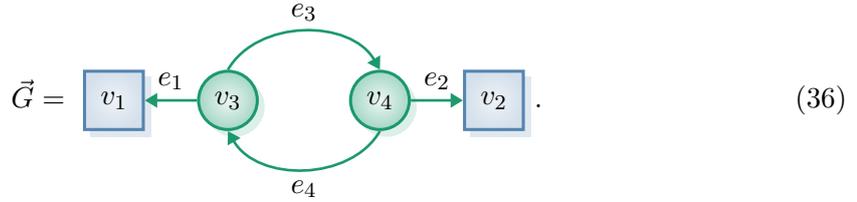
\begin{equation}
\label{eq:example_cycle_overview}
    \graphnamedir=\centertikz{
        \node[unode,minimum size =22pt] (v1) at (-1, 0) {$v_3$};
        \node[unode, minimum size =22pt] (v2) at (1, 0) {$v_4$};
        \node[onode, minimum size =22pt] (v4) at (2.5, 0) {$v_2$};
        \node[onode, minimum size =22pt] (v3) at (-2.5, 0) {$v_1$};
        \draw[qleg] (v1.north) to[in = 120, out = 60] node[above] {$e_3$} (v2.north);
        \draw[qleg] (v2.south) to[in = 300, out = 240] node[below] {$e_4$}(v1.south);
        \draw[qleg] (v2) -- node[above] {$e_2$} (v4);
        \draw[qleg] (v1) -- node[above] {$e_1$} (v3);
    }.
\end{equation}
The edges are labelled as $e_1 = (v_3,v_1)$,  $e_2 = (v_4,v_2)$, $e_3 = (v_3,v_4)$ and $e_4 = (v_4,v_3)$.  We define a causal model by first associating finite dimensional Hilbert spaces to each edge: $\hilmaparg{i}$ to $e_i$. Then, we associate CPTP maps
$\chanmap^{3}:\linops\left(\hilmap_{4}\right)\mapsto \linops\left(\hilmap_{3}\otimes \hilmaparg{1}\right)$ to $\vertname_3$ and $\chanmap^{4}:\linops\left(\hilmap_{3}\right)\mapsto \linops\left(\hilmap_{4}\otimes \hilmaparg{2}\right)$ to $v_4$. To the observed vertices, $\vertname_1$ and $\vertname_2$, we associate collections of CP maps acting on $\rho_{1}\in\linops(\hilmaparg{1})$ and $\rho_{2}\in\linops(\hilmaparg{2})$, as
\begin{equation}
    \measmaparg{x_1}{(1)} (\rho_1) = \Tr(\povmarg{(1)}_{x_1}\rho_1) 
   \textup{ and } \measmaparg{x_2}{(2)} (\rho_2) = \Tr(\povmarg{(2)}_{x_2}\rho_2) 
\end{equation} 
where 
\begin{equation}
    \povmarg{(1)}=\left\{\povmelarg{\outcome_1}{(1)}\in\linops\left(\hilmaparg{1}\right)\right\}_{\outcome_1\in\outcomemaparg{1}}  
   \textup{ and } \povmarg{(2)}=\left\{\povmelarg{\outcome_2}{(2)}\in\linops\left(\hilmaparg{2}\right)\right\}_{\outcome_2\in\outcomemaparg{2}}
\end{equation} 
are POVMs, and $\outcomemaparg{1}$ and $\outcomemaparg{2}$ are finite sets.

The total map associated to the observed event $(x_1,x_2)\in \outcomemaparg{1}\times\outcomemaparg{2}$ is
\begin{equation}
    \etot^{x_1,x_2} = \measmaparg{x_1}{(1)}\otimes\measmaparg{x_2}{(2)} \otimes \chanmap^{3} \otimes \chanmap^{4}
\end{equation}
The probability of a joint observed event associated with outcomes $x_1$ and $x_2$, is
\begin{equation}
    \prob(x_1,x_2)_{\graphnamedir} = \frac{\selfcycle\left(\etot^{x_1,x_2}\right)}{\sum_{y_1,y_2}\selfcycle\left(\etot^{y_1,y_2}\right)}.
\end{equation}

\subsection{Interventions and signalling}
\label{sec:signalling}
Signalling refers to the possibility for one party (or agent), $\labarg{A}$, to transmit information to another, $\labarg{B}$, by means of their interventions on physical systems. In the language of our causal models, these physical systems are represented by the edges of the graph, and signalling between $\labarg{A}$ and $\labarg{B}$ corresponds to information being transmitted along the subset of systems (edges) that link them. In the causal modelling literature, signalling is probed through so-called \textit{interventions}. These are modification of the causal model at hand, which introduce a set of possible operations on the edges associated to $\labarg{A}$ and an observable on the edges associated to $\labarg{B}$. 

\begin{definition}[Intervention]
\label{def:intervention}
    Let $\cm=(\graphnamedir,\{\hilmaparg{\edgename}\}_{\edgename\in\edgesetdir},\{\mathcal{T}^\vertname\}_{\vertname\in\vertsetdir})$ be a causal model and consider an ordered pair of disjoint subsets of edges, i.e., $(\labarg{A},\labarg{B})$ such that $\labarg{A}, \labarg{B}\subseteq \edgesetdir$ and $\labarg{A}\cap\labarg{B}=\emptyset$, which we call \emph{labs}. An intervention is defined by:
    \begin{itemize}
        \item a collection of CPTP maps labelled by $a\in\outcomemaparg{A}$
        \begin{equation}
            \chanmap_{\labarg{A}}=\left\{\chanmapsetting{\labarg{A}}{a}: \linops\left(\hilmaparg{\labarg{A}}\right)\mapsto \linops\left(\hilmaparg{\labarg{A}}\right)\right\}_{a\in \outcomemaparg{A}}
        \end{equation}
        associated to the lab $\labarg{A}$;
        \item a collection of quantum instruments labelled by $b\in\outcomemaparg{B}$ and each having outcome $y\in\outcomemaparg{Y}$  
        \begin{equation}
            \instrmap_{\labarg{B}}=\left\{\left\{\instrmapsetting{\labarg{B}}{y|b}: \linops\left(\hilmaparg{\labarg{B}}\right)\mapsto \linops\left(\hilmaparg{\labarg{B}}\right)\right\}_{y\in\outcomemaparg{Y}}\right\}_{b\in \outcomemaparg{B}}
        \end{equation}
        associated to the lab $\labarg{B}$.
    \end{itemize}
\end{definition}
Signalling is established by observing correlations between $\labarg{A}$’s freely chosen interventions and the observable outcomes at $\labarg{B}$. The freedom in $\labarg{A}$’s choice is essential: only correlations between such independent choices and $\labarg{B}$’s outcomes can witness signalling, thereby distinguishing genuine causal influence from mere pre-existing correlations. Thus, we define the statistics of an intervention through considering the joint probability of an observed event in the causal model and the outcome of the instrument in $\labarg{B}$, and signalling.

\begin{definition}[Statistics of an intervention]
\label{def:statisticsint}
    Let $\cm=(\graphnamedir,\{\hilmaparg{\edgename}\}_{\edgename\in\edgesetdir},\{\mathcal{T}^\vertname\}_{\vertname\in\vertsetdir})$ be a causal model, consider an intervention with the ordered pair $(\labarg{A},\labarg{B})$ as labs and operations $\chanmap_{\labarg{A}}$ and $\instrmap_{\labarg{B}}$. Let us denote a joint observed event associated with the vertices $\overtset$ as $\outcome := \{\outcome_\vertname \in \outcomemaparg{\vertname}\}_{\vertname\in\overtset}$, and let the total map of $\cm$ be
    \begin{align}
        \bar{\etot}=\Big\{\etot^\outcome : \textstyle\linops\left(\hilmap'_{\edgesetdir}\right) \mapsto \linops\left(\hilmap_{\edgesetdir}\right) \Big\}_{\outcome}
    \end{align} 
    where $\hilmap'_{\edgesetdir}\cong\hilmap_{\edgesetdir}$.
    Then, if $\sum_{x,y}\selfcycle\left(\left(\instrmapsetting{\labarg{B}}{y|b}\otimes \chanmapsetting{\labarg{A}}{a}\right) \circ \etot^\outcome \right)\neq 0$, the statistics of the intervention labelled by $(a,b)\in\outcomemaparg{\labarg{A}}\times\outcomemaparg{\labarg{B}}$ is defined as
    \begin{equation}
        \prob(x,y|a,b)_I := \frac{\selfcycle\left(\left(\instrmapsetting{\labarg{B}}{y|b}\otimes \chanmapsetting{\labarg{A}}{a}\right)\circ \etot^\outcome\right)}{\sum_{x,y}\selfcycle\left(\left(\instrmapsetting{\labarg{B}}{y|b}\otimes \chanmapsetting{\labarg{A}}{a}\right)\circ \etot^\outcome\right)}
    \end{equation}
    otherwise we say that the intervention is inconsistent and the statistics is left undefined.
\end{definition}

The above definition can be equivalently obtained through defining a causal model for each choice of intervention on $\labarg{A}$ and $\labarg{B}$ on a graph obtained from $\graphnamedir$ but with the addition of intervention vertices for $\labarg{A}$ and $\labarg{B}$ (see~\cref{eq:example_noint,eq:example_int1,eq:example_int2} of the example below) and use the probability rule in~\cref{prop:probs as self cycles} to evaluate the statistics of intervention. This procedure is described in a follow-up work in preparation~\cite{Interventions_paper}. 

\begin{definition}[Signalling]
\label{def: signalling}
    Let $\cm=(\graphnamedir,\{\hilmaparg{\edgename}\}_{\edgename\in\edgesetdir},\{\mathcal{T}^\vertname\}_{\vertname\in\vertsetdir})$ be a causal model, consider an intervention with the ordered pair $(\labarg{A},\labarg{B})$ as labs and operations 
    \begin{equation}
        \chanmap_{\labarg{A}}=\left\{\chanmapsetting{\labarg{A}}{a}\right\}_{a\in\outcomemaparg{A}} \textup{ and } \instrmap_{\labarg{B}}=\left\{\left\{\instrmapsetting{\labarg{B}}{y|b}\right\}_{y\in\outcomemaparg{Y}}\right\}_{b\in\outcomemaparg{B}}.
    \end{equation} 
    We say that $\labarg{A}$ signals to $\labarg{B}$ in $\cm$ relative to the interventions $\chanmap_{\labarg{A}}$ and $\instrmap_{\labarg{B}}$, if there exists $y\in \outcomemaparg{Y}$, $b\in \outcomemaparg{B}$ and $a,a'\in \outcomemaparg{A}$, such that 
    \begin{align}
        \prob(y|a,b)_{I}\neq \prob(y|a',b)_{I}.
    \end{align}
\end{definition}
\paragraph{Example} Let us show how signalling can be studied in the cyclic example that we introduced before:
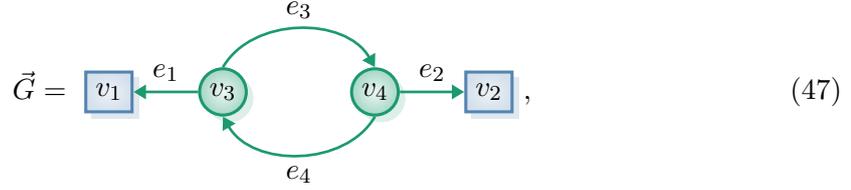
\begin{equation}
\label{eq:example_noint}
    \graphnamedir=\centertikz{
        \node[unode] (v1) at (-1, 0) {$v_3$};
        \node[unode] (v2) at (1, 0) {$v_4$};
        \node[onode] (v4) at (2.5, 0) {$v_2$};
        \node[onode] (v3) at (-2.5, 0) {$v_1$};
        \draw[qleg] (v1.north) to[in = 120, out = 60] node[above] {$e_3$} (v2.north);
        \draw[qleg] (v2.south) to[in = 300, out = 240] node[below] {$e_4$}(v1.south);
        \draw[qleg] (v2) -- node[above] {$e_2$} (v4);
        \draw[qleg] (v1) -- node[above] {$e_1$} (v3);
    },
\end{equation}
consider an intervention where $\labarg{A}=\{e_3,e_4\}$ and $\labarg{B}={e_1}$, i.e., 
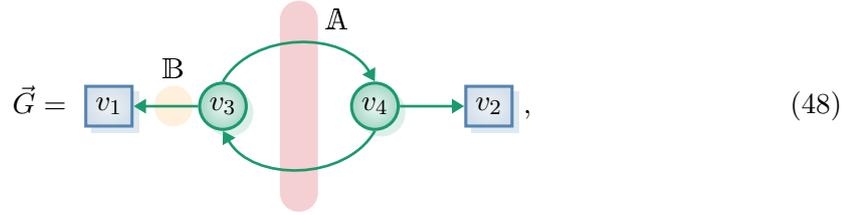
\begin{equation}
    \graphnamedir=\centertikz{
        \node at (0.5,1.15) {$\labarg{A}$};
        \node at (-1.65,0.5) {$\labarg{B}$};
        \draw[line width=5mm,red!20,line cap=round] (0,1.15) -- (0,-1.15);
        \draw[line width=5mm,yellow!20,line cap=round] (-1.65,0.02) -- (-1.65,-0.02);
        \node[unode] (v1) at (-1, 0) {$v_3$};
        \node[unode] (v2) at (1, 0) {$v_4$};
        \node[onode] (v4) at (2.5, 0) {$v_2$};
        \node[onode] (v3) at (-2.5, 0) {$v_1$};
        \draw[qleg] (v1.north) to[in = 120, out = 60]  (v2.north);
        \draw[qleg] (v2.south) to[in = 300, out = 240] (v1.south);
        \draw[qleg] (v2) --  (v4);
        \draw[qleg] (v1) --  (v3);
    },
\end{equation}
For simplicity, we can consider a causal graph which accounts for the addition of these intervention nodes. From the definition of interventions, we treat $\labarg{A}$ as unobserved (CPTP map) and $\labarg{B}$ as observed (instruments):
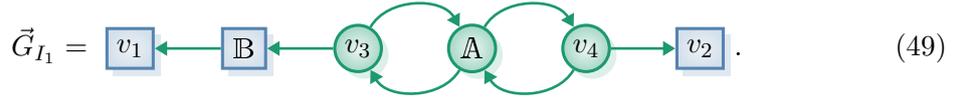
\begin{equation}
\label{eq:example_int1}
    \graphnamedir_{I_1}=\centertikz{
        \node[unode] (v3) at (-1.5, 0) {$v_3$};
        \node[unode] (v4) at (1.5, 0) {$v_4$};
        \node[onode] (v2) at (3, 0) {$v_2$};
        \node[onode] (v1) at (-4.5, 0) {$v_1$};
        \node[unode] (A) at (0,0) {$\labarg{A}$};
        \node[onode] (B) at (-3,0) {$\labarg{B}$};

        \draw[qleg] (v4) -- (v2);
        \draw[qleg] (v3) -- (B);
        \draw[qleg] (B) -- (v1);
        \draw[qleg] (v4) to[in = 300, out = 240] (A);
        \draw[qleg] (A) to[in = 300, out = 240] (v3);
        
        \draw[qleg] (v3) to[in = 120, out = 60] (A);
        \draw[qleg] (A) to[in = 120, out = 60] (v4);
    }.
\end{equation}
One can easily convince themself that applying operations in $\labarg{A}$ will in general affect the distribution in $\labarg{B}$, unless there is fine-tuning in the map associated to $\vertname_3$. Differently, if we consider interventions where $\labarg{A}={e_2}$ and $\labarg{B}={e_1}$, i.e.,
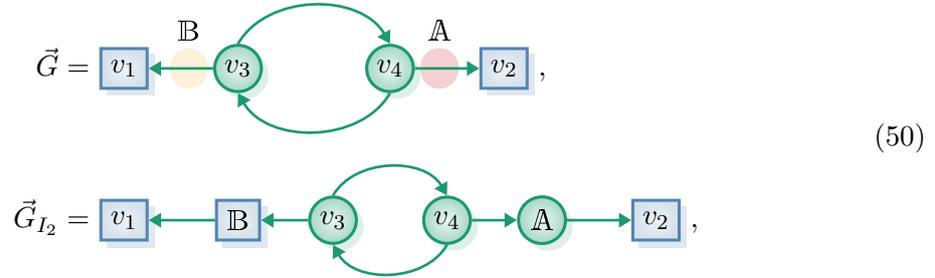
\begin{equation}
\label{eq:example_int2}
\begin{split}
     \graphnamedir=&\centertikz{
      \node at (1.65,0.5) {$\labarg{A}$};
        \node at (-1.65,0.5) {$\labarg{B}$};
        \draw[line width=5mm,red!20,line cap=round] (1.65,0.02) -- (1.65,-0.02);
        \draw[line width=5mm,yellow!20,line cap=round] (-1.65,0.02) -- (-1.65,-0.02);
        \node[unode] (v1) at (-1, 0) {$v_3$};
        \node[unode] (v2) at (1, 0) {$v_4$};
        \node[onode] (v4) at (2.5, 0) {$v_2$};
        \node[onode] (v3) at (-2.5, 0) {$v_1$};
        \draw[qleg] (v1.north) to[in = 120, out = 60]  (v2.north);
        \draw[qleg] (v2.south) to[in = 300, out = 240] (v1.south);
        \draw[qleg] (v2) --  (v4);
        \draw[qleg] (v1) --  (v3);
    },\\
    \graphnamedir_{I_2}=&\centertikz{
        \node[unode] (v3) at (-0.75, 0) {$v_3$};
        \node[unode] (v4) at (0.75, 0) {$v_4$};
        \node[onode] (v2) at (3.5, 0) {$v_2$};
        \node[onode] (v1) at (-3.5, 0) {$v_1$};
        \node[onode] (B) at (-2,0) {$\labarg{B}$};
        \node[unode] (A) at (2,0) {$\labarg{A}$};        
        \draw[qleg] (v3.north) to[in = 120, out = 60]  (v4.north);
        \draw[qleg] (v4.south) to[in = 300, out = 240] (v3.south);
        \draw[qleg] (v3) --  (B);
        \draw[qleg] (B) --  (v1);
        \draw[qleg] (v4) -- (A);
        \draw[qleg] (A) --  (v2);
    },
\end{split}
\end{equation}
we see that modifications in $\labarg{A}$ will at most affect the distribution in $\vertname_2$ but not the distribution in $\labarg{B}$.
The fact that we have no-signalling in this case is immediate from analysing graph separation properties (see~\Cref{app:graphsep}) of 
$\graphnamedir_{I_1}$. In particular, within the framework of~\cite{Quantum_paper} this follows directly from the notion of $p$-separation, which provides a systematic graphical criterion to infer conditional independencies—and hence no-signalling relations—from the structure of the intervention causal graph. For finite-dimensional quantum causal models on cyclic graphs, $p$-separation has been shown to be sound and complete \cite{Quantum_paper}, meaning that whenever two vertices are $p$-separated given a third, the corresponding conditional independence holds in the observed probability distribution. In the present example, one can show that $\labarg A$ and $\labarg B$ are $p$-separated, and therefore conclude immediately that there is no-signalling from $\labarg A$ to $\labarg B$. A more detailed analysis of interventions and the use of graph separation in quantum cyclic causal models will be presented in upcoming work focused specifically on this topic~\cite{Interventions_paper}.

\section{Mappings between tensor networks and cyclic causal models}
\label{sec: mappings}
Tensor networks and causal models, reviewed in the previous sections, are both graphical frameworks for describing relations between processes, but with different aims and assumptions. In this section, we connect them concretely by formulating mappings from one to the other. 

\subsection{Mapping causal models to tensor networks}
\label{sec:CMtoTNmap}
Let us first formulate the mapping from a causal model to a unique tensor network.

\begin{definition}[Mapping $\cm$ to $\tn$]
\label{def:mapCMtoTNv2}

Let $\cm=(\graphnamedir,\{\hilmaparg{\edgename}\}_{\edgename\in\edgesetdir}, \{\mathcal{T}^\vertname\}_{\vertname\in\vertsetdir})$ be a causal model. We map it to a tensor network $\tn=\left(\graphnameud,\{\hilmaparg{\edgename}\}_{\edgename\in\edgesetud},\{\projarg{\vertname}\}_{\vertname\in\vertsetud}\right)$ on an undirected graph $\graphnameud=(\vertsetud,\edgesetud)$, as follows:
\begin{myitem}
    \item the undirected graph $\graphnameud$ has the same vertex and edge set as $\graphnamedir$, where the edges are deprived of ordering\footnote{It might be that $\graphnamedir$ has edges $(v,w)$ and $(w,v)$ forming a cycle of two vertices. In this case, add to $\graphnameud$ two edges $\{v,w\}$ and associate the Hilbert spaces according to those of the causal model.}, i.e., to each ordered pair $(v,w)\in\edgesetdir$ corresponds a set $\{v,w\}\in\edgesetud$;
    \item the Hilbert space associated to $\edgename\in\edgesetud$ is the same associated to the corresponding directed edge in $\graphnamedir$ in $\cm$;
    \item to each vertex $\vertname\in\vertsetud$ we associate $\projarg{\vertname}=\CJ_{\outedges{\vertname}|\inedges{\vertname}}\left(\chanmaparg{\vertname}\right)$ if $\vertname$ is unobserved in the causal graph of $\cm$, and $\projarg{\vertname}=\CJ_{\outedges{\vertname}|\inedges{\vertname}}\left(\sum_{x\in\outcomemaparg{\vertname}}\measmaparg{\vertname}{\outcome}\right)$ if $\vertname$ is observed.
    
    Thus, 
    \begin{equation}
        \rhop=\bigotimes_{\vertname\in\vertsetud}\projarg{\vertname}  = \CJ_{\edgesetdir_{\textup{out}}|\edgesetdir_{\textup{in}}}\left(\sum_x\etot^x\right),
    \end{equation}
    where $\hilmaparg{\edgesetdir_{\textup{in}}}\cong\hilmaparg{\edgesetdir_{\textup{out}}}\cong\hilmaparg{\edgesetdir}$ and $\etot^x$ is the total map of $\cm$.
\end{myitem}
     We denote the image tensor network of a given causal model as $\mapCMtoTN(\cm)=\tn$.
\end{definition}
Notice that the mapping gives a well-defined tensor network as for all $\vertname\in\vertsetud$, $\projarg{\vertname}$ is a valid density operator. Indeed, this follows by noticing that the sum over the CP maps of observed vertices is CPTP and by~\cref{lemma:inv_choi}.

\begin{remark}
\label{rem:obsunobs}
    
The mapping of~\cref{def:mapCMtoTNv2} is defined for arbitrary causal models through marginalising over the set of CP maps associated to observed vertices before applying the Choi-Jamiołkowski isomorphism. Equivalently, one can formulate the mapping for causal models with only unobserved vertices as the results involving signalling and causal influence are unaffected by the way we model observed variables. Indeed, according to~\cref{def: signalling} observed variables, which are not arising from interventions, are marginalised.

From a mathematical perspective, considering only causal models with unobserved vertices is not restrictive, as observed outcomes and their statistics can be studied explicitly through interventions. However, from a fundamental perspective, one can think of purely unobserved causal models as describing underlying evolutions in a theory that are not being changed through interventions of the agents under consideration.
\end{remark}

\paragraph{Acyclic example}
Let us consider as an example the causal model of a Bell scenario, on the graph
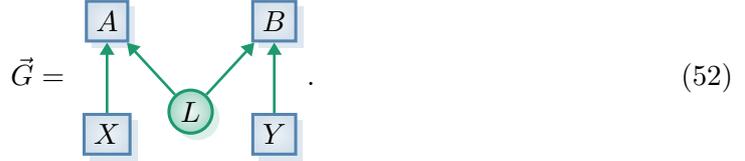
\begin{equation}
    \graphnamedir=\centertikz{
    \begin{scope}[xscale=2.2,yscale=1.5]
        \node[onode] (x) at (0,0) {$X$};
        \node[onode] (y) at (1,0) {$Y$};
        \node[onode] (a) at (0,1) {$A$};
        \node[onode] (b) at (1,1) {$B$};
        \node[unode] (l) at (0.5,0.2) {$L$};
        \draw[qleg] (x) -- (a);
        \draw[qleg] (y) -- (b);
        \draw[qleg] (l) -- (a);
        \draw[qleg] (l) -- (b);
    \end{scope}
    }.
\end{equation}
Let us label the Hilbert spaces as follows: $\hilmaparg{(X,A)} = \hilmaparg{X}$, $\hilmaparg{(L,A)} = \hilmaparg{L_1}$, $\hilmaparg{(L,B)} = \hilmaparg{L_2}$, and $\hilmaparg{(Y,B)} = \hilmaparg{Y}$.
To the exogenous vertices we associate the states
\begin{equation}
    \rho_{X}^x =p_x^X\ketbra{x}_X \quad \rho_{Y}^y =p_y^Y\ketbra{y}_Y \quad \rho_L = \ketbra{\bellstate}_{L_1L_2}
\end{equation} 
and to the observed vertices $A$ and $B$ the maps 
\begin{equation}
    \measmaparg a A(\cdot) = \Tr_{L_1X}\left[\povmelarg{a}{A}(\cdot)\right] \quad \textup{and} \quad  \measmaparg b B(\cdot) = \Tr_{L_2Y}\left[\povmelarg{b}{B}(\cdot)\right]. 
\end{equation}

The image tensor network is defined on the graph
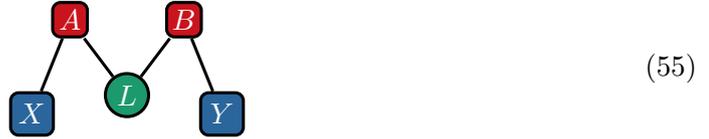
\begin{equation}
\centertikz{
        \node[tnode, fill=green, minimum size=13pt] (l) at (0,0){\color{white}{$L$}};
        \node[tnrect, fill=red, minimum size=13pt] (a) at (-0.75,1.){\color{white}{$A$}};
        \node[tnrect, fill=red, minimum size=13pt] (b) at (0.75,1.){\color{white}{$B$}};
        \node[tnrect, fill=blue, minimum size=16pt] (x) at (-1.25,-0.25){\color{white}{$X$}};
        \node[tnrect, fill=blue, minimum size=16pt] (y) at (1.25,-0.25){\color{white}{$Y$}};
        \draw[mybaseline] (x) -- (a);
        \draw[mybaseline] (y) -- (b);
        \draw[mybaseline] (l) -- (a);
        \draw[mybaseline] (l) -- (b);
}
\end{equation}
and has tensors:
\begin{equation}
    \projarg{X} =\sum_x p_x^X\ketbra{x}_X, \quad \projarg{Y} =\sum_yp_y^Y\ketbra{y}_Y, \quad \projarg{L} = \ketbra{\bellstate}_{L_1L_2}
\end{equation}
and 
\begin{equation}
    \projarg{A} =\CJ_{|XL_1}\left(\sum_a\measmaparg a A\right), \quad  \projarg{B} =\CJ_{|YL_2}\left(\sum_b\measmaparg b B\right).
\end{equation}

\paragraph{Cyclic example} Let us consider a causal model whose underlying graph is cyclic. We have
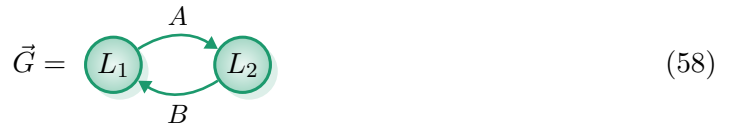
\begin{equation}
   \graphnamedir= \centertikz{
        \begin{scope}[xscale=1.7,yscale=2.1]
            \node[unode] (l1) at (0,0) {$L_1$};
            \node[unode] (l2) at (1,0) {$L_2$};
            \draw[qleg] (l1) to[out=30,in=150] node[above] {\small$A$} (l2);
            \draw[qleg] (l2) to[out=210,in=330] node[below] {\small$B$} (l1);
        \end{scope}
    }
\end{equation}
where we labelled the Hilbert spaces as $\hilmaparg{(L_1,L_2)}=\hilmaparg{A}$ and $\hilmaparg{(L_2,L_1)}=\hilmaparg{B}$. We associate to the vertices the CPTP maps
\begin{equation}
    \begin{split}
        \chanmaparg{L_1}:\hilmaparg{B}\mapsto\hilmaparg{A} \quad \textup{and} \quad \chanmaparg{L_2}:\hilmaparg{A}\mapsto\hilmaparg{B}.
    \end{split}
\end{equation}
The image tensor network is defined on the graph
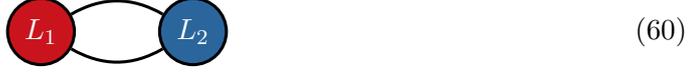
\begin{equation}
\centertikz{
        \node[tnode, fill=red, minimum size=25pt] (l1) at (-1,0){\color{white}{$L_1$}};
        \node[tnode, fill=blue, minimum size=25pt] (l2) at (1,0){\color{white}{$L_2$}};
        \draw[mybaseline] (l1) to[out=30,in=150](l2);
        \draw[mybaseline] (l2) to[out=210,in=330] (l1);
}
\end{equation}
and has tensors
\begin{equation}
    \begin{split}
        \projarg{L_1}=\CJ_{A|B}\left(\chanmaparg{L_1}\right) \quad \textup{and} \quad \projarg{L_2}=\CJ_{B|A}\left(\chanmaparg{L_2}\right).
    \end{split}
\end{equation}

\subsection{Mapping tensor networks to causal models}
\label{sec:TNtoCMmap}
Now we present the mapping in the opposite direction, from tensor networks to causal models. The goal is to construct a mapping which always allows to obtain a causal model from a given tensor network. We achieve this in two steps: first we present a mapping which preserves the graph of the tensor network, up to introducing directionality, but does not provide a causal model for arbitrary directions, then we generalise it showing that, up to introducing further cyclicity, each choice of direction in the tensor networks maps to a valid causal model.

\subsubsection{Mapping for a restricted case}

Firstly, we define how to construct a causal graph associated to a given tensor network.

\begin{definition}[Causal graph of $\tn$]
\label{def:graphundirtodir}
    Let us consider an undirected graph $\graphnameud=(\vertsetud,\edgesetud)$ with vertices labelled by distinct integers\footnote{Such labelling can always be arbitrarily chosen.}. Given a binary string $D = \{d_\edgename\}_{\edgename\in\edgesetud}$, $d_\edgename\in\{0,1\}$, we construct a directed causal graph $\graphnamedir_D=(\vertsetdir,\edgesetdir)$ as follows:
    \begin{myitem}
        \item the vertex set $\vertsetdir$ of $\graphnamedir_D$ equals those of $\graphnameud$, i.e., $\vertsetdir=\vertsetud$, and all vertices are unobserved;
        \item for each edge $\edgename=\{v_i,v_j\}\in\edgesetud$ with $i<j$, add to $\edgesetdir$ the ordered pair $(v_i,v_j)$ if $d_\edgename=0$ or the ordered pair $(v_j,v_i)$ if $d_\edgename=1$.
    \end{myitem}
\end{definition}

In words, the string $D$ encodes a choice of direction $d_\edgename$ for each undirected edge in the graph $\graphnameud$. The directed graph $\graphnamedir_D$ is then obtained implementing such choice. For each causal graph constructed in this way, we can now attempt to define a causal model.

\begin{definition}[Mapping $\tn$ to $\cm$]
\label{def: TNtoCM}
    Let $\tn=(\graphnameud,\{\hilmaparg{\edgename}\}_{\edgename\in\edgesetud},\{\projarg{\vertname}\}_{\vertname\in\vertsetud})$ be a tensor network, and consider a binary string $D\in\{0,1\}^{|\edgesetud|}$. Let $\graphnamedir_D$ be the directed graph constructed through~\cref{def:graphundirtodir}. If the density operators of the tensor network satisfy
    \begin{equation}
    \label{eq: trace_condition}
        \Tr_{\outedges{\vertname}}[\projarg{\vertname}] = \frac{\id_{\inedges{\vertname}}}{d_{\inedges{\vertname}}}
    \end{equation}
     for all $\vertname\in\vertsetdir$ where $\outedges{\vertname}$ and $\inedges{\vertname}$ are relative to $\graphnamedir_D$and where $d_{\inedges{\vertname}}=\dim(\hilmaparg{\inedges{\vertname}})$, we map the tensor network to a causal model $\cm=(\graphnamedir_D,\{\hilmaparg{\edgename}\}_{\edgename\in\edgesetdir}, \{\mathcal{T}^\vertname\}_{\vertname\in\vertsetdir})$ obtained as follows:
    \begin{myitem}
        \item to each directed edge $\edgename=(v,w)\in\edgesetdir$ associate the Hilbert space associated with the corresponding undirected edge in $\graphnameud$, $\hilmap_\edgename$;
        \item to each vertex $\vertname\in\vertsetdir$, unobserved by construction, associate the CPTP map 
        \begin{equation}
            \chanmaparg{\vertname} = \CJin_{\outedges{\vertname}|\inedges{\vertname}}\left(\projarg{\vertname}\right).
        \end{equation}
    \end{myitem}
    We denote the image causal model of a given tensor network as $\mapTNtoCM_D(\tn)$.
\end{definition}

Given a choice of orientation of the undirected graph underlying the tensor network, the above definition defines a causal model only if the trace condition in~\cref{eq: trace_condition} is satisfied. Such condition ensures that the maps associated to vertices of the causal model are CPTP, and as a consequence, it ensures that the tensor network is mapped to a well-defined causal model. According to this construction, a tensor network can be mapped to multiple causal models or none depending on the maps. For instance, consider the tensor network
\begin{equation}
\centertikz{
        \node[tnrect, fill=red, minimum size=18pt] (A) at (-1,0){\color{white}{$A$}};
        \node[tnrect, fill=blue, minimum size=18pt] (B) at (1,0){\color{white}{$B$}};
        \draw[mybaseline] (A) -- (B);
}
\end{equation}
where $\projarg{A}$ and $\projarg{B}$ are two density operators different from the maximally mixed state. Since there is only one edge, two choices of $D$ are possible, leading to the directed graphs
\begin{equation}
    D={0}:\quad \graphnamedir_0= \centertikz{
        \node[unode] (A) at (-1,0){$A$};
        \node[unode] (B) at (1,0){$B$};
        \draw[qleg] (A) -- (B);
} \quad D={1}:\quad \graphnamedir_1 =\centertikz{
        \node[unode] (A) at (-1,0){$A$};
        \node[unode] (B) at (1,0){$B$};
        \draw[qleg] (B) -- (A);
}.
\end{equation}
One can define a causal model on $\graphnamedir_0$ and $\graphnamedir_1$, if respectively
\begin{equation}
    \begin{split}
        D={0}:\quad \Tr(\projarg{A})=1 \textup{ and } \projarg{B}=\frac{\id}{d}\\
        D={1}:\quad \projarg{A}=\frac{\id}{d} \textup{ and } \Tr(\projarg{B})=1
    \end{split}
\end{equation}
since we assumed that $\projarg{A}$ and $\projarg{B}$ are both different from the maximally mixed state, the tensor network cannot be mapped to a causal model. 
Notice that if we assume that only one between $\projarg{A}$ and $\projarg{B}$ is equal to the maximally mixed state, we obtain exactly one causal model associated to the tensor network. While if $\projarg{A}=\projarg{B}=\frac{\id}{d}$, both choices lead to a valid causal model.

In the next paragraph, we propose a generalisation of the mapping that allows us to map a given tensor network to a causal model for each choice of $D$.
In~\Cref{sec:holography}, we will study more complex examples of the mapping and their application, specifically in the case of holographic tensor network~\cite{Pastawski_2015}.
\subsubsection{Mapping for the general case}
\label{sec:TNtoCMmapgen}

In~\cref{def: TNtoCM}, we only map a tensor network to causal models which fulfils the partial trace condition in~\cref{eq: trace_condition}. If such a condition is not satisfied in $\graphnamedir_{D}$ for any $D\in\{0,1\}^{|\edgesetud|}$, then the mapping of~\cref{def: TNtoCM} is not defined for such a tensor network. In this section, we provide a more general mapping which allows us to map a given tensor network to valid causal models for all choices of direction $D$.
The idea is to construct another directed graph $\graphnamedir_D^{\circlearrowleft}$ which introduces self-cycles for all vertices that do not fulfil~\cref{eq: trace_condition}. Then, by using the results of \cite{jean_2025} we can map the tensor network to a valid causal model on $\graphnamedir_D^{\circlearrowleft}$.
\begin{restatable}{lemma}{TNtoCMgen}
    \label{lemma: TNtoCMgeneral}
    Let $\tn=(\graphnameud,\{\hilmaparg{\edgename}\}_{\edgename\in\edgesetud},\{\projarg{\vertname}\}_{\vertname\in\vertsetud})$ be a tensor network, $D\in\{0,1\}^{|\edgesetud|}$ a binary string and $\graphnamedir_D$ the associated directed graph constructed according to~\cref{def:graphundirtodir}, with ingoing and outgoing edges to $\vertname\in\vertsetdir_D$ denotes as $\inedges{\vertname}$ and $\outedges{\vertname}$. 
    For all $\vertname\in \vertsetud$, the density operator $\projarg{\vertname}\in\linops\left(\hilmaparg{\incedges{\vertname}}\right)\cong\linops\left(\hilmaparg{\inedges{\vertname}}\otimes\hilmaparg{\outedges{\vertname}}\right)$, satisfies
    \begin{equation}
        \label{eq: cycle_TNtoCMgen}
        \CJin_{\outedges{\vertname}|\inedges{\vertname}}\left(\projarg{\vertname}\right)=\alpha_\vertname\selfcycle_{A_\vertname}(\chanmaparg{\vertname}),
        \end{equation}
   where $\alpha_v$ is a real strictly positive constant and $\chanmaparg{\vertname}$ is a CPTP map
    \begin{equation}
        \label{eq: CPTP_TNtoCMgen}
            \chanmaparg{\vertname}: \linops(\hilmaparg{A_\vertname}\otimes\hilmaparg{\inedges{\vertname}}) \mapsto \linops(\hilmaparg{A_\vertname}\otimes \hilmaparg{\outedges{\vertname}}),
        \end{equation}
    and $\hilmaparg{A_\vertname}$ is an ancillary Hilbert space of dimension $\dim_{A_\vertname}= 2\dim_{\inedges{\vertname}}$.
\end{restatable}
\begin{proof}
    The proof uses results from~\cite{jean_2025}, see~\Cref{app:proofs}.
\end{proof}

\begin{definition}[General mapping $\tn$ to $\cm$]
\label{def: TNtoCMgeneral}
   Let $\tn=(\graphnameud,\{\hilmaparg{\edgename}\}_{\edgename\in\edgesetud},\{\projarg{\vertname}\}_{\vertname\in\vertsetud})$ be a tensor network and $D\in \{0,1\}^{|\edgesetud|}$  a binary string. Consider the directed graph $\graphnamedir_D = (\vertsetdir,\edgesetdir)$ constructed from $\graphnameud$ and $D$ through~\cref{def:graphundirtodir}, and let $\vertset^\neg\subseteq\vertsetdir$ be the set of vertices for which~\cref{eq: trace_condition} is not satisfied, i.e., $\vertname\in\vertset^\neg$ if and only if
    \begin{equation}
        \Tr_{\outedges{\vertname}}[\projarg{\vertname}] \neq \frac{\id_{\inedges{\vertname}}}{d_{\inedges{\vertname}}}.
    \end{equation}
    Then, define the directed graph $\graphnamedir_D^{\circlearrowleft}= (\vertsetdir^\circlearrowleft,\edgesetdir^\circlearrowleft)$ from $\graphnamedir_D$ as follows: $\vertsetdir^\circlearrowleft=\vertsetdir$ and $\edgesetdir^\circlearrowleft = \edgesetdir\cup\{(\vertname,\vertname)\}_{\vertname\in\vertset^\neg}$, i.e., the graphs are the same up to adding self-loops for all vertices in $\vertset^{\neg}$. We map the tensor network $\tn$ to a causal model on the directed graph $\graphnamedir^{\circlearrowleft}_D$ as follows:
    \begin{myitem}
        \item to each directed edge $\edgename=(v,w)\in\edgesetdir\subseteq \edgesetdir^{\circlearrowleft}$ associate in $\graphnamedir_D^\circlearrowleft$ the Hilbert space associated with the corresponding undirected edge in $\graphnameud$, $\hilmap_\edgename$; 
        \item to each vertex $\vertname\in \vertsetdir\backslash \vertset^{\neg}$ associate the CPTP map\footnote{Notice that this is CPTP since \cref{eq: trace_condition} is satisfied by construction}
         \begin{equation}
            \chanmaparg{\vertname} = \CJin_{\outedges{\vertname}|\inedges{\vertname}}\left(\projarg{\vertname}\right);
        \end{equation}
        
        \item to each directed edge $\edgename=(v,v)\in\edgesetdir^{\circlearrowleft}\backslash \edgesetdir$, associate the Hilbert space $\hilmaparg{A_\vertname}$ of the ancilla system $A_\vertname$ as defined in \cref{lemma: TNtoCMgeneral};
        \item to each vertex $\vertname\in\vertsetdir^{\neg}$, associate the CPTP map 
        \begin{equation}
            \chanmaparg{\vertname}: \linops\left(\hilmaparg{A_\vertname}\otimes \hilmaparg{\inedges{\vertname}}\right) \mapsto \linops\left(\hilmaparg{A_\vertname}\otimes \hilmaparg{\outedges{\vertname}}\right)
        \end{equation}
        satisfying 
        \begin{equation}
            \selfcycle_{A_\vertname}(\chanmaparg{\vertname})=\frac{1}{\alpha_\vertname}\CJin_{\outedges{\vertname}|\inedges{\vertname}}\left(\projarg{\vertname}\right),
        \end{equation}
       for some positive non-zero constant $\alpha_v$. Such map and the constant $\alpha_v$ are explicitly constructed in the proof of \cref{lemma: TNtoCMgeneral}~(see~\Cref{app:proofs}).
    \end{myitem}
    We denote the image causal model of a given tensor network as $\mapTNtoCMgen_D(\tn)$. 
    
\end{definition}

Notice that the generalised mapping trivially reduces to the one of~\cref{def: TNtoCM} if the partial trace condition in~\cref{eq: trace_condition} is satisfied for all $\vertname\in\vertsetud$. Indeed, in this case the set $\vertset^\neg$ is empty, and the associations corresponding to the first two items above are equal to those of~\cref{def: TNtoCM}.

\begin{remark}
    The generalised mapping presented in~\cref{def: TNtoCMgeneral} explicitly allows to construct the relevant CPTP map, up to introducing an ancillary system having twice the dimension of the input space. Specifically, the ancilla contains a system isomorphic to the input space tensored with a qubit space. A different result of~\cite{jean_2025} proves the existence of a CPTP map which also satisfies~\cref{eq: cycle_TNtoCMgen} with ancilla space isomorphic to the input space only. This, however, is an existence argument, and does not fully specify the map $\chanmaparg{\vertname}$. 
Thus, one could modify the above definition of the mapping and obtain an existence statement instead of a constructing one, with ancillas of half the size. 
In either case, the result is a causal model on $\graphnamedir_D^\circlearrowleft$ where~\cref{eq: cycle_TNtoCMgen} holds for every vertex $\vertname\in\vertsetud$, but the exact maps and dimension of the Hilbert space associated to the additional self-cycle edges $(\vertname,\vertname)$ would differ. 

\end{remark}

\paragraph{Example} Let us clarify the generalised mapping using the example of the previous section on the tensor network
\begin{equation}
\graphnameud=\centertikz{
        \node[tnrect, fill=red, minimum size=18pt] (A) at (-1,0){\color{white}{$A$}};
        \node[tnrect, fill=blue, minimum size=18pt] (B) at (1,0){\color{white}{$B$}};
        \draw[mybaseline] (A) -- (B);
}
\end{equation}
where $\projarg{A}$ and $\projarg{B}$ are two density operators different from the maximally mixed state. As we showed before, both choices of direction for the unique edge, i.e.,
\begin{equation}
    \graphnamedir_{D=0}:\quad \centertikz{
        \node[unode] (A) at (-1,0){$A$};
        \node[unode] (B) at (1,0){$B$};
        \draw[qleg] (A) -- (B);
} \quad  \graphnamedir_{D=1}:\quad \centertikz{
        \node[unode] (A) at (-1,0){$A$};
        \node[unode] (B) at (1,0){$B$};
        \draw[qleg] (B) -- (A);
}
\end{equation}
cannot be mapped to a causal model using~\cref{def: TNtoCM}. 
Through the generalised mapping, we construct causal models on the directed graphs:
\begin{equation}
    \graphnamedir^\circlearrowleft_{D=0}:\quad \centertikz{
        \node[unode] (A) at (-1,0){$A$};
        \node[unode] (B) at (1,0){$B$};
        \draw[qleg] (A) -- (B);
        \draw[qleg,rounded corners=15pt] (B) -- ++(0.6,-0.8) -- ++(0.6,0.8) -- ++(-0.6,0.8) -- (B);
} \quad  \graphnamedir^\circlearrowleft_{D=1}:\quad \centertikz{
        \node[unode] (A) at (-1,0){$A$};
        \node[unode] (B) at (1,0){$B$};
        \draw[qleg] (B) -- (A);
        \draw[qleg,rounded corners=15pt] (A) -- ++(-0.6,-0.8) -- ++(-0.6,0.8) -- ++(0.6,0.8) -- (A);
}.
\end{equation}
To give an explicit example, consider a tensor network on $\graphnameud$ with a qubit Hilbert space associated to the unique edge, and let us denote the system as $S$, i.e., $\hilmap_e\equiv\hilmap_S = \mathbb{C}^2$, and take $P_A=\ketbra{+}$ and $P_B=\ketbra{0}$.
Then, for $D=0$ the image causal model on $\graphnamedir^\circlearrowleft_{D=0}$ is given by associating a qubit space to the ancilla system, which we denote as $H$, and maps $\rho_{A}=P^A$ and $\chanmaparg{B}:\linops(\hilmap_S\otimes \hilmaparg{H})\mapsto\linops(\hilmap_S\otimes \hilmaparg{H})$ where $\chanmaparg{B}=\textup{CNOT}_{S\to H}$ with control on the system $S$.
Indeed, we have
\begin{equation}
    \CJin_{|S}(\ketbra{0})(\rho_S) = d_H \sum_{i,j}\braaket{0}{i}\bra{j}\rho_S\ket{i}\braaket{j}{0} = d_H\bra{0}\rho_S\ket{0} \ketbra{0}_S
\end{equation}
and on the other hand
\begin{equation}
\selfcycle_H(\text{CNOT}_{S\to H})(\rho_S)= \sum_{i,j}\bra{i}(\text{CNOT}_{S\to H}(\rho_S\otimes \ketbraa{i}{j})\ket{j} = \bra{0}\rho_S\ket{0} \ketbra{0}_S.
\end{equation}
For $D=1$, the result is analogous up to applying a Hadamard gate before the CNOT.

The intuition behind this construction is that the instances where the ancilla system is flipped lead to a null post-selection probability in the self-cycle. Thus, this effectively reduces the CPTP map on the system and ancilla to a CP map on the system which represents the action in the case where the ancilla is not flipped (in the CNOT, this corresponds to projecting the system, which is the control, on $\ket{0}$).

\subsection{Properties of the mappings}
\label{sec:prop_maps}
In this section, we provide properties of the mappings that we have just constructed (\cref{def:mapCMtoTNv2,def: TNtoCM,def: TNtoCMgeneral}). The proofs of the results in this section can be found in~\Cref{app:proofs}.

\paragraph{Composition of mappings}
A first relevant property involves the composition of the two mappings, which is the content of the following lemma.
\begin{restatable}[]{lemma}{compmaps}
\label{lemma:composition_maps}
     Let $\cm=(\graphnamedir,\{\hilmaparg{\edgename}\}_{\edgename\in\edgesetdir}, \{\mathcal{T}^\vertname\}_{\vertname\in\vertsetdir})$ be a causal model such that $\overtset=\emptyset$ and assume that the vertices of $\graphnamedir$ are labelled by distinct integers. Define the binary string $D_{\graphnamedir} = \{d_e\}_{\edgename\in\edgesetdir}$ where for $\edgename=(\vertname_i,\vertname_j)\in\edgesetdir$, $d_\edgename=0$ if $j\geq i$ and $d_\edgename=1$ if $j<i$, i.e., the sting $D_{\graphnamedir}$ simply records the direction of each edge in the directed graph $\graphnamedir$. Then, it holds:
    \begin{equation}
        \label{eq:isomfromCM}\mapTNtoCM_{D_{\graphnamedir}}\circ\mapCMtoTN(\cm)=\cm,
    \end{equation}
    where $\mapCMtoTN$ is as in~\cref{def:mapCMtoTNv2} and $\mapTNtoCM$ as in~\cref{def: TNtoCM}.
    Similarly, consider a tensor network $\tn = \left(\graphnameud,\{\hilmaparg{\edgename}\}_{\edgename\in\edgesetud},\{\projarg{\vertname}\}_{\vertname\in\vertsetud}\right)$ and any $D\in \{0,1\}^{|\edgesetud|}$ such that the mapping $\mapTNtoCM_{D}$ is well defined, then it holds:
    \begin{equation}
    \label{eq:isomfromTN}
        \mapCMtoTN\circ\mapTNtoCM_{D}(\tn)=\tn.
    \end{equation}
\end{restatable}
\begin{proof}
  See~\Cref{app:proofs}.  
\end{proof}

Notice that, when keeping track of the direction of the initial directed graph, the mappings define an isomorphism. This has to be expected as a given tensor network is uniquely mapped into a causal model only for a fixed and allowed $D$. This property is analogous to the Choi-Jamiołkowsi isomorphism, which is indeed an isomorphism only for a fixed choice of input and output in the inverse Choi. \Cref{eq:isomfromCM} is valid only for causal models with no observed vertices. This is not restrictive for our purpose, as already argued in~\cref{rem:obsunobs}. Indeed, as will see in the following section, for studying signalling properties of the mappings, the relevant quantities always involve marginalization over the observed outcomes of the model.\\

In the case of the generalised mapping,~\cref{eq:isomfromCM} is trivially satisfied by replacing $\mapTNtoCM_{D_{\graphnamedir}}$ with $\mapTNtoCMgen_{D_{\graphnamedir}}$ since all partial trace conditions are satisfied by construction, thus the two mappings are the same for such choice of direction $D_{\graphnamedir}$. However,~\cref{eq:isomfromTN} is not exactly satisfied for the generalised mapping because of the addition of self-loops in the causal model. The following holds instead:
\begin{restatable}[]{lemma}{compmapsgen}
\label{lemma: compositiongen}
     Consider a tensor network $\tn = \left(\graphnameud,\{\hilmaparg{\edgename}\}_{\edgename\in\edgesetud},\{\projarg{\vertname}\}_{\vertname\in\vertsetud}\right)$ and $D\in \{0,1\}^{|\edgesetud|}$, then it holds:
    \begin{equation}
    \label{eq:isomfromTNgen}
        \mapCMtoTN\circ\mapTNtoCMgen_{D}(\tn)=\left(\graphnameud^{\circlearrowleft},\{\hilmaparg{\edgename}^{\circlearrowleft}\}_{\edgename\in\edgesetud^{\circlearrowleft}},\{\projarg{\vertname}^{\circlearrowleft}\}_{\vertname\in\vertsetud}\right),
    \end{equation}
    where
    \begin{myitem}
        \item the undirected graph $\graphnameud^{\circlearrowleft}$ has the same vertex and edge set as $\graphnamedir^{\circlearrowleft}$, the directed graph of the causal model $\mapTNtoCMgen_{D}(\tn)$ according to~\cref{def: TNtoCMgeneral}, where the edges are deprived of ordering, i.e., to each ordered pair $(v,w)\in\edgesetdir^{\circlearrowleft}$ corresponds a set $\{v,w\}\in\edgesetud^{\circlearrowleft}$;
        \item $\hilmaparg{\edgename}^{\circlearrowleft} = \hilmaparg{\edgename}$ for all $\edgename\in\edgesetud\subseteq \edgesetud^{\circlearrowleft}$;
        \item for all $\vertname\in\vertsetud$ it holds
        \begin{equation}
            \bra{\ubellstate}_{\edgesetud(v)}\projarg{\vertname}^{\circlearrowleft}\ket{\ubellstate}_{\edgesetud(v)} = \beta_\vertname\projarg{\vertname}
        \end{equation}
        where $\beta_\vertname>0$ is a real constant and $\edgesetud(v)$ is the set of edges adjacent to $\vertname$ in $\graphnameud^{\circlearrowleft}$ but not in $\graphnameud$, i.e.,      $\edgesetud(v)=\incedges{v}_{\graphnameud^{\circlearrowleft}}\setminus\incedges{v}_{\graphnameud}$.
    \end{myitem}
\end{restatable}
\begin{proof}
  See~\Cref{app:proofs}. 
\end{proof}
There is nevertheless a natural sense in which the original tensor network is recovered. Starting from the right-hand side of~\cref{eq:isomfromTNgen}, one may perform only a subset of the tensor network contractions, i.e., those corresponding to the added self-cycle edges, while leaving all other edges uncontracted.
This suggests introducing the notion of a reduced tensor network, analogous to the concept of reduced process matrices or, more generally, to partial composition of cyclic quantum networks as discussed in~\cite{VilasiniRennerPRA}. By a reduced tensor network, we mean a network obtained by contracting only a specified subset of edges of a larger network, thereby producing a new tensor network on the corresponding subgraph. In this sense, the original tensor network is recovered as the reduced tensor network on the subgraph $\graphnameud$ of $\graphnameud^\circlearrowleft$, obtained by contracting precisely the self-cycle edges that were introduced.

\paragraph{Tensor contraction and self-cycle}The following lemma proves that two quantities which are relevant respectively for tensor networks and causal models, namely the tensor network contraction and the self-cycle, are linked through the mappings.
This property is crucial in connecting causal influence to known notions of causality.

\begin{restatable}{lemma}{contselfcyclegen}
\label{lemma: TNtoCMgeneral_correl}
Let us consider a causal model on $\graphnamedir$ and a tensor network on $\graphnameud$, $\cm$ and $\tn$, such that $\mapCMtoTN(\cm)=\tn$ or $\mapTNtoCMgen_D(\tn)=\cm$ for some $D\in\{0,1\}^{|\edgesetud|}$. The marginalised total map of $\cm$, $\sum_x \etot^x =\etot$, and the link and total state of $\tn$, $\ket{L}$, $\rhop$, satisfy:
\begin{equation}
    \beta \;\selfcycle\left(\etot\right) = \bra{L}\rhop\ket{L},
\end{equation}
for some real constant $\beta>0$ determined solely by the constants $\alpha_v$ of~\cref{def: TNtoCMgeneral} and the dimensions of $\hilmaparg{\edgesetud}$.
\end{restatable}
\begin{proof}
    See~\Cref{app:proofs}.
\end{proof}

Then, the following corollary follows straightforwardly.

\begin{restatable}[]{corollary}{cordifferentD}
\label{cor:differentD}
Given a tensor network $\tn$ on $\graphnameud$ and two choices of directions $D_1,D_2\in\{0,1\}^{|E|}$, then the marginalised total maps of $\mapTNtoCMgen_{D_1}(\tn)$ and $\mapTNtoCMgen_{D_2}(\tn)$ satisfy
\begin{equation}
    \beta\selfcycle\left(\etot_1\right) = \selfcycle\left(\etot_2\right),  
\end{equation}
for some real constant $\beta>0$ determined by the constants $\alpha_v$ of~\cref{def: TNtoCMgeneral} and the dimensions of $\hilmaparg{\edgesetud}$.
\end{restatable}
This property has an interesting feature: a given causal model $\cm_1$ can be first mapped into a tensor network through $\mapCMtoTN(\cm_1)=\tn$, then the image tensor network can be mapped to a different causal model from the original one, $\cm_2$. Because of~\cref{cor:differentD}, it holds:
\begin{equation}
     \selfcycle\left(\etot_1\right) = \beta\selfcycle\left(\etot_2\right),  
\end{equation}
where $\etot_i$ is the marginalised total map of the causal model $\cm_i$ and $\beta=1$ if the mapping from $\tn$ to $\cm_2$ is associated with a valid choice of directions according to~\cref{def: TNtoCM}, $D$, while generally $\beta\in\mathbb{C}$.
Thus, one can use the mapping to construct a family of causal models that are closely related to each other. In the following, we will explore this feature even further, arguing that this class of causal models, arising from the same tensor network can be understood as space-time rotations of each other.

\section{Signalling in causal models vs causal influence in tensor networks}
\label{sec:signallingandci}
In this section, we show a connection between causal influence in tensor networks and signalling in causal models through the mappings defined in~\Cref{sec: mappings}.

First, we analyse the definition of causal influence in tensor networks and present a slight modification of it, providing it with a more operational interpretation, which will allow us to definitively conclude that it does capture causation in an operational sense, as opposed to generic correlations. 

\subsection{Operational analysis of causal influence proposed in~\cite{Cotler_2019}}
\label{sec: new ci}

The mappings of~\Cref{sec: mappings} highlight connections between the tensor network contraction and the self-cycle operation. 
These two quantities are relevant respectively in defining causal influence in a tensor network (see~\cref{def: causal influence old}) and probabilities in a casual model (see~\cref{prop:probs as self cycles}). Thus, we can use the connection between probabilities, which are operational quantities, and self-cycle to provide causal influence with an operational motivation (see \Cref{sec:open questions ci}). 

First, we recall the definition of causal influence provided in~\cite{Cotler_2019} and analyse its operational interpretation in connection with probability theory.

\oldci*

The arguments leading to the above definition have been discussed in~\Cref{sec:tn causal influence intro}. Now we know that the correlation functions appearing in~\cref{def: causal influence old} are connected to the self-cycle operation, which is related to probabilities in cyclic causal models~\cref{lemma: TNtoCMgeneral_correl}. We can thus rephrase the arguments leading to~\cref{def: causal influence old} in terms of probabilities. For now, we will make these arguments in terms of ``virtual'' measurements, whose physicality is unclear. Thus, the following are to be understood as intuitive arguments clarifying~\cref{def: causal influence old}, these will be made precise though the connection between causal influence and signalling in causal models (\Cref{app:cyctoacyc,rem:virtualmmt}):
\begin{itemize}
    \item \textbf{The tensor network contraction:} The tensor network contraction of \cref{def: tn contraction} can be written as $\bra{L}\rhop\ket{L}=\Tr\left(\ketbra{L}\rhop\right)\in\mathbb{C}$. If we define a binary measurement $\{\ketbra{L},\id-\ketbra{L}\}$, associated with outcomes $\{L,\neg L\}$, we can understand the tensor network contraction as a probability of observing outcome corresponding to the POVM element $\ketbra{L}$ on the state $\rhop$, and denote it with $P(L)_{\rhop}$.
    This ``virtual'' measurement on the $\rhop$ will then be linked to observable probabilities in acyclic casual models via our mapping from tensor networks to causal models (see~\Cref{app:cyctoacyc,rem:virtualmmt}).
    \item \textbf{The denominator:} The role of the denominator is intuitively explained in~\cite{Cotler_2019} by saying that it is needed in order to isolate the influence on $\region{B}$ of modifications in $\region{A}$ from the influence that these have on the whole network contraction. With the interpretation given above, we can understand the denominator as a probability of observing outcome corresponding to the POVM element $\ketbra{L}$ on the state $U_\region{A}\rhop U^\dagger_{\region{A}}$, i.e., $P(L)_{U_\region{A}\rhop U^\dagger_{\region{A}}}$.  

    We remark that the denominator may be zero, in case $P(L)_{U_\region{A}\rhop U^\dagger_{\region{A}}}=0$.

    \item \textbf{The operation inserted in $\region{B}$:} The operator $O_{\region{B}}$ inserted in $\region{B}$ is required to be Hermitian, so that no predefined direction is induced, and such that $O^\dagger_{\region{B}}O_{\region{B}}\leq \id_\region{B}$. Consider now a quantum instrument, corresponding to a measurement of a quantum system including the quantum evolution of the measured state. This is described by a linear, completely positive and trace preserving map, which can be decomposed in the Kraus form as
    \begin{equation}
        \mathcal{M}_{\region{B}X|\region{B}}(\rho_{\region{B}}) = \sum_x \ketbra{x}_X\otimes O^x_\region{B}\rho_{\region{B}} {O^{x}_{\region{B}}}^\dagger
    \end{equation}
    for some set of positive linear operators, $\mathcal{O}_\region{B}=\{O^x_{\region{B}}\in\linops(\hilmaparg{\region{B}})\}_{x}$, which we call Kraus operators and satisfy $\sum_{x} {O^{x}_{\region{B}}}^\dagger O^x_\region{B}=\id_{\region{B}}$. The un-normalised state of the system $\region{B}$ after the instrument is applied is
    \begin{equation}
        \mathcal{M}^{X=x}_{\region{B}|\region{B}}(\rho_{\region{B}}) = O^x_\region{B}\rho_{\region{B}} {O^{x}_{\region{B}}}^\dagger,
    \end{equation} 
    the above contains the probability associated to observing the outcome $X=x$, which is given by $\Tr( O^x_\region{B}\rho_{\region{B}} {O^{x}_{\region{B}}}^\dagger)$.
    Notice that each $O^x_{\region{B}}$ satisfies ${O^{x}_{\region{B}}}^\dagger O^x_\region{B}\leq\id_{\region{B}}$.
    
    Thus, combining this with the previous points, we can understand the operation inserted in $\region{B}$ in \cref{def: causal influence old} as being a Kraus operator corresponding to a specific outcome of a quantum instrument. The correlation function is then related to the probability of observing such outcome, i.e., $P(L,x|\mathcal{O}_\region{B})_{U_\region{A}\rhop U^\dagger_{\region{A}}}$.
    With this analogy in mind, we can think of the action of modifying the network as applying a quantum instrument, associated with an observable outcome, in the region $\region{B}$ and observing whether the probability of such outcome is affected by local modifications in $\region{A}$. 
\end{itemize}

Thus, by putting together the above points we find that we can understand the quantity $M_{\textup{CHQY}}$ as
\begin{equation}
    M_{\textup{CHQY}}(U_\region{A}:O_\region{B})= \frac{\bra{L}(U_\region{A}\otimes O_\region{B})\rhop (U_\region{A}^\dagger\otimes O_\region{B}^\dagger)\ket{L}}{\bra{L}(U_\region{A}\otimes \id_\region{B})\rhop (U_\region{A}^\dagger\otimes \id_\region{B})\ket{L}}=\frac{P(L,\bar{x}|\mathcal{O}_{\region{B}})_{U_\region{A}\rhop U^\dagger_{\region{A}}}}{P(L)_{U_\region{A}\rhop U^\dagger_{\region{A}}}},
\end{equation}
for an arbitrary set $\mathcal{O}_{\region{B}}=\{O^x_{\region{B}}\}_{x\in\mathcal{X}}$ such that $O_\region{B}=O_\region{B}^{\bar{x}}\in \mathcal{O}_\region{B}$ for some $\bar{x}\in\mathcal{X}$ provided in~\cref{def: causal influence old}.
This definition should have the goal of isolating the effect of $U_{\region{A}}$ on the outcome of the measurement in $\region{B}$ from its effect on the outcome of the tensor network contraction, $L$. However, this effect, which would be naturally portrayed through conditional probabilities, cannot be justified operationally. Indeed, in general, $\sum_xP(L,x|\mathcal{O}_{\region{B}})_{U_\region{A}\rhop U^\dagger_{\region{A}}} \neq P(L)_{U_\region{A}\rhop U^\dagger_{\region{A}}}$, thus the quantity $ M_{\textup{CHQY}}$ is not a probability distribution and ~\cref{def: causal influence old} does not have a straightforward operational interpretation in terms of probabilities.

In addition,~\cref{def: causal influence old} does not account for the possibility that the denominator is zero, in which case $M_{\textup{CHQY}}(U_\region{A}:O_\region{B})$ is ill-defined.

\subsection{Defining operational causal influence in tensor networks}
\label{sec:newCIdef}
Let us now provide a different definition of causal influence in a tensor network and show its operational meaning. 

\begin{definition}[Operational causal influence in a tensor network]
\label{def: causal influence}
Consider a tensor network $\tn = \left(\graphnameud,\{\hilmaparg{\edgename}\}_{\edgename\in\edgesetud},\{\projarg{\vertname}\}_{\vertname\in\vertsetud}\right)$ and an ordered pair of disjoint subsets of edges, i.e., $(\region{A},\region{B})$ such that $\region{A},\region{B}\subset\edgesetud$ and $\region{A}\cap\region{B}=\emptyset$, which we call \textup{regions}. Let $U_\region{A}\in\linops(\hilmaparg{\region{A}})$ be a unitary and $\mathcal{O}_\region{B}=\{O^y_\region{B}\in\linops(\hilmaparg{\region{B}})\}_{y\in\mathcal{Y}}$ a quantum instrument with Hermitian Kraus operators $O^y_\region{B}$. We define the correlation function
    \begin{equation}
        M(U_\region{A}:O^y_\region{B}|\mathcal{O}_\region{B}) = \frac{\bra{L}(U_\region{A}\otimes O^y_\region{B})\rhop (U_\region{A}^\dagger\otimes {O^y_\region{B}}^\dagger)\ket{L}}{\sum_{y}\bra{L}(U_\region{A}\otimes O^y_\region{B})\rhop (U_\region{A}^\dagger\otimes  {O^y_\region{B}}^\dagger)\ket{L}},
    \end{equation}
    if $\sum_{y}\bra{L}(U_\region{A}\otimes O^y_\region{B})\rhop (U_\region{A}^\dagger\otimes  {O^y_\region{B}}^\dagger)\ket{L}\neq 0$. Otherwise, we say that the pair $U_\region{A}$ and $\mathcal{O}_\region{B}$ are inconsistent relative to the tensor network.
    The operational quantum causal influence from $\region{A}$ to $\region{B}$ is zero if and only if 
    \begin{equation}
        M(U_\region{A}:O^y_\region{B}|\mathcal{O}_\region{B}) = M(U'_\region{A}:O^y_\region{B}|\mathcal{O}_\region{B})
    \end{equation}
    for all unitaries $U_\region{A}$ $U_\region{A}'$ and quantum instruments $\mathcal{O}_\region{B}=\{O_\region{B}^y\}_y$ such that $(U_\region{A},\mathcal{O}_\region{B})$ and $(U_\region{A}',\mathcal{O}_\region{B})$ are both consistent pairs.
\end{definition}
Note that above definition of causal influence depends not just on the particular instrument element $\mathcal{O}_\region{B}^y$ but on the the entire instrument $\mathcal{O}_\region{B}=\{\mathcal{O}_\region{B}^y\}_y$, which we have made explicit by introducing a conditioning on $\mathcal{O}_\region{B}$ in the notation for the quantity $M$.
Furthermore, the above definition also takes care of situations where the denominator is zero, thus providing a well-defined notion of the quantity $M$ and causal influence. In addition, we can interpret $M$ as
\begin{equation}
\label{eq:virtualmmt}
\begin{split}
     M(U_\region{A}:O^y_\region{B}|\mathcal{O}_\region{B}) &= \frac{\bra{L}(U_\region{A}\otimes O^y_\region{B})\rhop (U_\region{A}^\dagger\otimes {O^y_\region{B}}^\dagger)\ket{L}}{\sum_{y}\bra{L}(U_\region{A}\otimes O^y_\region{B})\rhop (U_\region{A}^\dagger\otimes  {O^y_\region{B}}^\dagger)\ket{L}}\\&=\frac{P(L,y|\mathcal{O}_{\region{B}})_{U_\region{A}\rhop U^\dagger_{\region{A}}}}{\sum_{y}P(L,y|\mathcal{O}_{\region{B}})_{U_\region{A}\rhop U^\dagger_{\region{A}}}}
     = P(y|L,\mathcal{O}_{\region{B}})_{U_\region{A}\rhop U^\dagger_{\region{A}}},
\end{split}
\end{equation}
by following the same interpretation of correlation functions as probabilities.
Conditional probabilities are precisely the object that allows us to isolate the contribution to the observable in $\region{B}$ from the contribution to the whole tensor network contraction. Indeed, it evaluates the distribution in $\region{B}$ for fixed $L$. We refer to~\cref{def: causal influence} as ``operational causal influence'', however, when clear from context we may simply call this measure ``causal influence''.

Going back to the open questions presented in~\Cref{sec:open questions ci}, we can now lift the doubt on whether one can define a notion of causal influence in tensor networks which genuinely portrays causation and not solely correlations. Indeed,~\cref{def: causal influence} provides an asymmetric correlation function where the observed distribution associated to the affected region $\region{B}$ is evaluated for different choices of operations in $\region{A}$. This is analogous to definitions of signalling in causal models, as we will further investigate in the next section.
\begin{remark}
    Generally, the two definitions of causal influence differ from each other. Indeed, while one only depends on one Kraus operator of a quantum instrument, the other depends on the whole set of operators. Thus, the two differ when different completions of the Kraus operator affect the network contraction in the denominator. The setups where the definition coincide can be studied using linearity results valid for causal models~\cite{Quantum_paper} and the mappings of the previous section. For instance, we expect the two notions of causal influence to coincide for tensor networks whose image causal model through the mappings of~\Cref{sec: mappings} is a process matrix~\cite{Barrett_2021}. Indeed, these are known to be a linear subset of generally cyclic causal models, which implies, thanks to results of~\cite{Quantum_paper}, that the denominators appearing in the two definitions of causal influence are constant and equal.
\end{remark}

\subsection{Bridging signalling and operational causal influence}
\label{sec:signalling and CI}
In this section, we show that the mappings of~\Cref{sec: mappings} link operational causal influence in a tensor network~(\cref{def: causal influence}) to signalling with unitaries in causal model as in~\cref{def: signalling}.
Firstly, let us define a time-reversible intervention on a given causal model.

\begin{definition}[Time-reversible intervention]\label{def:timerevint} Let $\cm=(\graphnamedir,\{\hilmaparg{\edgename}\}_{\edgename\in\edgesetdir},\{\mathcal{T}^\vertname\}_{\vertname\in\vertsetdir})$ be a causal model and consider an ordered pair of labs $(\labarg{A},\labarg{B})$, i.e., disjoint subsets of edges such that $\labarg{A}, \labarg{B}\subseteq \edgesetdir$ and $\labarg{A}\cap\labarg{B}=\emptyset$.
We call \textup{time-reversible intervention} an intervention defined by:
    \begin{itemize}
        \item the set of all unitary channels on $\labarg{A}$ labelled by a continuous variable $a\in\outcomemaparg{A}$ associated to the lab $\labarg{A}$, i.e.,
        \begin{equation}
            \{\mathcal{U}_{\labarg{A}}^{a}: \linops\left(\hilmaparg{\labarg{A}}\right)\mapsto \linops\left(\hilmaparg{\labarg{A}}\right)\}_{a\in \outcomemaparg{A}},
        \end{equation} 
        such that for all $\rho\in\linops\left(\hilmaparg{\labarg{A}}\right)$ and $a\in\outcomemaparg{A}$
        \begin{equation}
            \mathcal{U}_{\labarg{A}}^{a}(\rho)=U^a\rho {U^a}^\dagger
        \end{equation}
       where $U^a\in\linops(\hilmaparg{\labarg{A}})$ is unitary;
        \item the set of all Hermitian quantum instruments with discrete outcome $y\in\outcomemaparg{Y}$, labelled by a continuous variable $b\in\outcomemaparg{B}$, associated to the lab $\labarg{B}$, i.e.,
        \begin{equation}
            \left\{\left\{\instrmapsetting{\labarg{B}}{y|b}: \linops\left(\hilmaparg{\labarg{B}}\right)\mapsto \linops\left(\hilmaparg{\labarg{B}}\right)\right\}_{y\in\outcomemaparg{Y}}\right\}_{b\in \outcomemaparg{B}}
        \end{equation}
        such that for all $\rho\in\linops(\hilmaparg{\labarg{B}})$, $b\in\outcomemaparg{B}$ and $y\in\outcomemaparg{Y}$
        \begin{equation}
            \mathcal{M}^{y|b}_{\labarg{B}}(\rho_\labarg{B}) = O^{y|b}_\labarg{B}\rho_{\labarg{B}} {O^{y|b}_{\labarg{B}}}^\dagger,  
        \end{equation}
        where $O^{y|b}_\labarg{B}\in\linops(\hilmaparg{\labarg{B}})$ are positive, Hermitian and $\sum_y {O^{y|b}_{\labarg{B}}}^\dagger O^{y|b}_\labarg{B} = \id_{\labarg{B}}$.
    \end{itemize}
\end{definition}
The interventions that we study are time-reversible since the set of unitaries is closed under reversing inputs and outputs and Hermitian operators are invariant under such transformation. The definition is motivated by causal influence in tensor networks, where these operations are considered as they do not introduce an artificial time direction. While isometries are more general time-reversible operations, these are necessarily unitaries because, by construction of interventions, we only consider operations with isomorphic input and output spaces.

Let us now show the connection between signalling and operational causal influence.

\begin{restatable}[Connecting signalling and operational causal influence]{theorem}{ciassignalling}
\label{lem:signalling_CI}
Let us consider a causal model on $\graphnamedir$ and a tensor network on $\graphnameud$, $\cm$ and $\tn$, such that $\mapCMtoTN(\cm)=\tn$ or $\mapTNtoCMgen_D(\tn)=\cm$ for some $D\in\{0,1\}^{|\edgesetud|}$. Let $(\region{A},\region{B})$ be an ordered pair of regions in $\graphnameud$, i.e., disjoint subsets of $\edgesetdir$, and $(\labarg{A},\labarg{B})$ be the corresponding ordered pair of labs (cf.~\cref{def:intervention}) in the directed graph of $\cm$. Then, $\labarg{A}$ signals to $\labarg{B}$ in $\cm$ through a time-reversible intervention if and only if the operational causal influence from $\region{A}$ to $\region{B}$ in $\tn$ is non-zero.
\end{restatable}
\begin{proof}
    See~\Cref{app:proofs}.
\end{proof}

With the above theorem, we have proven the central result of the mapping, which allows to interchangeably work between operational causal influence and signalling. Because of this connection, we will often use the terms of labs and regions interchangeably when studying causal influence through an image causal model or signalling through the image tensor network. 
\begin{remark}
   \label{rem:virtualmmt}
   Thanks to the above connection between signalling in causal models and operational causal influence in tensor networks, we can translate the analysis of~\Cref{sec: new ci} (relating to the ``virtual'' measurement) in terms of measurements in acyclic casual models. Specifically, we recall~\cref{eq:virtualmmt}:
   \begin{equation}
    \begin{split}
         M(U_\region{A}:O^y_\region{B}|\mathcal{O}_\region{B}) &= \frac{\bra{L}(U_\region{A}\otimes O^y_\region{B})\rhop (U_\region{A}^\dagger\otimes {O^y_\region{B}}^\dagger)\ket{L}}{\sum_{y}\bra{L}(U_\region{A}\otimes O^y_\region{B})\rhop (U_\region{A}^\dagger\otimes  {O^y_\region{B}}^\dagger)\ket{L}}\\&=\frac{P(L,y|\mathcal{O}_{\region{B}})_{U_\region{A}\rhop U^\dagger_{\region{A}}}}{\sum_{y}P(L,y|\mathcal{O}_{\region{B}})_{U_\region{A}\rhop U^\dagger_{\region{A}}}}
         = P(y|L,\mathcal{O}_{\region{B}})_{U_\region{A}\rhop U^\dagger_{\region{A}}},
    \end{split}
    \end{equation}
    which hints at connections between operational causal influence and conditional probabilities linked to a ``virtual measurement'' yielding outcome $L$. In~\Cref{app:cyctoacyc}, we make this connection precise using our mapping from tensor networks to cyclic causal models, and the results of~\cite{Quantum_paper} which showed that for any possibly cyclic causal model, one can define a family of acyclic causal models with post-selection which share the distribution of the original model up to conditioning on the post-selections being successful. Since $M(U_\region{A}:O^y_\region{B}|\mathcal{O}_\region{B})$ equals the observed probability on the outcome $y$ of the cyclic causal model $\cm=\mapTNtoCMgen_D(\tn)$ (see proof of \cref{lem:signalling_CI} in \Cref{app:proofs} for details), this quantity can be further identified with probabilities conditioned on successful post-selection in the family of acyclic casual models associated to $\cm$. 
    Hence, we have: 
    \begin{equation}
   \begin{split}
     P(y|L,\mathcal{O}_{\region{B}})_{U_\region{A}\rhop U^\dagger_{\region{A}}}&=M(U_\region{A}:O^{y}_\region{B}|\mathcal{O}_\region{B}) \\
     &=\prob(y)_I = \probacyc\left(y \middle| \{\postoutcome_i = \ok\}_{\postvertname_i\in\psvertset}\right)_{\graphtele},
\end{split}
\end{equation}
where we denoted with $\graphtele$ the causal graph underlying one of such acyclic causal models and with $\postvertname_i\in\psvertset$ the set of vertices on which we post-select, and with $t_i=\checkmark$ the instance corresponding to successful post-selection. This allows us to better understand the role of the conditioning on this virtual measurement.

Since such property holds for a whole family of acyclic causal models with post-selection, the quantity $M$ has a robust operational interpretation in terms of conditional probabilities across various related acyclic causal models. For more details, we refer to~\Cref{app:cyctoacyc}.
\end{remark}

\begin{remark}
    In the causal modelling literature, one often refers to the connectivity of the graph and causal model when talking about causal influence. The connectivity does not always coincide with signalling relations, because of \textit{fine-tuning}, for instance in the case of causal mechanisms which ``ignore'' some inputs. In this work, whenever we mention ``causal influence'', it is in relation to tensor networks and we have shown that this notion does not correspond to what is meant in literature for ``causal influence'' in a causal model, rather, as we have shown in this chapter, it is related to the notion of signalling in a causal model.
\end{remark}

\subsection{An application: discrete rotations of causal models}
\label{sec:rotCM}

The results of the above section have an interesting consequence: a single tensor network can generate multiple, genuinely distinct causal models by choosing different direction assignments. This phenomenon can be viewed as a ``discrete space-time rotation'' between causal models. For instance, one can consider a unitary associated to a vertex of a causal model $\cm_1$
\begin{equation}
    \sum_{ijkl} U_{ij}^{kl}\;\ket{l}_D\ket{k}_C\bra{j}_B\bra{i}_A=\centertikz{
        \node[chan, minimum height = 25pt, minimum width = 35pt] (U) at (0,0) {$\mathcal{U}$};
        \draw[qleg] (U.125) --++ (0,0.5) node[above] {\indexstyle{C}};
        \draw[qleg] (U.55) --++ (0,0.5) node[above] {\indexstyle{D}};
        \draw[qleg] ($(U.235)-(0,0.5)$)node[below] {\indexstyle{A}} -- (U.235);
        \draw[qleg] ($(U.305)-(0,0.5)$)node[below] {\indexstyle{B}} -- (U.305);
        \begin{scope}[shift={(2,-0.5)}]
            \draw[-stealth] (0.2,0) --  (0.2,1) node[above] {$t$};
            \draw[-stealth] (0,0.2) --(1,0.2)  node[right] {$x$};
        \end{scope}
        
    }.
\end{equation}
Although a priori no background space-time is assumed, the input-output directions of the causal model allow us to infer one. Here, the subsystems $A$ and $B$ can be seen as ``space-like'' separated, and similarly for $C$ and $D$, while the inferred direction of time flows from $A$ and $B$ to $C$ and $D$. Hence, we obtain the space-time directions drawn above.

However, one can first map $\cm_1$ to a tensor network $\tn$, then map the latter to a different causal model $\cm_2$ where, for instance, the edges labelled with $B$ and $C$ have the opposite direction than in $\cm_1$, i.e., the map associated to $\cm_2$ in the new model is
\begin{equation}
    \sum_{ijkl} U_{ij}^{kl}\;\ket{l}_D\bra{k}_C\ket{j}_B\bra{i}_A = \centertikz{
        \node[chan, minimum height = 25pt, minimum width = 35pt] (U) at (0,0) {$\mathcal{Q}$};
        \draw[qleg] ($(U.125)+(-0.25,0.5)$)node[above] {\indexstyle{C}}--(U.125);
        
        \draw[qleg] (U.55) --++ (0.25,0.5) node[above] {\indexstyle{D}};
        \draw[qleg] ($(U.235)-(0.25,0.5)$)node[below] {\indexstyle{A}} -- (U.235);
        \draw[qleg] (U.305) --++(0.25,-0.5)node[below] {\indexstyle{B}};
        
    } \quad \textup{ or } \quad \centertikz{
        \node[chan, minimum height = 25pt, minimum width = 35pt] (U) at (0,0) {$\tilde{\mathcal{Q}}$};
        \draw[qleg] ($(U.125)+(-0.25,0.5)$)node[above] {\indexstyle{C}}--(U.125);
        
        \draw[qleg] (U.55) --++ (0.25,0.5) node[above] {\indexstyle{D}};
        \draw[qleg] ($(U.235)-(0.25,0.5)$)node[below] {\indexstyle{A}} -- (U.235);
        \draw[qleg] (U.305) --++(0.25,-0.5)node[below] {\indexstyle{B}};
        \draw[qleg] (U.340) .. controls ($(U.340)+(0.5,-0.25)$) and ($(U.20)+(0.5,0.25)$) .. (U.20);

\begin{scope}[shift={(2,0.5)},rotate=-90]
            \draw[-stealth] (0.2,0) --  (0.2,1) node[above] {$t$};
            \draw[-stealth] (0,0.2) --(1,0.2)  node[right] {$x$};
        \end{scope}
        
    },
\end{equation}
where the self-loop in the second diagram might be necessary in case the new choice of directions requires the use of the generalised mapping (see~\cref{def: TNtoCMgeneral}). Following the same arguments used to infer a direction of space-time in $\cm_1$, here we find that these are ``rotated'' in $\cm_2$. A special case of space-time ``rotations'' of maps which has already been studied are so-called \textit{dual} unitaries~\cite{Piroli_2020}. These are unitaries $U_{CD|AB} = \sum_{ijkl} U_{ij}^{kl}\;\ket{l}_D\ket{k}_C\bra{j}_B\bra{i}_A$ such that also their ``rotated'' version $\tilde{U}_{BD|AC} = \sum_{ijkl} U_{ij}^{kl}\;\ket{l}_D\bra{k}_C\ket{j}_B\bra{i}_A$ is unitary. Here, we extend the concept of ``rotations'' of dual unitaries to ``rotations'' of arbitrary CPTP maps, potentially at the cost of adding self-cycles. 

Thus, from considering a set of causal models which are ``rotations'' of each other, one can infer an emerging direction of space and time for each model. Further principles, such as acyclicity, or preservation of relativistic causality (which excludes retro-causality), can be evoked to determine a preferred emerging direction of space-time. For instance, in the example above where $\cm_1$ and $\cm_2$ both arise from the same tensor network, in case the self loop is present in $\cm_2$ one would favour the space-time emergent from $\cm_1$, as this preserves relativistic causality.

The following corollary, which follows immediately from~\cref{lem:signalling_CI}, further establishes a connection, at the level of signalling between causal models which arise from the same tensor network, which can then be understood as an equivalence class of ``rotations'' of a causal model.
\begin{corollary}
    Consider a tensor network $\tn$ on a undirected graph $\graphnameud$, and let $\cm_D = \mapTNtoCMgen_D(\tn)$ for $D\in\{0,1\}^{|E|}$. Let $(\region{A},\region{B})$ be an ordered pair of regions in $\graphnameud$, and $(\labarg{A}_D,\labarg{B}_D)$ be the corresponding ordered pair of labs in the directed graph of $\cm_D$. Then, $\cm_D$ present the same time-reversible signalling relations for all $D$, i.e., $\labarg{A}_D$ time-reversibly signals to $\labarg{B}_D$ in $\cm_D$ if and only if $\labarg{A}_{D'}$ time-reversibly signals to $\labarg{B}_{D'}$ in $\cm_{D'}$.
\end{corollary}

Therefore, one can define equivalence classes of causal models --- those arising from the same tensor network --- with identical operational constraints. Indeed, although these causal models have different causal graphs, they nevertheless share exactly the same signalling and non-signalling relations. Notably, the observed non-signalling relations in each model may arise from different sources, for instance from $d$- or $p$-separation~\cite{pearl_2009,Quantum_paper} (see~\Cref{app:graphsep}) or from fine-tuning, depending on the underlying causal structure. This might connect fine-tuned relations to sound graph separation properties like $d$- or $p$-separation in a different model of the equivalence class.

\begin{remark}
    The statements about signalling relations are restricted to the set of time-reversible interventions~(\cref{def:timerevint}). Since these only involve unitaries and Hermitian operators, they remain well-defined independently of the direction of the edge they are defined to act on. Signalling with arbitrary interventions can be defined for a fixed causal model, where such interventions are valid quantum operations, to a tensor network, but it cannot be transferred to a different causal model because inverting the direction of the edges we intervene upon might lead to a non-CPTP map. 

However, one can transfer signalling relations in a causal model, arising from arbitrary maps as interventions on the labs $\labarg{A}$ and $\labarg{B}$, to the same signalling relations in any other model of the equivalence class which has the same directions for the edges $\labarg{A}$ and $\labarg{B}$. Explicitly, given $\cm_D$ and $\cm_{D'}$, obtained through $\mapTNtoCMgen_D$ and $\mapTNtoCMgen_{D'}$ from the same tensor network $\tn$, signalling from $\labarg{A}_D$ to $\labarg{B}_{D}$ with arbitrary interventions implies the same signalling from $\labarg{A}_{D'}$ to $\labarg{B}_{D'}$ if $\labarg{A}_D=\labarg{A}_{D'}$ and $\labarg{B}_D=\labarg{B}_{D'}$, i.e., the sets of edges are the same including their direction. Such causal models are one the ``rotation'' of the other, but can be considered locally equal, sine the direction of edges involved in the interventions are the same. Nonetheless, they share the same signalling structure without restrictions on which maps one inserts in the labs.
\end{remark}

\section{Application to holography}
\label{sec:holography}

In this section, we apply the mapping to Holographic tensor networks, showing that properties of the graph underlying the image causal model straightforwardly imply the absence of causal influence in the tensor network. 

\subsection{The holographic tensor network}
In AdS-CFT~\cite{maldacena-1998,witten-1998}, a duality is conjectured between a $(d+1)$-dimensional quantum gravity theory in Anti de-Sitter (AdS) space, called the \textit{bulk theory}, and a $d$-dimensional conformal field theory (CFT) which lives on a space isomorphic to the conformal boundary of AdS, called the \textit{boundary theory}. This correspondence implies a connection between degrees of freedom in the bulk and the boundary. Toy models for AdS-CFT construct quantum error correction codes where the degrees of freedom in the bulk are non-locally encoded in the boundary theory~\cite{almheiri-2014}. These are known as \textit{holographic tensor networks}~\cite{Pastawski_2015,hayden-2016}. 

\paragraph{Perfect tensors and networks}
The building block of holographic tensor networks are so-called \textit{perfect tensors}.

\begin{definition}[Perfect tensor~\cite{Pastawski_2015}]
\label{def:perfect tensor}
Consider a tensor of rank $2n$ with indices $\mathcal{I}=\{1,\dots, 2n\}$ where the $j$-th index takes values $i_j \in \outcomemaparg{j}$, i.e., 
\begin{equation}
    T  = \{T_{i_1,\dots,i_{2n}}\in\mathbb{C}\}_{i_j\in\outcomemaparg{j}, \; j=1,\dots,2n},
\end{equation} 
and define Hilbert spaces associated to each index $j\in\mathcal{I}$ as $\hilmaparg{j}= \textup{span}\{\ket{i_j}_j\}_{i_j\in\outcomemaparg{j}}$. 

The tensor $T$ is a \textup{perfect tensor} if any bipartition of indices into sets $S$ and $\bar{S}$ such that $|S|\leq|\bar{S}|$ defines an isometry from $\hilmaparg{S}\equiv\bigotimes_{j\in S} \hilmaparg{j}$ to $\hilmaparg{\bar{S}}\equiv\bigotimes_{k\in \bar{S}} \hilmaparg{k}$ up to a constant. Explicitly, for all $S,\bar{S}\subseteq\mathcal{I}$, such that $S\cup\bar{S}=\mathcal{I}$, $S\cap \bar{S}=\emptyset$ and $|S|\leq|\bar{S}|$, the linear operator 
\begin{equation}
    V= C\sum_{i_1,\dots, i_{2n}}T_{i_1,\dots,i_{2n}} \bigotimes_{j\in \bar{S}} \ket{i_{j}}_{j}\bigotimes_{k\in S} \bra{i_k}_{k} \equiv C\sum_{i_S,i_{\bar{S}}} T_{i_S,i_{\bar{S}}} \ket{i_{\bar{S}}}_{\bar{S}}\bra{i_S}_{S}
\end{equation}
    satisfies $V^\dagger V = \id_{S}$, for some constant $C$. Notice that we have introduced a compact notation, where $i_S \equiv \{i_j\}_{j\in S}$ and $\ket{\psi}_S\in \hilmaparg{S}$ and similarly for $\bar{S}$.
\end{definition}
For example, consider $n=3$, $\mathcal{I}=\{1,2,3,4,5,6\}$ and the tensor $T_{i_1,\dots,i_6}$, represented as a vector~(see~\cref{rem:densityop})
\begin{equation}\label{eq: perfect t}
    \centertikz{
        \def\lato{0.6cm}
            \foreach \ang in {0,60,...,300}
                \draw[mybaseline] (0,0) -- ++(\ang+90:\lato*1.3);
        \node[mybaseline, regular polygon, draw, regular polygon sides = 6, minimum size = 1.cm, fill=yellow] (p) at (0,0) {};
    }\equiv \ketbra{\Psi}_{\mathcal{I}}, \textup{ where }
    \ket{\Psi}_{\mathcal{I}} =  \sum_{i_1,\dots,i_6} T_{i_1,\dots,i_6} \ket{i_1}_{1}\otimes\dots\otimes\ket{i_6}_{6}.
\end{equation}

 Let us consider the bipartition $S = \{5, 6\}$ and $\bar{S} = \{1, 2, 3, 4\}$, the perfect tensor property implies that the linear map
\begin{equation}
    V = C\sum_{i_1,\dots,i_6} T_{i_1,\dots,i_6} \ket{i_1}_{1}\otimes\ket{i_2}_{2}\otimes\ket{i_3}_{3}\otimes\ket{i_4}_{4}\otimes\bra{i_5}_{5}\otimes\bra{i_6}_{6}
\end{equation}
satisfies $V^\dagger V = \id_{S}$ for some constant $C$. Therefore, it is an isometry and defines an isometric channel $\mathcal{V}_{\bar{S}|S}(\rho_S) = V\rho_S V^\dagger$.

\begin{restatable}[Isometric channel of a perfect tensor]{lemma}{perfecttnch}
\label{lem:isomchan perfect tn}
Given a perfect tensor $T$ of rank $2n$, to which we associate the state
\begin{equation}
    \ket{\Psi}_{\mathcal{I}} = \sum_{i_1,\dots,i_{2n}} T_{i_1,\dots,i_{2n}} \ket{i_1}_{1}\otimes\dots\otimes\ket{i_{2n}}_{2n},
\end{equation}
and a bipartition of its indices into $S$ and $\bar{S}$ such that $|S|\leq|\bar{S}|$, the corresponding isometric channel, up to a constant $C$, is given by
\begin{equation}
    C\CJin_{\bar{S}|S}(\ketbra{\Psi}_{S\bar{S}}) (\cdot) = V(\cdot)V^\dagger.
\end{equation}
\end{restatable}
\begin{proof}
    See~\Cref{app:proofs}.
\end{proof}

Given the perfect tensor which we considered before with $n=3$, we can use~\cref{def:perfect tensor} to define isometric channels for any bipartition with the correct cardinality, for instance:
\begin{equation}
\begin{split}
     \bar{S} = 1,2,3,4,5,\; S =6;\; & V_{|S|=1} = C_1\sum_{i_1,\dots, i_6} T_{i_1,\dots, i_6} \ket{i_1}_1\otimes \ket{i_2}_3\otimes\ket{i_3}_3\otimes\ket{i_4}_4\otimes \ket{i_5}_5\otimes\bra{i_6}_6\\
    \bar{S} = 1,2,3,4,\; S =5,6;\; & V_{|S|=2} = C_2\sum_{i_1,\dots, i_6} T_{i_1,\dots, i_6} \ket{i_1}_1\otimes \ket{i_2}_3\otimes\ket{i_3}_3\otimes\ket{i_4}_4\otimes \bra{i_5}_5\otimes\bra{i_6}_6\\
    \bar{S} = 1,2,3,\; S =4,5,6;\; & V_{|S|=3} = C_3\sum_{i_1,\dots, i_6} T_{i_1,\dots, i_6} \ket{i_1}_1\otimes \ket{i_2}_3\otimes\ket{i_3}_3\otimes\bra{i_4}_4\otimes \bra{i_5}_5\otimes\bra{i_6}_6
\end{split}   
\end{equation}
and so on. 

In what follows, we will consider perfect tensor networks, which are obtained from contracting perfect tensors. For instance, 
\begin{equation}\label{eq: Hol state}
    H = \centertikz{
       \foreach \ang in {0,60, ...,300}{
            \draw[mybaseline] (\ang+90:0.6cm*3) -- (0,0);
            \draw[mybaseline] (\ang+90:0.6cm*3) to [in = \ang+210, out = \ang+330] (\ang+60:0.6cm*3.5);
            \draw[mybaseline] (\ang+60+60:0.6cm*3.5) to [in = \ang+210, out = \ang+330]  (\ang+90:0.6cm*3); 
            \foreach \an in {0,60,120}{
                    \draw[mybaseline] (\ang+90:0.6cm*3) -- ++(\an+30+\ang:0.5);
                }
            \foreach \an in {0,60,120,180}{
                    \draw[mybaseline] (\ang+60:0.6cm*3.5) -- ++(\ang+\an-30:0.4);
                }
           } 
        \foreach \ang in {0,60, ...,300}{
           \node[mybaseline, regular polygon, draw, regular polygon sides = 6, minimum size = 0.8cm, fill=yellow!40, rotate = \ang] (p1) at (\ang+90:0.6cm*3) {};
            \node[mybaseline, regular polygon, draw, regular polygon sides = 6, minimum size = 0.6cm, fill=yellow!40, rotate = \ang] (p2) at (\ang+60:0.6cm*3.5) {};
         }   
        \node[mybaseline, regular polygon, draw, regular polygon sides = 6, minimum size = 1.2cm, fill=yellow!40] (p) at (0,0) {};
        \node[mybaseline, regular polygon, draw, regular polygon sides = 6, minimum size = 1.2cm, fill=yellow!40] (p) at (0,0) {};
        \node (i) at ($(0,0)!0.5!(30:0.6cm*3)$) {};
        \node (i) at ($(0,0)!0.5!(150:0.6cm*3)$){};
    } \equiv \ketbra{\Xi}
\end{equation}
is obtained through contracting perfect tensors with $2n=6$ indices, like the one in~\cref{eq: perfect t}. Notice that the above tensor network is not fully contracted, as the tensors on the boundary have open legs representing uncontracted indices. Therefore, the network is itself associated with a tensor with as many indices as the open edges. In the above example, we would have  $H=\{H_{i_1,\dots,i_{24},j_1,\dots,j_{18}}\}$, where the indices $i_l$ are associated to the open legs of the third layer of perfect tensors (four for each tensor) and the indices $j_m$ to the open legs of the second layer of perfect tensors (three for each tensor).
Each perfect tensor is associated with a pure state $\ketbra{\Psi}$, thus the whole network itself is also associated with a pure state on a larger Hilbert space, which we denote as $\ketbra{\Xi}$. We often call the \textit{bulk} of the network the contracted edges, and the \textit{boundary} the uncontracted edges in~\cref{eq: Hol state}.

In what follows, we are interested in studying operational causal influence of the network. For that purpose, one has to consider the fully contracted network that is obtained through contracting the network $H$ with its complex conjugate. 
In this case, one can prove that the network contraction equals
\begin{equation}
    \bra{L}\rhop\ket{L}=|\braket{\Xi}|^2.
\end{equation}
In what follows, we only draw diagrams representing the uncontracted network, but study properties of the fully contracted network. 
\paragraph{Causal influence in the Holographic tensor network}
Studying causal influence in the Holographic tensor network corresponds to choosing two regions as subsets of contracted edges of the perfect network and then considering the relevant correlation functions. We denote the state corresponding to a choice of unitary and hermitian instrument, $U_\region{A}$ and $O_{\region{B}}$, as:
\begin{equation}\label{eq: with ops}
    \ketbra{\Xi(U_\region{A}, O_\region{B})} \equiv \centertikz{
       \foreach \ang in {0,60, ...,300}{
            \draw[mybaseline] (\ang+90:0.6cm*3) -- (0,0);
            \draw[mybaseline] (\ang+90:0.6cm*3) to [in = \ang+210, out = \ang+330] (\ang+60:0.6cm*3.5);
            \draw[mybaseline] (\ang+60+60:0.6cm*3.5) to [in = \ang+210, out = \ang+330]  (\ang+90:0.6cm*3); 
            \foreach \an in {0,60,120}{
                    \draw[mybaseline] (\ang+90:0.6cm*3) -- ++(\an+30+\ang:0.5);
                }
            \foreach \an in {0,60,120,180}{
                    \draw[mybaseline] (\ang+60:0.6cm*3.5) -- ++(\ang+\an-30:0.4);
                }
           } 
        \foreach \ang in {0,60, ...,300}{
           \node[mybaseline, regular polygon, draw, regular polygon sides = 6, minimum size = 0.8cm, fill=yellow!40, rotate = \ang] (p1) at (\ang+90:0.6cm*3) {};
            \node[mybaseline, regular polygon, draw, regular polygon sides = 6, minimum size = 0.6cm, fill=yellow!40, rotate = \ang] (p2) at (\ang+60:0.6cm*3.5) {};
         }   
        \node[mybaseline, regular polygon, draw, regular polygon sides = 6, minimum size = 1.2cm, fill=yellow!40] (p) at (0,0) {};
        \node[tnode, fill = red] (U) at ($(0,0)!0.5!(30:0.6cm*3)$) {};
        \node [above = 0 of U] {\indexstyle{U_\region{A}}};
        \node[tnrect, fill = blue,minimum size = 5pt] (O1) at ($(0,0)!0.5!(150:0.6cm*3)$) {};
        \node [above = 0 of O1] {\indexstyle{O_\region{B}}};
    }.
\end{equation}
The un-normalised state $\ket{\Xi(U_\region{A}, O_\region{B})}$ satisfies:
\begin{equation}
    \bra{L}(U_\region{A}\otimes O_\region{B})\rhop(U^\dagger_\region{A}\otimes O^\dagger_\region{B})\ket{L}=|\braket{\Xi(U_\region{A}, O_\region{B})}|^2.
\end{equation}

If we are able to express the unitary acting in the bulk (a subset of contracted edges) of the network as a, possibly different, unitary acting at the boundary (a subset of uncontracted edges), i.e.,
\begin{equation}
    \ket{\Xi(U_\region{A}, O_\region{B})} = \tilde{U}_{\text{boundary}}\ket{\Xi(\id_\region{A}, O_\region{B})}\equiv \tilde{U}\ket{\tilde{\Xi}(O_\region{B})},
\end{equation}
the causal influence from $\region{A}$ to $\region{B}$ is trivially zero 

Indeed, if this is the case
\begin{equation}
    \begin{split}
        M(U_\region{A}:O_\region{B}|\mathcal{O}_{\region{B}}) &= \bra{L}(U_\region{A}\otimes O_\region{B})\rhop(U^\dagger_\region{A}\otimes O^\dagger_\region{B})\ket{L} \\
        &= |\braket{\Xi(U_\region{A}, O_\region{B})}|^2\\ &=|\bra{\tilde{\Xi}(O_\region{B})}\tilde{U}^\dagger\tilde{U}\ket{\tilde{\Xi}(O_\region{B})}|^2=
        |\braket{\tilde{\Xi}(O_\region{B})}|^2
    \end{split}
\end{equation}
for all unitaries $U_\region{A}$ and Hermitian instruments $\mathcal{O}_{\region{B}}$. 

In the context of tensor networks, this is achieved by ``pushing'' the unitary to the boundary.
Indeed, isometries $V$ arising from perfect tensors have the property that any operator $O$ acting on an appropriate subset of its edges can be replaced by an equal norm operator $O'$ acting on the remaining edges, because
\begin{equation}
    VO= VOV^\dagger V = (VOV^\dagger)V \equiv O'V.
\end{equation}
Thus, in the Holographic network, operators in the bulk can be pushed to the boundary. In~\cite{Pastawski_2015}, it is shown how one can push the operator using only the circuit contained in a so-called \textit{minimal geodesics}, which can be understood as a line from two points in the boundary bounding a region containing the operator we wish to push, and cutting the minimum amount of edges. For instance, in the tensor network on the left there is one bulk operator, denoted with the red circle. The corresponding geodesics is drawn with the dashed line. Then, according to~\cite{Pastawski_2015}, the operator can be pushed to the boundary as in the right network:
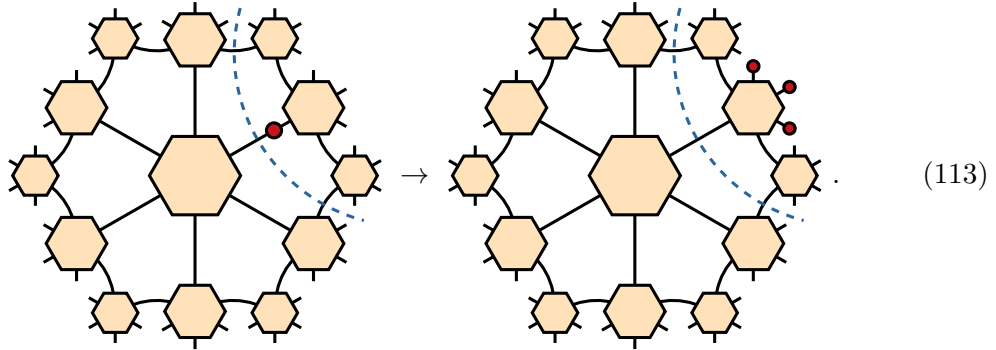
\begin{equation}
   \centertikz{
    \begin{scope}
         \foreach \ang in {0,60, ...,300}{
            \draw[mybaseline] (\ang+90:0.6cm*3) -- (0,0);
            \draw[mybaseline] (\ang+90:0.6cm*3) to [in = \ang+210, out = \ang+330] (\ang+60:0.6cm*3.5);
            \draw[mybaseline] (\ang+60+60:0.6cm*3.5) to [in = \ang+210, out = \ang+330]  (\ang+90:0.6cm*3); 
            \foreach \an in {0,60,120}{
                    \draw[mybaseline] (\ang+90:0.6cm*3) -- ++(\an+30+\ang:0.5);
                }
            \foreach \an in {0,60,120,180}{
                    \draw[mybaseline] (\ang+60:0.6cm*3.5) -- ++(\ang+\an-30:0.4);
                }
           } 
        \foreach \ang in {0,60, ...,300}{
           \node[mybaseline, regular polygon, draw, regular polygon sides = 6, minimum size = 0.8cm, fill=yellow!40, rotate = \ang] (p1) at (\ang+90:0.6cm*3) {};
            \node[mybaseline, regular polygon, draw, regular polygon sides = 6, minimum size = 0.6cm, fill=yellow!40, rotate = \ang] (p2) at (\ang+60:0.6cm*3.5) {};
         }   
        \node[mybaseline, regular polygon, draw, regular polygon sides = 6, minimum size = 1.2cm, fill=yellow!40] (p) at (0,0) {};
        \node[tnode, fill = red] (U) at ($(0,0)!0.5!(30:0.6cm*4)$) {};
        \draw[mybaseline,dashed,draw=blue] (75:2.3cm) to [in = 75+90, out = 345-90]  (345:2.3cm);
    \end{scope}
    }\to \centertikz{
    \begin{scope}
         \foreach \ang in {0,60, ...,300}{
            \draw[mybaseline] (\ang+90:0.6cm*3) -- (0,0);
            \draw[mybaseline] (\ang+90:0.6cm*3) to [in = \ang+210, out = \ang+330] (\ang+60:0.6cm*3.5);
            \draw[mybaseline] (\ang+60+60:0.6cm*3.5) to [in = \ang+210, out = \ang+330]  (\ang+90:0.6cm*3); 
            \foreach \an in {0,60,120}{
                    \draw[mybaseline] (\ang+90:0.6cm*3) -- ++(\an+30+\ang:0.5);
                }
            \foreach \an in {0,60,120,180}{
                    \draw[mybaseline] (\ang+60:0.6cm*3.5) -- ++(\ang+\an-30:0.4);
                }
           } 
        \foreach \ang in {0,60, ...,300}{
           \node[mybaseline, regular polygon, draw, regular polygon sides = 6, minimum size = 0.8cm, fill=yellow!40, rotate = \ang] (p1) at (\ang+90:0.6cm*3) {};
            \node[mybaseline, regular polygon, draw, regular polygon sides = 6, minimum size = 0.6cm, fill=yellow!40, rotate = \ang] (p2) at (\ang+60:0.6cm*3.5) {};
         }   
        \node[mybaseline, regular polygon, draw, regular polygon sides = 6, minimum size = 1.2cm, fill=yellow!40] (p) at (0,0) {};
        \foreach \an in {0,60,120}{
                    \node[mybaseline,inner sep=1.5pt, ellipse, draw=black, fill = red] (U) at ($(300+90:0.6cm*3)+(\an+30+300:0.55)$) {};
                }

        \foreach \an in {0,60,120,180}{
                }
            \draw[mybaseline,dashed,draw=blue] (75:2.3cm) to [in = 75+90, out = 345-90]  (345:2.3cm);    
    \end{scope} 
    }.
\end{equation}

However, in some instances the procedure to ``push'' the networks cannot be performed, thus leaving the question on whether there is causal influence unsolved. Consider the following example (reproduced from Figure 18 of~\cite{Cotler_2019}) where we denote with a red circle the region $\region{A}$ and with a blue square the region $\region{B}$:
\begin{equation}
    \ketbra{\Xi(U_\region{A}, O_\region{B})} = \centertikz{
       \foreach \ang in {0,60, ...,300}{
            \draw[mybaseline] (\ang+90:0.6cm*3) -- (0,0);
            \draw[mybaseline] (\ang+90:0.6cm*3) to [in = \ang+210, out = \ang+330] (\ang+60:0.6cm*3.5);
            \draw[mybaseline] (\ang+60+60:0.6cm*3.5) to [in = \ang+210, out = \ang+330]  (\ang+90:0.6cm*3); 
            \foreach \an in {0,60,120}{
                    \draw[mybaseline] (\ang+90:0.6cm*3) -- ++(\an+30+\ang:0.5);
                }
            \foreach \an in {0,60,120,180}{
                    \draw[mybaseline] (\ang+60:0.6cm*3.5) -- ++(\ang+\an-30:0.4);
                }
           } 
        \foreach \ang/\b in {0/1,60/2, 120/3, 180/4, 240/5, 300/6}{
           \node[mybaseline, regular polygon, draw, regular polygon sides = 6, minimum size = 0.8cm, fill=yellow!40, rotate = \ang] (p1) at (-\ang+90:0.6cm*3) {};
           \node at (-\ang+90:0.6cm*3) {};
            \node[mybaseline, regular polygon, draw, regular polygon sides = 6, minimum size = 0.6cm, fill=yellow!40, rotate = \ang] (p2) at (-\ang+60:0.6cm*3.5) {};
            \node at (-\ang+60:0.6cm*3.5) {};
         }   
         \begin{scope}[xscale=-1, yscale=-1]
            \draw[mybaseline,dashed,draw=blue] (75:2.3cm) to [in = 75+90, out = 345-90]  (345:2.3cm);
        \end{scope}
        \node[mybaseline, regular polygon, draw, regular polygon sides = 6, minimum size = 1.2cm, fill=yellow!40] (p) at (0,0) {};
        \node[tnode, fill = red] (U) at ($(0,0)!0.5!(-150:0.6cm*3)$) {};
        \node[tnrects, fill = blue,minimum size = 3pt] (O1) at ($(0,0)!0.5!(-30:0.6cm*3)$) {};
        \node[tnrects, fill = blue,minimum size = 3pt] (O3) at ($(0,0)!0.5!(30:0.6cm*3)$) {};
        \node[tnrects, fill = blue,minimum size = 3pt] (O4) at ($(0,0)!0.5!(270:0.6cm*3)$) {};
        \node[tnrects, fill = blue,minimum size = 3pt] (O5) at (230:0.6cm*2.9) {};
        \node[tnrects, fill = blue,minimum size = 3pt] (U) at ($(120+90:0.6cm*3)+(120+30+120:0.55)$) {};
        \node[tnrects, fill = blue,minimum size = 3pt] (U) at ($(120+90:0.6cm*3)+(60+30+120:0.55)$) {};

    }.
\end{equation}
Here, one cannot push the unitary operator, represented as a circle, through the boundary. Indeed, the presence of the tensors in $\region{B}$, represented as blue squares, inside the geodesics prevents us from using the perfect tensor property to push the $\region{A}$ operator to the boundary. Thus, from the geodesic criterion alone, one cannot conclude whether $\region{A}$ influences $\region{B}$ in this case.

The goal of the next sections will be to map the tensor network to a causal model, analyse signalling and use it to infer causal influence in the tensor network.
Since each hexagon is a perfect tensor, we can freely choose the input-output systems to apply~\cref{def: TNtoCMgeneral} and obtain a causal model, the only requirement being that our choice is consistent for the whole network and that the number of output systems is smaller than the number of input systems. However, the latter requirement, which exploits the perfect tensor property, is not necessary. In general, one can apply the generalised mapping of~\cref{def: TNtoCMgeneral} and still obtain a valid cyclic causal model when the input and output choice does not define an isometric channel but generally a CP.

This freedom allows us to make choices that simplify studying signalling, for instance with graph-separation theorems, and to consider models which do not respect the perfect tensor property, at the expense of possibly adding cycles.

\subsection{Mapping to causal models and application of graph-separation theorems}
\label{sec:graphsepHolo}
In the previous sections, we constructed a mapping from tensor networks to causal models which preserves causal influence as time-reversible signalling in the causal model. Here, we use the mapping to apply graph-separation results, proven for causal models, to argue for the absence of causal influence from one region to another.  Let us first provide an intuition behind graph separation results (for more details we refer to~\Cref{app:graphsep}).

\paragraph{Graph-separation theorems}
Graph separation properties, such as $d$-separation~\cite{pearl_2009}, are purely defined at the level of a graph, irrespective of causal models living on it. The $d$-separation property, for instance, says that two sets of vertices $V_1$ and $V_2$ are either $d$-separated or $d$-connected conditioned on a third set of vertices, $V_3$, and we denote it respectively as $(V_1\perp^dV_2|V_3)$ or $(V_1\not\perp^dV_2|V_3)$. 

For example, in the graph, $A\rightarrow C \leftarrow B$, called a \emph{collider graph}, $A$ and $B$ are $d$-separated but they become $d$-connected conditioned on $C$, hence $A\perp^dB$ and $A\not\perp^dB|C$. 
In acyclic graphs, the $d$-separation theorem states that: {\bf (Soundness)} if $(V_1\perp^dV_2|V_3)$, then the probability distribution over the outcomes of these vertices also presents the same conditional independence in \emph{any} causal model on the graph, and {\bf (Completeness)} if $(V_1\not\perp^dV_2|V_3)$, then \emph{there exists} a causal model on the graph whose distribution exhibits the corresponding conditional dependence. 

For cyclic graphs, this theorem typically fails~\cite{pearl_2009,Neal_2000}, but a different graph-separation property, $p$-separation, has been proven to be sound and complete for all finite-dimensional cyclic causal models~\cite{Quantum_paper}.
Thus, both for cyclic and acyclic causal models one can apply graph separation theorems to conclude conditional independencies, which are fundamental to determine no-signalling relation, solely from the structure of the graph.

\paragraph{Causal model of a Holographic tensor network}
We now use our mapping from tensor networks to causal models to immediately conclude that there is no causal influence from $\region{A}$ to $\region{B}$ in the following example, where we label the tensors for clarity:
\begin{equation}\label{eq: ex Hstate labels}
    \ketbra{\Xi(U_\region{A}, O_\region{B})} = \centertikz{
       \foreach \ang in {0,60, ...,300}{
            \draw[mybaseline] (\ang+90:0.6cm*3) -- (0,0);
            \draw[mybaseline] (\ang+90:0.6cm*3) to [in = \ang+210, out = \ang+330] (\ang+60:0.6cm*3.5);
            \draw[mybaseline] (\ang+60+60:0.6cm*3.5) to [in = \ang+210, out = \ang+330]  (\ang+90:0.6cm*3); 
            \foreach \an in {0,60,120}{
                    \draw[mybaseline] (\ang+90:0.6cm*3) -- ++(\an+30+\ang:0.5);
                }
            \foreach \an in {0,60,120,180}{
                    \draw[mybaseline] (\ang+60:0.6cm*3.5) -- ++(\ang+\an-30:0.4);
                }
           } 
        \foreach \ang/\b in {0/1,60/2, 120/3, 180/4, 240/5, 300/6}{
           \node[mybaseline, regular polygon, draw, regular polygon sides = 6, minimum size = 0.8cm, fill=yellow!40, rotate = \ang] (p1) at (-\ang+90:0.6cm*3) {};
           \node at (-\ang+90:0.6cm*3) {\indexstyle{b.\b}};
            \node[mybaseline, regular polygon, draw, regular polygon sides = 6, minimum size = 0.6cm, fill=yellow!40, rotate = \ang] (p2) at (-\ang+60:0.6cm*3.5) {};
            \node at (-\ang+60:0.6cm*3.5) {\indexstyle{c.\b}};
         }   
        \node[mybaseline, regular polygon, draw, regular polygon sides = 6, minimum size = 1.2cm, fill=yellow!40] (p) at (0,0) {\indexstyle{a}};
        \node[tnode, fill = red] (U) at ($(0,0)!0.5!(30:0.6cm*3)$) {};
        \node[tnrects, fill = blue,minimum size = 3pt] (O1) at ($(0,0)!0.5!(90:0.6cm*3)$) {};
        \node[tnrects, fill = blue,minimum size =3pt] (O2) at ($(0,0)!0.5!(150:0.6cm*3)$) {};
        \node[tnrects, fill = blue,minimum size = 3pt] (O3) at ($(0,0)!0.5!(210:0.6cm*3)$) {};
        \node[tnrects, fill = blue,minimum size = 3pt] (O5) at (50:0.6cm*2.9) {};
        \node[tnrects, fill = blue,minimum size = 3pt] (O6) at (10:0.6cm*2.9) {};
        \node[tnrects, fill = blue,minimum size = 3pt] (O4) at ($(0,0)!0.5!(270:0.6cm*3)$) {};
    }.
\end{equation}

Let us make the following choice of input and output for each tensor:

\begin{equation}\centertikz{
        \node[mybaseline, regular polygon, draw, regular polygon sides = 6, minimum size = 1.2cm, fill=yellow!40] (a) at (0,0) {\indexstyle{a}};
        \node[mybaseline, regular polygon, draw, regular polygon sides = 6, minimum size = 0.8cm, fill=yellow!40, rotate = 0] (b1) at (-0+90:0.6cm*3) {};
        \node at (-0+90:0.6cm*3) {\indexstyle{b.1}};
        \node[mybaseline, regular polygon, draw, regular polygon sides = 6, minimum size = 0.8cm, fill=yellow!40, rotate = 60] (b2) at (-60+90:0.6cm*3) {};
        \node at (-60+90:0.6cm*3) {\indexstyle{b.2}};
        \node[mybaseline, regular polygon, draw, regular polygon sides = 6, minimum size = 0.8cm, fill=yellow!40, rotate = 120] (b3) at (-120+90:0.6cm*3) {};
        \node at (-120+90:0.6cm*3) {\indexstyle{b.3}};
        \node[mybaseline, regular polygon, draw, regular polygon sides = 6, minimum size = 0.8cm, fill=yellow!40, rotate = 180] (b4) at (-180+90:0.6cm*3) {};
        \node at (-180+90:0.6cm*3) {\indexstyle{b.4}};
        \node[mybaseline, regular polygon, draw, regular polygon sides = 6, minimum size = 0.8cm, fill=yellow!40, rotate = 240] (b5) at (-240+90:0.6cm*3) {};
        \node at (-240+90:0.6cm*3) {\indexstyle{b.5}};
        \node[mybaseline, regular polygon, draw, regular polygon sides = 6, minimum size = 0.8cm, fill=yellow!40, rotate = 300] (b6) at (-300+90:0.6cm*3) {};
        \node at (-300+90:0.6cm*3) {\indexstyle{b.6}};
        \node[mybaseline, regular polygon, draw, regular polygon sides = 6, minimum size = 0.6cm, fill=yellow!40, rotate = 0] (c1) at (-0+60:0.6cm*3.5) {};
        \node at (-0+60:0.6cm*3.5) {\indexstyle{c.1}};
        \node[mybaseline, regular polygon, draw, regular polygon sides = 6, minimum size = 0.6cm, fill=yellow!40, rotate = 60] (c2) at (-60+60:0.6cm*3.5) {};
        \node at (-60+60:0.6cm*3.5) {\indexstyle{c.2}};
        \node[mybaseline, regular polygon, draw, regular polygon sides = 6, minimum size = 0.6cm, fill=yellow!40, rotate = 120] (c3) at (-120+60:0.6cm*3.5) {};
        \node at (-120+60:0.6cm*3.5) {\indexstyle{c.3}};
        \node[mybaseline, regular polygon, draw, regular polygon sides = 6, minimum size = 0.6cm, fill=yellow!40, rotate = 180] (c4) at (-180+60:0.6cm*3.5) {};
        \node at (-180+60:0.6cm*3.5) {\indexstyle{c.4}};
        \node[mybaseline, regular polygon, draw, regular polygon sides = 6, minimum size = 0.6cm, fill=yellow!40, rotate = 240] (c5) at (-240+60:0.6cm*3.5) {};
        \node at (-240+60:0.6cm*3.5) {\indexstyle{c.5}};
        \node[mybaseline, regular polygon, draw, regular polygon sides = 6, minimum size = 0.6cm, fill=yellow!40, rotate = 300] (c6) at (-300+60:0.6cm*3.5) {};
        \node at (-300+60:0.6cm*3.5) {\indexstyle{c.6}};
        \draw[mybaseline,->] (a)--(b1);   
        \draw[mybaseline,->] (a)--(b2);   
        \draw[mybaseline,->] (a)--(b3);   
        \draw[mybaseline,->] (a)--(b4);   
        \draw[mybaseline,->] (a)--(b5);   
        \draw[mybaseline,->] (a)--(b6); 
        \draw[mybaseline,->] (b1.330)to[in=210,out=330](c1.210);   
        \draw[mybaseline,->] (b1.210)to[in=330,out=210](c6.30);   
        \draw[mybaseline,->] (b1.30)--++ (30:0.2);   
        \draw[mybaseline,->] (b1.90)--++(90:0.2);   
        \draw[mybaseline,->] (b1.150)--++(150:0.2); 
        \draw[mybaseline,->] (c1.270)to[in=90+60,out=270](b2.90);   
        \draw[mybaseline,->] (c2.90)to[in=210+60,out=90+60](b2.210);   
        \draw[mybaseline,->] (b2.30)--++ (30+60:0.2);   
        \draw[mybaseline,->] (b2.270)--++(270+60:0.2);   
        \draw[mybaseline,->] (b2.330)--++(330+60:0.2);
        \draw[mybaseline,->] (b3.330)to[in=150+60,out=330+120](c2.150);   
        \draw[mybaseline,->] (b3.90)to[in=330+120,out=90+120](c3.330);   
        \draw[mybaseline,->] (b3.270)--++ (270+120:0.2);   
        \draw[mybaseline,->] (b3.210)--++(210+120:0.2);   
        \draw[mybaseline,->] (b3.150)--++(150+120:0.2);
        \draw[mybaseline,->] (b4.210)to[in=30+120,out=210+180](c3.30);   
        \draw[mybaseline,->] (b4.330)to[in=210+180,out=330+180](c4.210);   
        \draw[mybaseline,->] (b4.30)--++ (30+180:0.2);   
        \draw[mybaseline,->] (b4.90)--++(90+180:0.2);   
        \draw[mybaseline,->] (b4.150)--++(150+180:0.2);
        \draw[mybaseline,->] (b5.210)to[in=90+240,out=210+240](c5.90);   
        \draw[mybaseline,->] (b5.90)to[in=270+180,out=90+240](c4.270);   
        \draw[mybaseline,->] (b5.30)--++ (30+240:0.2);   
        \draw[mybaseline,->] (b5.330)--++(330+240:0.2);   
        \draw[mybaseline,->] (b5.270)--++(270+240:0.2);
        \draw[mybaseline,->] (b6.90)to[in=330+300,out=90+300](c6.330);   
        \draw[mybaseline,->] (b6.330)to[in=150+240,out=330+300](c5.150);   
        \draw[mybaseline,->] (b6.150)--++ (150+300:0.2);   
        \draw[mybaseline,->] (b6.210)--++(210+300:0.2);   
        \draw[mybaseline,->] (b6.270)--++(270+300:0.2);
        \draw[mybaseline,->] (c1.30)--++ (30:0.2);     
        \draw[mybaseline,->] (c1.90)--++(90:0.2); 
        \draw[mybaseline,->] (c1.150)--++(150:0.2); 
        \draw[mybaseline,->] (c1.330)--++(330:0.2); 
        \draw[mybaseline,->] (c2.30)--++ (30+60:0.2);   
        \draw[mybaseline,->] (c2.270)--++(270+60:0.2);   
        \draw[mybaseline,->] (c2.210)--++(210+60:0.2); 
        \draw[mybaseline,->] (c2.330)--++(330+60:0.2);
        \draw[mybaseline,->] (c3.90)--++ (90+120:0.2); 
        \draw[mybaseline,->] (c3.150)--++(150+120:0.2);
        \draw[mybaseline,->] (c3.210)--++(210+120:0.2);   
        \draw[mybaseline,->] (c3.270)--++(270+120:0.2); 
        \draw[mybaseline,->] (c4.30)--++ (30+180:0.2); 
        \draw[mybaseline,->] (c4.90)--++(90+180:0.2);
        \draw[mybaseline,->] (c4.330)--++(330+180:0.2);   
        \draw[mybaseline,->] (c4.150)--++(150+180:0.2); 
        \draw[mybaseline,->] (c5.30)--++ (30+240:0.2); 
        \draw[mybaseline,->] (c5.330)--++(330+240:0.2);
        \draw[mybaseline,->] (c5.270)--++(270+240:0.2);   
        \draw[mybaseline,->] (c5.210)--++(210+240:0.2); 
        \draw[mybaseline,->] (c6.90)--++ (90+300:0.2); 
        \draw[mybaseline,->] (c6.150)--++(150+300:0.2);
        \draw[mybaseline,->] (c6.210)--++(210+300:0.2);   
        \draw[mybaseline,->] (c6.270)--++(270+300:0.2);
        \node[tnode, fill = red] (U) at ($(0,0)!0.5!(30:0.6cm*3)$) {};
        \node[tnrects, fill = blue,minimum size = 3pt] (O1) at ($(0,0)!0.5!(90:0.6cm*3)$) {};
        \node[tnrects, fill = blue,minimum size =3pt] (O2) at ($(0,0)!0.5!(150:0.6cm*3)$) {};
        \node[tnrects, fill = blue,minimum size = 3pt] (O3) at ($(0,0)!0.5!(210:0.6cm*3)$) {};
        \node[tnrects, fill = blue,minimum size = 3pt] (O5) at (52:0.6cm*2.9) {};
        \node[tnrects, fill = blue,minimum size = 3pt] (O6) at (8:0.6cm*2.9) {};
        \node[tnrects, fill = blue,minimum size = 3pt] (O4) at ($(0,0)!0.5!(270:0.6cm*3)$) {};
    },
\end{equation}
and use the mapping from tensor networks to causal models to define a model on the above graph. Notice that the graph above is acyclic and the mapping leads to a valid causal model for each tensor as the number of chosen inputs of each tensor are less or equals than those of the outputs.

The mapping also provides us with causal mechanisms associated to each vertex of the graph. On the other hand, we can apply graph-separation theorems to the obtained graph to already infer no-signalling properties. Indeed, as can be easily proven using the definition of $d$-separation (see~\Cref{app:graphsep}), the vertex represented by the red circle is $d$-separated from the blue squares. More intuitively, this $d$-separation captures that, relative to the partial order induced by the acyclic causal model, none of the blue squares are in the future of the red circle (and vice versa), and we would thus expect that the red circle cannot signal to the blue squares. This is indeed ensured by the $d$-separation theorem. Because of the soundness of $d$-separation, any causal model on the graph will not show correlations between the random variable associated to the labs $\region{A}$ and the one associated to the labs $\region{B}$, i.e.,
\begin{equation}
\label{eq:Htn ex}
    \prob{(y|a,b)_{\textup{I}}} =  \prob{(y|a',b)_{\textup{I}}}
\end{equation}
for all $a,a'\in\outcomemaparg{A}$.\footnote{More precisely, while $d$-separation of the red circle from the blue squares used above is a sufficient condition to ensure no-signalling from $\region{A}$ (red circle) to $\region{B}$ (blue squares), a tighter condition would be to include a parentless node, say $a$ pointing to the red circle (modelling the free choice of operation performed at $\region{A}$) and consider whether $a$ is $d$-separated from the blue squares. Indeed, we can have situations where the blue squares are $d$-connected to the red circle (e.g., when the red circle lies in the ``future'' of the blue squares or equivalently when there exists a directed path from the blue squares to the red circle) but remain $d$-separated from the parentless node $a$. Even in this case, the $d$-separation theorem would guarantee no-signalling as it would prevent $a$ from getting correlated with the outcomes $y$ at the blue squares.} Using~\cref{lem:signalling_CI}, proving that there is causal influence from $\region{A}$ to $\region{B}$ if and only if there is signalling, one can immediately conclude that there is no causal influence from $\region{A}$ to $\region{B}$.

Observe that the method used above consists in choosing a mapping from tensor networks to causal models and construct the corresponding graph including the relevant in intervention. In this example, we were able to construct an acyclic graph and use the $d$-separation theorem to infer no signalling directly from the connectivity of the graph. If the graph is cyclic, which is the case for a general mapping, one can use $p$-separation~\cite{Quantum_paper} instead. For instance, consider the example of Figure 18 in~\cite{Cotler_2019}: 
\begin{equation}
\label{eq:TNpconn}
    \ketbra{\Xi(U_\region{A}, O_\region{B})} = \centertikz{
       \foreach \ang in {0,60, ...,300}{
            \draw[mybaseline] (\ang+90:0.6cm*3) -- (0,0);
            \draw[mybaseline] (\ang+90:0.6cm*3) to [in = \ang+210, out = \ang+330] (\ang+60:0.6cm*3.5);
            \draw[mybaseline] (\ang+60+60:0.6cm*3.5) to [in = \ang+210, out = \ang+330]  (\ang+90:0.6cm*3); 
            \foreach \an in {0,60,120}{
                    \draw[mybaseline] (\ang+90:0.6cm*3) -- ++(\an+30+\ang:0.5);
                }
            \foreach \an in {0,60,120,180}{
                    \draw[mybaseline] (\ang+60:0.6cm*3.5) -- ++(\ang+\an-30:0.4);
                }
           } 
        \foreach \ang/\b in {0/1,60/2, 120/3, 180/4, 240/5, 300/6}{
           \node[mybaseline, regular polygon, draw, regular polygon sides = 6, minimum size = 0.8cm, fill=yellow!40, rotate = \ang] (p1) at (-\ang+90:0.6cm*3) {};
           \node at (-\ang+90:0.6cm*3) {\indexstyle{b.\b}};
            \node[mybaseline, regular polygon, draw, regular polygon sides = 6, minimum size = 0.6cm, fill=yellow!40, rotate = \ang] (p2) at (-\ang+60:0.6cm*3.5) {};
            \node at (-\ang+60:0.6cm*3.5) {\indexstyle{c.\b}};
         }    
        \node[mybaseline, regular polygon, draw, regular polygon sides = 6, minimum size = 1.2cm, fill=yellow!40] (p) at (0,0) {\indexstyle{a}};
        \node[tnode, fill = red] (U) at ($(0,0)!0.5!(-150:0.6cm*3)$) {};
        \node[tnrects, fill = blue,minimum size = 3pt] (O1) at ($(0,0)!0.5!(-30:0.6cm*3)$) {};
        \node[tnrects, fill = blue,minimum size = 3pt] (O3) at ($(0,0)!0.5!(30:0.6cm*3)$) {};
        \node[tnrects, fill = blue,minimum size = 3pt] (O4) at ($(0,0)!0.5!(270:0.6cm*3)$) {};
        \node[tnrects, fill = blue,minimum size = 3pt] (O5) at (230:0.6cm*2.9) {};
        \node[tnrects, fill = blue,minimum size = 3pt] (U) at ($(120+90:0.6cm*3)+(120+30+120:0.55)$) {};
        \node[tnrects, fill = blue,minimum size = 3pt] (U) at ($(120+90:0.6cm*3)+(60+30+120:0.55)$) {};   
    }.
\end{equation}
One can easily see that there is no choice of directions between the input-output spaces of $a$ and $b.5$ that would allow to make use of the perfect tensor property and have all blue squares outside the future of the red circle, with respect to the partial order induced by the causal model. Indeed, depending on the choice of direction of the edge $\{a,b.5\}$ either the tensor $a$ or the tensor $b.5$ would require to have more inputs than outputs in order to ensure that the blue squares are in the past of the unitary.  

We can make use of the generalised mapping to map the network to a cyclic causal model. Let us focus on the tensors $a$ and $b.5$ only, as one can easily see that a consistent choice of directions for the other tensors, which also fulfils the input-output dimension requirements, can always be found.\footnote{This is the case because the two tensors in consideration have at most one edge in common with any other tensor. Hence, for both directions of such edge one can choose the remaining directions on the other tensor to have valid dimensions and to be consistent with the rest of the network. See, for instance, the previous example.}

Consider the two choices of directions for the edge $\{a,b.5\}$:
\begin{equation}
\centertikz{
    \foreach \ang in {0,60}{
        \draw[mybaseline,->] (\ang+90:0.6cm*2.4) -- (\ang+90:0.6cm*0.9);
    }
    \foreach \ang in {120,180,240,300}{
        \draw[mybaseline,->] (0,0) --(\ang+90:0.6cm*2.4);
    }
    \draw[mybaseline] (120+90:0.6cm*3) -- (0,0);

    \foreach \an in {0,300}{
            \draw[mybaseline,->] (120+90:0.6cm*3) -- ++(\an+30+120:0.85);
    }
    \foreach \an in {60,120,180}{
            \draw[mybaseline,,<-] ($(120+90:0.6cm*3)+(\an+30+120:0.35)$) -- ++(\an+30+120:0.52);
    }
                
        \foreach \ang/\b in {240/5}{
           \node[mybaseline, regular polygon, draw, regular polygon sides = 6, minimum size = 0.8cm, fill=yellow!40, rotate = \ang] (p1) at (-\ang+90:0.6cm*3) {};
           \node at (-\ang+90:0.6cm*3) {\indexstyle{b.\b}};

         }    
        \node[mybaseline, regular polygon, draw, regular polygon sides = 6, minimum size = 1.2cm, fill=yellow!40] (p) at (0,0) {\indexstyle{a}};
        \node[tnode, fill = red] (U) at ($(0,0)!0.5!(-150:0.6cm*3)$) {};
        \node[tnrects, fill = blue,minimum size = 3pt] (O1) at ($(0,0)!0.5!(330:0.6cm*3)$) {};
        \node[tnrects, fill = blue,minimum size = 3pt] (O3) at ($(0,0)!0.5!(30:0.6cm*3)$) {};
        \node[tnrects, fill = blue,minimum size = 3pt] (O4) at ($(0,0)!0.5!(270:0.6cm*3)$) {};
        \node[tnrects, fill = blue,minimum size = 3pt] (U) at ($(120+90:0.6cm*3)+(120+30+120:0.7)$) {};
        \node[tnrects, fill = blue,minimum size = 3pt] (U) at ($(120+90:0.6cm*3)+(60+30+120:0.7)$) {}; 
        \node[tnrects, fill = blue,minimum size = 3pt] (U) at ($(120+90:0.6cm*3)+(180+30+120:0.7)$) {};
    } \quad \textup{ or } \quad  
    \centertikz{
        \foreach \ang in {120,180,240,300}{
            \draw[mybaseline,->] (\ang+90:0.6cm*2.4) -- (\ang+90:0.6cm*0.9);
        }
        \foreach \ang in {0,60}{
            \draw[mybaseline,->] (0,0) --(\ang+90:0.6cm*2.4);
        }
        \draw[mybaseline] (120+90:0.6cm*3) -- (0,0);

        \foreach \an in {0,300}{
                \draw[mybaseline,->] (120+90:0.6cm*3) -- ++(\an+30+120:0.85);
        }
        \foreach \an in {60,120,180}{
                \draw[mybaseline,,<-] ($(120+90:0.6cm*3)+(\an+30+120:0.35)$) -- ++(\an+30+120:0.52);
        }
                
        \foreach \ang/\b in {240/5}{
           \node[mybaseline, regular polygon, draw, regular polygon sides = 6, minimum size = 0.8cm, fill=yellow!40, rotate = \ang] (p1) at (-\ang+90:0.6cm*3) {};
           \node at (-\ang+90:0.6cm*3) {\indexstyle{b.\b}};

         }    
        \node[mybaseline, regular polygon, draw, regular polygon sides = 6, minimum size = 1.2cm, fill=yellow!40] (p) at (0,0) {\indexstyle{a}};
        \node[tnode, fill = red] (U) at ($(0,0)!0.5!(-150:0.6cm*3)$) {};
        \node[tnrects, fill = blue,minimum size = 3pt] (O1) at ($(0,0)!0.5!(330:0.6cm*3)$) {};
        \node[tnrects, fill = blue,minimum size = 3pt] (O3) at ($(0,0)!0.5!(30:0.6cm*3)$) {};
        \node[tnrects, fill = blue,minimum size = 3pt] (O4) at ($(0,0)!0.5!(270:0.6cm*3)$) {};
        \node[tnrects, fill = blue,minimum size = 3pt] (U) at ($(120+90:0.6cm*3)+(120+30+120:0.7)$) {};
        \node[tnrects, fill = blue,minimum size = 3pt] (U) at ($(120+90:0.6cm*3)+(60+30+120:0.7)$) {}; 
        \node[tnrects, fill = blue,minimum size = 3pt] (U) at ($(120+90:0.6cm*3)+(180+30+120:0.7)$) {};
    }.
\end{equation}
This represents a specific choice of directions, in general the direction of the $\{a,b.5\}$ and the requirement that all blue square operators should not be in the causal past of the unitary imposes requirements on the edges involving operators. Specifically all edges of this kind that are connected to the tensor where $\{a,b.5\}$ is ingoing, also have to be ingoing. In the left choice of the diagram, this constraint is applied in $b.5$ and in the diagram on the right it is applied in $a$. The other directions can be freely chosen in each case respectively, and the diagrams above are just an example of choice.

In both options one of the tensors certainly does not fulfil the perfect tensor properties, as it has four inputs and two outputs. Thus, when mapping it to a causal model we have to make use of the generalised mapping of~\cref{def: TNtoCMgeneral}, which introduces a self cycle in the vertex of such tensor, leading to:
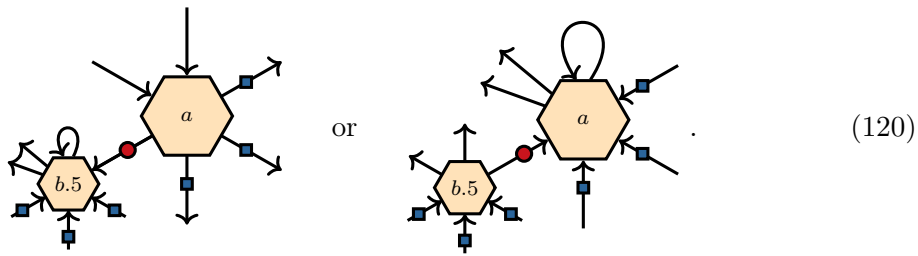
\begin{equation}
\centertikz{
    \foreach \ang in {0,60}{
        \draw[mybaseline,->] (\ang+90:0.6cm*2.4) -- (\ang+90:0.6cm*0.9);
    }
    \foreach \ang in {120,180,240,300}{
        \draw[mybaseline,->] (0,0) --(\ang+90:0.6cm*2.4);
    }
    \draw[mybaseline] (120+90:0.6cm*3) -- (0,0);

    \foreach \an in {10,350}{
            \draw[mybaseline,->] (120+90:0.6cm*3) -- ++(\an+30+120:0.85);
    }
    \foreach \an in {60,120,180}{
            \draw[mybaseline,,<-] ($(120+90:0.6cm*3)+(\an+30+120:0.35)$) -- ++(\an+30+120:0.52);
    }

    \draw[mybaseline,->] ($(120+90:0.6cm*3)+(290+30+120:0.37)$) .. controls ($(120+90:0.6cm*3)+(290+30+120:0.37)+(0.25,0.5)$) and ($(120+90:0.6cm*3)+(310+30+120:0.37)+(-0.25,0.5)$) .. ($(120+90:0.6cm*3)+(310+30+120:0.37)$);
                
        \foreach \ang/\b in {240/5}{
           \node[mybaseline, regular polygon, draw, regular polygon sides = 6, minimum size = 0.8cm, fill=yellow!40, rotate = \ang] (p1) at (-\ang+90:0.6cm*3) {};
           \node at (-\ang+90:0.6cm*3) {\indexstyle{b.\b}};

         }    
        \node[mybaseline, regular polygon, draw, regular polygon sides = 6, minimum size = 1.2cm, fill=yellow!40] (p) at (0,0) {\indexstyle{a}};
        \node[tnode, fill = red] (U) at ($(0,0)!0.5!(-150:0.6cm*3)$) {};
        \node[tnrects, fill = blue,minimum size = 3pt] (O1) at ($(0,0)!0.5!(330:0.6cm*3)$) {};
        \node[tnrects, fill = blue,minimum size = 3pt] (O3) at ($(0,0)!0.5!(30:0.6cm*3)$) {};
        \node[tnrects, fill = blue,minimum size = 3pt] (O4) at ($(0,0)!0.5!(270:0.6cm*3)$) {};
        \node[tnrects, fill = blue,minimum size = 3pt] (U) at ($(120+90:0.6cm*3)+(120+30+120:0.7)$) {};
        \node[tnrects, fill = blue,minimum size = 3pt] (U) at ($(120+90:0.6cm*3)+(60+30+120:0.7)$) {}; 
        \node[tnrects, fill = blue,minimum size = 3pt] (U) at ($(120+90:0.6cm*3)+(180+30+120:0.7)$) {};
    } \quad \textup{ or } \quad  
    \centertikz{
        \foreach \ang in {120,180,240,300}{
            \draw[mybaseline,->] (\ang+90:0.6cm*2.4) -- (\ang+90:0.6cm*0.9);
        }
        \foreach \ang in {50,70}{
            \draw[mybaseline,->] (0,0) --(\ang+90:0.6cm*2.4);
        }
       \begin{scope}[xscale=-1]
           \draw[mybaseline,->] (10+90:0.6cm*0.9) .. controls ($(90:0.6cm*0.9)+(-0.75,1)$) and ($(90:0.6cm*0.9)+(0.75,1)$) .. (-10+90:0.6cm*0.9);
       \end{scope}
        
        \foreach \an in {0,300}{
                \draw[mybaseline,->] (120+90:0.6cm*3) -- ++(\an+30+120:0.85);
        }
        \foreach \an in {60,120,180}{
                \draw[mybaseline,,<-] ($(120+90:0.6cm*3)+(\an+30+120:0.35)$) -- ++(\an+30+120:0.52);
        }
                
        \foreach \ang/\b in {240/5}{
           \node[mybaseline, regular polygon, draw, regular polygon sides = 6, minimum size = 0.8cm, fill=yellow!40, rotate = \ang] (p1) at (-\ang+90:0.6cm*3) {};
           \node at (-\ang+90:0.6cm*3) {\indexstyle{b.\b}};

         }    
        \node[mybaseline, regular polygon, draw, regular polygon sides = 6, minimum size = 1.2cm, fill=yellow!40] (p) at (0,0) {\indexstyle{a}};
        \node[tnode, fill = red] (U) at ($(0,0)!0.5!(-150:0.6cm*3)$) {};
        \node[tnrects, fill = blue,minimum size = 3pt] (O1) at ($(0,0)!0.5!(330:0.6cm*3)$) {};
        \node[tnrects, fill = blue,minimum size = 3pt] (O3) at ($(0,0)!0.5!(30:0.6cm*3)$) {};
        \node[tnrects, fill = blue,minimum size = 3pt] (O4) at ($(0,0)!0.5!(270:0.6cm*3)$) {};
        \node[tnrects, fill = blue,minimum size = 3pt] (U) at ($(120+90:0.6cm*3)+(120+30+120:0.7)$) {};
        \node[tnrects, fill = blue,minimum size = 3pt] (U) at ($(120+90:0.6cm*3)+(60+30+120:0.7)$) {}; 
        \node[tnrects, fill = blue,minimum size = 3pt] (U) at ($(120+90:0.6cm*3)+(180+30+120:0.7)$) {};

    }.
\end{equation}
Signalling relations in these graph can be addressed using a different graph separation property, $p$-separation, which is sound and complete for finite dimensional cyclic causal models\footnote{Tensor networks are also defined for finite dimensional, discrete indices, hence this is not restrictive. The settings labelling the interventions that we consider might be infinite dimensional. However, one can easily see that systems, i.e., Hilbert spaces or indices, associated to edges which are not involved in loops, may not be finite without compromising the results of~\cite{Quantum_paper}. Settings are associated to exogenous vertices, thus there are never involved in loops.}. In this example, one can prove that the self-loops $p$-connect the unitary to the region denoted with squares.
By the $p$-separation theorem~\cite{Quantum_paper}, this means that generically, there can be signalling in causal models on such a graph (and by our results, causal influence in the associated tensor network), except in the case of certain fine-tuned causal models. Checking whether or not the scenario involves such fine-tuning that washes out the signalling (or operational causal influence), would require explicit calculation of the associated quantities by applying~\cref{def: signalling} (or~\cref{def: causal influence}) for the given causal model (or tensor network), and cannot be solely determined by analysing the graph. 

\begin{remark}
    While the holographic models for tensor networks heavily rely on the perfect tensor property for analysing causal influence, causal model methods can be applied to arbitrary tensors. In case the mapping of~\cref{def: TNtoCMgeneral} can be applied, one might be able to construct an acyclic graph and make considerations using the $d$-separation theorem. More generally, one can apply~\cref{def: TNtoCMgeneral} and obtain a cyclic causal model with self-loops added to a subset of vertices, as prescribed by the generalised mapping, and use the $p$-separation theorem, which generalises $d$-separation to cyclic scenarios. 

\end{remark}

\section{Discussion and outlook}
\label{sec:conclusion}

In this work, we established explicit mappings between tensor networks and quantum causal models, enabling a causal interpretation of arbitrary tensor networks (as detailed in the summary of contributions in \Cref{sec: summary_contrib}). In particular, this makes it possible to apply causal reasoning backed by sound and complete graph-separation theorems even in the absence of any a priori notion of directionality. 

We showed that a causal model uniquely determines a tensor network, while the inverse mapping is non-unique, with a single tensor network corresponding to multiple distinct causal models on fundamentally different causal graphs (depending on a choice of orientation for each undirected edge). These mappings motivated the introduction of an operational notion of causal influence for tensor networks. A notable consequence is the existence of families of genuinely different models sharing identical no-signalling relations, suggesting a form of causal equivalence. Finally, we illustrated the use of this framework by applying causal-inference tools to holographic tensor networks to infer causal influence.

We now discuss interesting future directions for building on this framework and results:
\begin{itemize}
    \item \textbf{Rotations of causal models:} The discovery of equivalence classes of ``rotated'' causal models, with the same signalling structure under time-symmetric interventions but different causal graphs, opens several interesting questions: which properties of causal models are preserved under such rotations? are there novel properties arising through this equivalence, similarly to how dual unitaries exhibit interesting features along light-like directions~\cite{Piroli_2020,Zhou_2022}?
    In addition, this result suggests a potential connection between fine-tuned relations in one causal model and sound graph-separation properties in another representative of the same equivalence class. 
    Indeed, the observed non-signalling relations may arise from different sources across the equivalence class, for instance from $d$- or $p$-separation~\cite{pearl_2009,Quantum_paper} or from fine-tuning, depending on the underlying causal structure. 
    \item \textbf{Graph-separation theorem for tensor networks on undirected graphs:} Graph-separation criteria, such as $d$- and $p$-separation, are defined for directed graphs and proven to be sound and complete for subsets causal models. One might wonder whether these criteria induce a corresponding separation theorem for tensor networks defined on undirected graphs via the mappings developed in this work. Preliminary evidence contained in the tensor-network examples of~\Cref{sec:graphsepHolo} suggests that this may indeed be the case. For instance, in the tensor network of~\cref{eq:TNpconn}, the squared region $\region{B}$ cannot be separated from the unitary $\region{A}$ under any orientation of edges that yields a valid causal model through our original mapping: the fact that $\region{B}$ occupies three edges of the same tensor to which $\region{A}$ is connected prevents a choice of input–output directions consistent with the perfect-tensor property that would $d$-separate the two regions. While a generalised mapping allowing additional cyclicity can be constructed, the resulting self-cycles $p$-connect $\region{A}$ and $\region{B}$, again preventing separation. By contrast, in the tensor network of~\cref{eq: ex Hstate labels}, where two regions share fewer edges with a common tensor, suitable edge orientations exist that yield an acyclic causal model in which the regions are $d$-separated and hence causally disconnected. This suggests that one could define a purely graph-theoretic separation criterion for undirected graphs that might be sound and complete for tensor networks, inheriting graph-separation properties of the associated equivalence class of causal models. 
    In the special case of holographic tensor networks, one could formally relate such property to the geodesic-based criterion mentioned above~\cite{Pastawski_2015,Cotler_2019}.

    \item \textbf{Emergence of space-time from operational properties:} Tensor networks are used to define the geometry of space-time in terms of the structure of the network, which determines properties such as its curvature and the distance between two points. For instance, the so-called ``MERA'' tensor network~\cite{Vidal_2008}, with its specific local and global structure, has been studied to define distance in terms of the number of edges that must be traversed to connect two tensors of the network, thus inferring how a fundamental concept like distance might arise from the connectivity of processes. More generally, defining space-time properties using tensor networks remains an active area of research~\cite{May_2017,Cotler_2019, Vidal_2008}, which provides new insights into the fundamental nature of space-time and its relationship to other physical phenomena. Our framework maps a given tensor network to a causal model, a connection that could allow us to understand space-time properties emerging in tensor networks as well-known operational properties in causal models, and to be able to consistently apply powerful graph-theoretic methods from causal models to such problems.
    \item \textbf{Indefinite causal order:}
    Our work allows one to connect so-called indefinite causal order (ICO) processes or higher-order quantum operations~\cite{Oreshkov_2012,Chiribella_2013} to tensor network models, since these are known to be a subset of cyclic quantum causal structures~\cite{Araujo_2017} that can be causally modelled in the framework of~\cite{Quantum_paper} used here. This would lend further techniques and ideas for exploring realisations of such processes in regimes beyond classical space-time considered in \cite{VilasiniRennerPRA, VilasiniRennerPRL}, including quantum gravitational regimes where space-time geometry can exhibit quantum superposition~\cite{Zych_2019,Castro_Ruiz_2020,Paunkovi__2020, Moller_2024,vilasini_2025events,Kabel_2025}.
 Further developing the links between tensor network and causal models established here, could contribute to connecting tensor-network models of quantum gravity with quantum causal structures and higher-order quantum operations, offering a more unified language to describe causality beyond fixed causal orderings. Finally, it would also be interesting to consider the application of these ideas to study the conditions under which definite causal and temporal order may emerge from a world that allows arbitrary indefinite causal order processes.
\end{itemize}

\bigskip
\paragraph{Acknowledgements}
We thank Victor Gitton for his supervision and guidance during the course of the Master’s thesis \cite{Ferradini_2023} on which this work is based. CF and GM acknowledge support from the Swiss National Science Foundation via project No. 20QU-1\_225171, the NCCR SwissMAP, and the ETH Zurich Quantum Center. 
CF also acknowledges support from the ETH Foundation.
GM also acknowledges support from the CHIST-ERA project MoDIC.
VV's research at ETH was supported by an ETH Postdoctoral Fellowship. VV also acknowledges support from the PEPR integrated project EPiQ ANR-22-PETQ-0007 as part of Plan France 2030.
CF and VV acknowledge funding from the CA23115 - Relativistic Quantum Information (RQI) COST Action, in particular by means of a short term scientific mission (STSM) grant.

\begin{thebibliography}{LMGP{\etalchar{+}}11b}
\newcommand{\enquote}[1]{``#1''}
\providecommand{\url}[1]{\texttt{#1}}
\providecommand{\urlprefix}{URL }
\expandafter\ifx\csname urlstyle\endcsname\relax
  \providecommand{\doi}[1]{doi:\discretionary{}{}{}#1}\else
  \providecommand{\doi}{doi:\discretionary{}{}{}\begingroup \urlstyle{rm}\Url}\fi
\providecommand{\eprint}[2][]{\url{#2}}

\bibitem[ADH14]{almheiri-2014}
A.~Almheiri, X.~Dong, D.~Harlow, \enquote{{Bulk locality and quantum error correction in AdS/CFT}}, \emph{Journal of High Energy Physics}, 2015, \doi{10.1007/jhep04(2015)163}, 2014.

\bibitem[AGB17]{Araujo_2017}
M.~Ara\'{u}jo, P.~A. Gu\'{e}rin, {\"A}.~Baumeler, \enquote{Quantum computation with indefinite causal structures}, \emph{Physical Review A}, 96(5), \doi{10.1103/physreva.96.052315}, 2017.

\bibitem[BBC{\etalchar{+}}93]{Bennett1993}
C.~H. Bennett, G.~Brassard, C.~Cr\'epeau, R.~Jozsa, A.~Peres, W.~K. Wootters, \enquote{Teleporting an unknown quantum state via dual classical and {Einstein-Podolsky-Rosen} channels}, \emph{Phys. Rev. Lett.}, 70:1895--1899, \doi{10.1103/PhysRevLett.70.1895}, 1993.

\bibitem[BFPM21]{Bongers_2021}
S.~Bongers, P.~Forré, J.~Peters, J.~M. Mooij, \enquote{Foundations of structural causal models with cycles and latent variables}, \emph{The Annals of Statistics}, 49(5):2885–2915, \doi{10.1214/21-aos2064}, 2021.

\bibitem[BLO20]{Barrett_2019}
J.~Barrett, R.~Lorenz, O.~Oreshkov, \enquote{Quantum causal models}, \emph{arXiv:1906.10726 [quant-ph]}, \doi{10.48550/arXiv.1906.10726}, 2020.

\bibitem[BLO21]{Barrett_2021}
J.~Barrett, R.~Lorenz, O.~Oreshkov, \enquote{Cyclic quantum causal models}, \emph{Nature Communications}, 12(1):885, \doi{10.1038/s41467-020-20456-x}, 2021.

\bibitem[CDPV13]{Chiribella_2013}
G.~Chiribella, G.~M. D’Ariano, P.~Perinotti, B.~Valiron, \enquote{Quantum computations without definite causal structure}, \emph{Physical Review A}, 88(2), \doi{10.1103/physreva.88.022318}, 2013.

\bibitem[Cho75]{Choi1975}
M.-D. Choi, \enquote{{Completely positive linear maps on complex matrices}}, \emph{Linear Algebra Appl.}, 10(3):285--290, \doi{10.1016/0024-3795(75)90075-0}, 1975.

\bibitem[CHQY19]{Cotler_2019}
J.~Cotler, X.~Han, X.-L. Qi, Z.~Yang, \enquote{Quantum causal influence}, \emph{Journal of High Energy Physics}, 2019(42), \doi{10.1007/jhep07(2019)042}, 2019.

\bibitem[CJQW18]{Cotler_2018}
J.~Cotler, C.-M. Jian, X.-L. Qi, F.~Wilczek, \enquote{Superdensity operators for spacetime quantum mechanics}, \emph{Journal of High Energy Physics}, 2018(9), \doi{10.1007/jhep09(2018)093}, 2018.

\bibitem[CRGBB20]{Castro_Ruiz_2020}
E.~Castro-Ruiz, F.~Giacomini, A.~Belenchia, {\v C}.~Brukner, \enquote{Quantum clocks and the temporal localisability of events in the presence of gravitating quantum systems}, \emph{Nature Communications}, 11(1), \doi{10.1038/s41467-020-16013-1}, 2020.

\bibitem[CS16]{Costa_2016}
F.~Costa, S.~Shrapnel, \enquote{Quantum causal modelling}, \emph{New Journal of Physics}, 18(6):063032, \doi{10.1088/1367-2630/18/6/063032}, 2016.

\bibitem[Fer23]{Ferradini_2023}
C.~Ferradini, \enquote{Connecting tensor networks to quantum causal models with applications to holography}, \emph{ETH Zurich (Master thesis)}, \doi{10.3929/ethz-b-000671952}, 2023.

\bibitem[FGV25a]{Ferradini_2025C}
C.~Ferradini, V.~Gitton, V.~Vilasini, \enquote{Cyclic functional causal models beyond unique solvability with a graph separation theorem}, \emph{arXiv:2502.04171 [math.ST]}, \doi{10.48550/arXiv.2502.04171}, 2025.

\bibitem[FGV25b]{Quantum_paper}
C.~Ferradini, V.~Gitton, V.~Vilasini, \enquote{Cyclic quantum causal modelling with a graph separation theorem}, \emph{arXiv:2502.04168 [quant-ph]}, \doi{10.48550/arXiv.2502.04168}, 2025.

\bibitem[FM18]{forre_2018}
P.~Forré, J.~M. Mooij, \enquote{Constraint-based causal discovery for non-linear structural causal models with cycles and latent confounders}, \emph{arXiv:1807.03024 [stat.ML]}, \doi{10.48550/arXiv.1807.03024}, 2018.

\bibitem[FV]{Interventions_paper}
C.~Ferradini, V.~Vilasini, \enquote{Interventions in cyclic quantum causal models}, (in preparation).

\bibitem[HLP14]{Henson_2014}
J.~Henson, R.~Lal, M.~F. Pusey, \enquote{Theory-independent limits on correlations from generalized bayesian networks}, \emph{New Journal of Physics}, 16(11):113043, \doi{10.1088/1367-2630/16/11/113043}, 2014.

\bibitem[HNQ{\etalchar{+}}16]{hayden-2016}
P.~Hayden, S.~Nezami, X.-L. Qi, N.~Thomas, M.~A. Walter, Z.~Yang, \enquote{{Holographic duality from random tensor networks}}, \emph{Journal of High Energy Physics}, 2016(11), \doi{10.1007/jhep11(2016)009}, 2016.

\bibitem[Jam72]{Jamiolkowski1972}
A.~Jamio{\l}kowski, \enquote{{Linear transformations which preserve trace and positive semidefiniteness of operators}}, \emph{Rept. Math. Phys.}, 3:275--278, \doi{10.1016/0034-4877(72)90011-0}, 1972.

\bibitem[JSV25]{jean_2025}
E.~Jean, R.~Silva, V.~Vilasini, \enquote{An equivalence between time-symmetry and cyclic causality in quantum theory}, \emph{arXiv:22508.02463 [quant-ph]}, \doi{10.48550/arXiv.2508.02463}, 2025.

\bibitem[KHA{\etalchar{+}}25]{Kabel_2025}
V.~Kabel, A.-C. de~la Hamette, L.~Apadula, C.~Cepollaro, H.~Gomes, J.~Butterfield, {\v C}.~Brukner, \enquote{Quantum coordinates, localisation of events, and the quantum hole argument}, \emph{Communications Physics}, 8(1), \doi{10.1038/s42005-025-02084-3}, 2025.

\bibitem[LMGP{\etalchar{+}}11a]{Lloyd_2011_2}
S.~Lloyd, L.~Maccone, R.~Garcia-Patron, V.~Giovannetti, Y.~Shikano, \enquote{Quantum mechanics of time travel through post-selected teleportation}, \emph{Physical Review D}, 84(2):025007, \doi{10.1103/physrevd.84.025007}, 2011.

\bibitem[LMGP{\etalchar{+}}11b]{Lloyd_2011}
S.~Lloyd, L.~Maccone, R.~Garcia-Patron, V.~Giovannetti, Y.~Shikano, S.~Pirandola, L.~A. Rozema, A.~Darabi, Y.~Soudagar, L.~K. Shalm, A.~M. Steinberg, \enquote{Closed timelike curves via postselection: Theory and experimental test of consistency}, \emph{Physical Review Letters}, 106(4):040403, \doi{10.1103/physrevlett.106.040403}, 2011.

\bibitem[LN07]{Levin_2007}
M.~Levin, C.~P. Nave, \enquote{Tensor renormalization group approach to two-dimensional classical lattice models}, \emph{Physical Review Letters}, 99(12), \doi{10.1103/physrevlett.99.120601}, 2007.

\bibitem[Mal98]{maldacena-1998}
J.~Maldacena, \enquote{{The large $N$ limit of superconformal field theories and supergravity}}, \emph{Advances in Theoretical and Mathematical Physics}, 2(2):231--252, \doi{10.4310/atmp.1998.v2.n2.a1}, 1998.

\bibitem[May17]{May_2017}
A.~May, \enquote{Tensor networks for dynamic spacetimes}, \emph{Journal of High Energy Physics}, 2017(6), \doi{10.1007/jhep06(2017)118}, 2017.

\bibitem[Nea00]{Neal_2000}
R.~M. Neal, \enquote{On deducing conditional independence from d-separation in causal graphs with feedback (research note)}, \emph{Journal of Artificial Intelligence Research}, 12:87–91, \doi{10.1613/jair.689}, 2000.

\bibitem[OCB12]{Oreshkov_2012}
O.~Oreshkov, F.~Costa, {\v C}.~Brukner, \enquote{Quantum correlations with no causal order}, \emph{Nature Communications}, 3(1), \doi{10.1038/ncomms2076}, 2012.

\bibitem[PBCP20]{Piroli_2020}
L.~Piroli, B.~Bertini, J.~I. Cirac, T.~Prosen, \enquote{Exact dynamics in dual-unitary quantum circuits}, \emph{Physical Review B}, 101(9), \doi{10.1103/physrevb.101.094304}, 2020.

\bibitem[Pea09]{pearl_2009}
J.~Pearl, \emph{Causality}, Cambridge University Press, 2nd edition, \doi{10.1017/CBO9780511803161}, 2009.

\bibitem[PMM{\etalchar{+}}17]{Portmann_2017}
C.~Portmann, C.~Matt, U.~Maurer, R.~Renner, B.~Tackmann, \enquote{Causal boxes: Quantum information-processing systems closed under composition}, \emph{IEEE Transactions on Information Theory}, 1–1, \doi{10.1109/tit.2017.2676805}, 2017.

\bibitem[PV20]{Paunkovi__2020}
N.~Paunković, M.~Vojinović, \enquote{Causal orders, quantum circuits and spacetime: distinguishing between definite and superposed causal orders}, \emph{Quantum}, 4:275, \doi{10.22331/q-2020-05-28-275}, 2020.

\bibitem[PYHP15]{Pastawski_2015}
F.~Pastawski, B.~Yoshida, D.~Harlow, J.~Preskill, \enquote{Holographic quantum error-correcting codes: toy models for the bulk/boundary correspondence}, \emph{Journal of High Energy Physics}, 2015(6), \doi{10.1007/jhep06(2015)149}, 2015.

\bibitem[QY18]{Qi_2018}
X.-L. Qi, Z.~Yang, \enquote{Space-time random tensor networks and holographic duality}, \emph{arXiv:1801.05289 [hep-th]}, \doi{10.48550/arXiv.1801.05289}, 2018.

\bibitem[Ren22]{Renes_2022}
J.~M. Renes, \emph{{Quantum Information Theory}}, De Gruyter, ISBN 978-3-11-057024-3, \doi{10.1515/9783110570250}, 2022.

\bibitem[SMSY24]{Moller_2024}
N.~S.~Móller, B.~Sahdo, N.~Yokomizo, \enquote{Gravitational quantum switch on a superposition of spherical shells}, \emph{Quantum}, 8:1248, \doi{10.22331/q-2024-02-12-1248}, 2024.

\bibitem[VC04]{Verstraete_2004}
F.~Verstraete, J.~I. Cirac, \enquote{Renormalization algorithms for quantum-many body systems in two and higher dimensions}, \emph{arXiv:cond-mat/0407066 [cond-mat.str-el]}, \doi{10.48550/arXiv.cond-mat/0407066}, 2004.

\bibitem[VCYR25]{vilasini_2025events}
V.~Vilasini, L.-Q. Chen, L.~Ye, R.~Renner, \enquote{Events and their localisation are relative to a lab}, \emph{arxiv:2505.21797 [quant-ph]}, \doi{10.48550/arXiv.2505.21797}, 2025.

\bibitem[Vid03]{Vidal_2003}
G.~Vidal, \enquote{Efficient classical simulation of slightly entangled quantum computations}, \emph{Physical Review Letters}, 91(14), \doi{10.1103/physrevlett.91.147902}, 2003.

\bibitem[Vid08]{Vidal_2008}
G.~Vidal, \enquote{{Class of Quantum Many-Body States That Can Be Efficiently Simulated}}, \emph{Physical Review Letters}, 101(11), \doi{10.1103/physrevlett.101.110501}, 2008.

\bibitem[VR24a]{VilasiniRennerPRA}
V.~Vilasini, R.~Renner, \enquote{Embedding cyclic information-theoretic structures in acyclic space-times: No-go results for indefinite causality}, \emph{Physical Review A}, 110(2), \doi{10.1103/physreva.110.022227}, 2024.

\bibitem[VR24b]{VilasiniRennerPRL}
V.~Vilasini, R.~Renner, \enquote{Fundamental limits for realizing quantum processes in spacetime}, \emph{Physical Review Letters}, 133(8), \doi{10.1103/physrevlett.133.080201}, 2024.

\bibitem[Wit98]{witten-1998}
E.~Witten, \enquote{{Anti de Sitter space and holography}}, \emph{Advances in Theoretical and Mathematical Physics}, 2(2):253--291, \doi{10.4310/atmp.1998.v2.n2.a2}, 1998.

\bibitem[ZCPB19]{Zych_2019}
M.~Zych, F.~Costa, I.~Pikovski, {\v C}.~Brukner, \enquote{Bell’s theorem for temporal order}, \emph{Nature Communications}, 10(1), \doi{10.1038/s41467-019-11579-x}, 2019.

\bibitem[ZH22]{Zhou_2022}
T.~Zhou, A.~W. Harrow, \enquote{Maximal entanglement velocity implies dual unitarity}, \emph{Physical Review B}, 106(20), \doi{10.1103/physrevb.106.l201104}, 2022.

\end{thebibliography}

\newcommand{\etalchar}[1]{$^{#1}$}

\newpage

\appendix

\section{Proofs}
\label{app:proofs}
\subsection*{Proof of~\cref{lemma:inv_choi}}
\invpurechoi*
\begin{proof}
The following holds for the action of  $\CJin_{B|A}(M_{A'B})$ on any basis element $S_A=\ket{m}\bra{n}_A$
    \begin{equation}
          \CJin_{B|A}(M_{A'B})(\ket{m}\bra{n}_A)=d_A^2\sum_k p_k   \CJp_{B|A} (\ket{\Psi_k}_{A'B})(\ket{m}_A)\otimes \Big(  \CJp_{B|A} (\ket{\Psi_k}_{A'B})(\ket{n}_A)\Big)^\dagger.
    \end{equation}
    \begin{align}
    \begin{split}
        &\CJin_{A,B}(M_{A'B})(\ket{m}\bra{n}_A)\\
        &= d_{A} \sum_k p_k\Tr_{A'}\left(\sum_{i,j} \ketbraa{i}{j}_{A'} \bra{j}(\ket{m}\bra{n}_A)\ket{i}_A \ket{\Psi_k}\bra{\Psi_k}_{A'B}\right)\\
       &= d_{A} \sum_k p_k\sum_l \bra{l}_{A'}\left(\sum_{i,j} \ketbraa{i}{j}_{A'} \bra{j}(\ket{m}\bra{n}_A)\ket{i}_A \ket{\Psi_k}\bra{\Psi_k}_{A'B}\right)\ket{l}_{A'}\\
        &= d_{A} \sum_k p_k \left(\sum_{i,j} \bra{j}_{A'} \bra{j}(\ket{m}\bra{n}_A)\ket{i}_A \ket{\Psi_k}\bra{\Psi_k}_{A'B}\right)\ket{i}_{A'}\\
        &=d_A \sum_k p_k \left( \sum_j \bra{jj}_{AA'}(\ket{m}\bra{n}_A \otimes \ket{\Psi_k}\bra{\Psi_k}_{A'B})\sum_i \ket{ii}_{A'A}\right)\\
        &=d_A \sum_k p_k d_A\left( \frac{1}{\sqrt{d_A}}\sum_j \bra{jj}_{AA'}(\ket{m}_A\otimes \ket{\Psi_k}_{A'B})\right)\otimes \left(       ( \bra{n}_A \otimes \bra{\Psi_k}_{A'B})\frac{1}{\sqrt{d_A}}\sum_i \ket{ii}_{A'A}\right)\\
        &=d_A^2\sum_k p_k   \CJp_{B|A} (\ket{\Psi_k}_{A'B})(\ket{m}_A)\otimes \Big(  \CJp_{B|A} (\ket{\Psi_k}_{A'B})(\ket{n}_A)\Big)^\dagger.
        \end{split}
\end{align}

By linearity, it follows that we can write the following equality at the level of the operators.

\end{proof}

\subsection*{Proof of~\cref{lemma: TNtoCMgeneral}}
The proof of~\cref{lemma: TNtoCMgeneral} follows from results in~\cite{jean_2025} which are formulated for two-time operators. Here we adapt them to fit the language of this work.

First, let us prove the following proposition linking the self-cycle composition to pre- and post-selection through maximally entangled states.
\begin{lemma}
    \label{lemma: cycle_PCTC}
Let $\chanmap: \linops(\hilmaparg{S_1}\otimes \hilmaparg A)\mapsto \linops(\hilmaparg{S_2}\otimes \hilmaparg A)$  be a CPTP map with an ancilla $A$, for all $\rho_{S_1}\in\linops(\hilmaparg{S_1})$ it holds
\begin{equation}
    \selfcycle_A(\chanmap_{S_2A|S_1A}) (\rho_{S_1})= \bra{\Omega^+}_{AA'}\left[\chanmap_{S_2A|S_1A}(\rho_{S_1}\otimes \ket{\Omega^+}\bra{\Omega^+}_{AA'})\right]\ket{\Omega^+}_{AA'},
\end{equation}
where $\ket{\Omega^+}_{AA'}$ is the un-normalised maximally entangled state on $A'$ and $A$, $\hilmaparg{A'}\cong \hilmaparg{A}$.
\end{lemma}
\begin{proof}
This lemma follows from existing results on loop composition~\cite{Portmann_2017}. We repeat the proof here for completeness:
\begin{equation}
\begin{split}
& \bra{\Omega^+}_{AA'}\left[\chanmap_{S_2A|S_1A}(\rho_{S_1}\otimes \ket{\Omega^+}\bra{\Omega^+}_{AA'})\right]\ket{\Omega^+}_{AA'}\\
    &=\sum_{i,j,k,l} (\bra{i}_A\otimes\bra{i}_{A'}) \left[\mathcal{E}_{S_2A|S_1A}(\rho_{S_1}\otimes \ketbraa{k}{l}_A) \otimes \ketbraa{k}{l}_{A'}\right] (\ket{j}_A\otimes\ket{j}_{A'})\\
    &=\sum_{i,j,k,l} \bra{i}_A \left[\mathcal{E}_{S_2A|S_1A}(\rho_{S_1}\otimes \ketbraa{i}{j}_A)\right]\ket{j}_A=  \selfcycle_A(\chanmap_{S_2A|S_1A}) (\rho_{S_1}).
\end{split}
\end{equation}
\end{proof}

The following proposition is a direct consequence of the results in Section~5.1 of~\cite{jean_2025}.

\begin{prop}[Proposition 5.2 in~\cite{jean_2025} --- Adapted]
\label{prop:fromJVS} 
    For any linear operator $\hat{C}\in\linops(\hilmaparg{S_{1}},\hilmaparg{S_{2}})$, there exists a linear operator $V\in\linops(\hilmaparg{S_{1}}\otimes\hilmaparg{A},\hilmaparg{S_{2}}\otimes\hilmaparg{A} )$ where the ancillary system $A$ is isomorphic to the input system tensored with a qubit, i.e., $\hilmaparg{A}\cong\hilmaparg{S_{1}}\otimes\mathbb{C}^2$ such that
    \begin{equation}
        \hat{C} = \alpha \bra{\Omega^+}_{AA'} V\ket{\Omega^+}_{AA'}.
    \end{equation}
    for some non-zero constant $\alpha\in\mathbb{R}^+$ and isometry $V_{S_2A|S_1A}\in \linops\left(\hilmaparg{S_1}\otimes \hilmaparg{A},\hilmaparg{S_2}\otimes \hilmaparg{A}\right)$\footnote{Technically, this is an isometry only if $\dim_{S_1}\leq\dim_{S_2}$. However, if this it not the case one can simply embed $S_2$ into a larger system and define the operator accordingly.}
\end{prop}

\begin{proof}
    The proof of this proposition is adapted from the results of Section 5.1 in~\cite{jean_2025}.

    An arbitrary operator $\hat{C}$ can be written as:
    \begin{equation}
        \hat{C} = \beta \sum_{i=1}^{d_{1}} a_i \ket{\psi_i}_{S_{2}}\bra{i}_{S_{1}}
    \end{equation}
    where $\beta\in\mathbb{R}^+$, and for all $i$, $a_i\in (0,1]$, and $\ket{\psi_i}_{S_{2}}$ are normalised states in $\hilmaparg{S_{2}}$. Indeed, in arbitrary bases we can write $\hat{C}=\sum_{i,j}c_{i,j} \ket{j}_{S_{2}}\bra{i}_{S_{1}}$ for some complex coefficients $c_{i,j}$. For every $i$, we can define $b_i'\ket{\psi'_i}_{S_{2}} = \sum_j c_{i,j}\ket{j}_{S_{2}}$, where $b_i'\in\mathbb{C}$ and the states $\ket{\psi'_i}_{S_{2}}$ are normalised. Equivalently, we can express each $b_i'$ in the polar form and absorb the phase into the state $\ket{\psi'_i}_{S_{2}}$, and obtain $b_i'\ket{\psi'_i}_{S_{2}}=b_i\ket{\psi_i}_{S_{2}}$, where now $b_i\in\mathbb{R}^+$ and $\ket{\psi_i}_{S_{2}}$ are still normalised. Hence, $\hat{C}=\sum_{i}b_i \ket{\psi_i}_{S_{2}}\bra{i}_{S_{1}}$. Let $\beta=\max \left\{b_i\right\}_{i}$, and define $a_i=\frac{b_i}{\beta}$, then the claim follows.

    If there exists a set $\{\hat{C}_i\}_i$ of operators such that $\hat{C}_i\ket{i}_{S_1}= a_i\ket{\psi_i}_{S_{2}}$ holds for each $i$, then the following circuit realises $\hat{C}$ up to multiplication by a real constant:
    \begin{equation}
      \mathcal{E} = \centertikz{
            \node[draw=blue, minimum height = 75pt, minimum width = 115pt,fill = blue!20] at (1.,-0.65) {};
            \node[draw=black, minimum height = 30pt, minimum width = 45 pt] (C) at (0,0) {$C_i$};
            \node[draw=black, ellipse,minimum height = 20pt, minimum width = 20pt] (i) at (2.5,0) {\small$i$};
            \draw (C) -- (i);
            \draw (0,-2.75) -- (0,-2);
            \draw (0,-2) to [in = 270, out=90] (2.5,-0.35);          
            \draw (2.5,0.35) to [in = 270, out=90] (2.5,0.75);
            \draw (2.5,0.75) to [in = 90, out=90] (3.5,0.75);
            \draw (3.5,0.75) -- (3.5,-2);
            \draw (2.5,-2) to [in = 270, out=270] (3.5,-2);
            \draw (2.5,-2) -- (2.5,-1.75);
            \draw (2.5,-1.75) to [in = 270, out=90] (0,-0.52);
            \draw (0,0.52) to [in = 270, out=90] (0,1.5);
       },       
    \end{equation}
    where $\mathcal{E}$ is formed by a sequential composition of a unitary SWAP operation, followed by a controlled operation $\sum_{i} (C_i)_S\otimes \ketbra{i}_A$. Thus, applying such map to $\ket{j}_{S_1}$ and looping the ancilla system we get
    \begin{equation}
            \begin{split}
                \bra{\Omega^+}_{A'A}\left(\sum_{i} (\hat{C}_i)_{S_1} \otimes \ketbra{i}_{A}\right) \circ \textup{SWAP}_{S_1A} \left(\ket{\Omega^+}_{AA'} \otimes\ket{j}_{S_1}\right)
                = \hat{C}_j\ket{j}_{S_1} =  a_j\ket{\psi_j}_{S_2}.
            \end{split}
            \end{equation}
    However, the map $\chanmap$ inside the blue box above is not CPTP but only CP. 
    Let us now construct such set $\{\hat{C}_i\}_i$ and extend the map to be CPTP. 

    Consider a qubit $Q$ and a set of unitaries $\{V_i\}_{i}$, defined as
    \begin{equation}
        V_i = \begin{pmatrix}
                 a_i & -\sqrt{1-a_i^2} \\
                \sqrt{1-a_i^2} & a_i  
            \end{pmatrix},
    \end{equation}
    and extend it to a controlled unitary operation on $S_1$ and $Q$:
    \begin{equation}
        \hat{C}_{V_i} = \ketbra{i}_{S_1}\otimes (V_i)_Q +\sum_{j\neq i} \ketbra{j}_{S_1}\otimes \id_Q.
    \end{equation}
    Then, we define an operator $W_i$ on $S_1$, such that $W_i\ket{i}_{S_1}=\ket{\psi_i}_{S_2}$, acting on other basis elements such that $\bra{j}W_i^*W_i\ket{k}=\delta_{j,k}$\footnote{Notice that this is only possible if $\dim_{S_1}\leq\dim_{S_2}$. However, one can always embed the space $\hilmaparg{S_2}$ into a larger space with dimension $\dim_{S_1}$ and the proof follows unchanged.
    }. Thus, $W_i$ is an isometry. 
    Further consider 
    \begin{equation}
        U_i = (W_i\otimes \id_Q)\hat{C}_{V_i}.
    \end{equation}
    Then, we have
    \begin{equation}
      \centertikz{
            \node[color=blue] at (-1.25,0.75){$U_i=$};
            \node[draw=blue, minimum height = 70pt, minimum width = 95pt,fill = blue!20] at (1.,0.75) {};
            \node[draw=black, ellipse,minimum height = 20pt, minimum width = 20pt] at (0,0) {\small$i$};
            \node[draw=black, minimum height = 20pt, minimum width = 30 pt] at (0,1.5) {$W_i$};
            \node[draw=black, minimum height = 20pt, minimum width = 29 pt] at (2,0) {$V_i$};
            \draw (0,-1.) -- (0,-.35);
            \draw (0,.35) -- (0,1.15);
            \draw (0,2.25) -- (0,1.85);
            \draw (0.35,0) -- (1.5,0);
            \draw (2,-0.35) -- (2,-0.35-0.25);
            \draw (2,-0.35-0.25) to [in =270, out=270](3,-0.35-0.25);
            \draw (3,-0.35-0.25) -- (3,2);
            \draw (2,2) to [in =90, out=90](3,2);
            \draw (2,0.35) -- (2,2);
            \node[anchor = west, text width = 8cm] at (3,0.75){
            \begin{equation*}
            \begin{split}
                &\sum_{i} (W_i)_{S_1} \bra{\Omega^+}_{Q'Q}(\ketbra{i}_{S_1}\otimes (V_i)_{Q}) \ket{\Omega^+}_{QQ'} \ket{j}_{S_1}\\
                &=\sum_{k} (W_j)_{S_1}\ket{j}_{S_1}\otimes
                 \bra{k}V_j\ket{k}\\
                 &= 2a_j \ket{\psi_j}_{S_2} = 2\hat{C}_j\ket{j}_{S_1}
            \end{split}
            \end{equation*}
            };
       } .        
    \end{equation}

    Thus, the isometry
    \begin{equation}
        V_{S_1QA}=\sum_i\left[(W_i)_{S_1}\otimes \id_Q\right](\hat{C}_{V_i})_{QS_1}\otimes \ketbra{i}_A,
    \end{equation}
    satisfies
    \begin{equation}
      \begin{split}
          &\bra{\ubellstate}_{AQA'Q'}V\ket{\ubellstate}_{AQA'Q'}\ket{i}_{S_1}\\         
        &=\bra{\ubellstate}_{QQ'}\left[(W_i)_{S_1}\otimes \id_Q\right](\hat{C}_{V_i})_{QS_1}\ket{\ubellstate}_{QQ'}\ket{i}_{S_1}\\
        &=2a_i\ket{\psi_i}_{S_2}=\frac{2}{\beta} \hat{C}\ket{i}_{S_1}
      \end{split},
    \end{equation}
    since the above holds for all elements of an orthonormal basis $\ket{i}$ of the input, it holds for all input states (by linearity) and the equation in the result statement follows from above, up to grouping together the ancilla systems and relabelling $\alpha=\frac{\beta}{2}$. 
\end{proof}

We now recall the lemma and finally provide its proof.
\TNtoCMgen*
\begin{proof}
For simplicity let us leave out the index $\vertname$ and consider a density operator $P\in\linops(\hilmaparg{S_1}\otimes\hilmaparg{S_2})$, where we denoted the input and output spaces appearing in the Choi-Jamiołkowski respectively as $S_1$ and $S_2$.

    Since $P$ is a density operator, we can write its spectral decomposition as $P=\sum_k p_k \ketbra{\Psi_k}$ where $\sum_kp_k =1$ and $p_k \geq 0$ and $\braaket{\Psi_k}{\Psi_{k'}} = \delta_{k,k'}$.
    Then using~\cref{lemma:inv_choi}, we have
    \begin{align}
    \begin{split}
    \CJin_{S_2|S_1}\left(P\right)&=\CJin_{S_2|S_1}\left(\sum_k p_k \ket{\Psi_k}\bra{\Psi_k}\right)\\
&=d_{S_1}^2 \sum_k p_k\left(\CJp_{S_2|S_1}(\ket{\Psi_k})\right)\otimes \left(\CJp_{S_2|S_1}(\ket{\Psi_k})\right)^\dagger\\
&\equiv d_{S_1}^2 \sum_k p_k C_k\otimes \left(C_k\right)^\dagger,
    \end{split}
\end{align}
where we renamed the linear operator from $\hilmap_{S_1}$ to $\hilmap_{S_2}$ as $C_k:=\CJp_{S_2|S_1}(\ket{\Psi_k})$. 

Each $C_k$ is a non-zero linear operator from $\hilmaparg{S_1}$ to $\hilmaparg{S_2}$, thus from~\cref{prop:fromJVS} it follows that for any $k$, there exists an isometry $V_k$ and an ancilla $\hilmaparg{A}$ (isomorphic to $\mathbb{C}^2\otimes\hilmaparg{S_1}$) and a real constant $\beta_k >0$, satisfying
\begin{equation}
\label{eq: PCTC_unitaries}
    C_k=\beta_k\bra{\ubellstate}_{AA'} V_k\ket{\ubellstate}_{AA'}.
\end{equation}

Thus, we have
\begin{equation}
\label{eq:frompuretomixed1}
\begin{split}
    \CJin_{S_2|S_1}\left(P\right)&= d_{S_1}^2 \sum_k p_kC_k\otimes \left(C_k\right)^\dagger\\
    &\equiv d_{S_1}^2 \sum_k p_k \beta_k^2 \bra{\Omega^+}_{AA'} V_k\ket{\Omega^+}_{AA'}\otimes \bra{\Omega^+}_{AA'} (V_k)^\dagger\ket{\Omega^+}_{AA'}\\
    &= d_{S_1}^2 \sum_k  \bra{\Omega^+}_{AA'} E_k\ket{\Omega^+}_{AA'}\otimes \bra{\Omega^+}_{AA'} (E_k)^\dagger\ket{\Omega^+}_{AA'}
\end{split}
\end{equation}
where we defined the operators $E_k = \sqrt{p_k}\beta_k V_k$. 
Let us then consider the map
\begin{equation}
    \chanmap_{S_2A|S_1A}(\rho_{S_1A})=\gamma^{-1} \sum_k E_k(\rho_{S_1A}) (E_k)^\dagger,
\end{equation}
where $\gamma = \sum_k p_k \beta_k^2$, so that $\chanmap_{S_2A|S_1A}$ is CPTP. Indeed, we have
\begin{equation}
    \gamma^{-1}\sum_k (E_k)^\dagger E_k = \gamma^{-1} \sum_k p_k \beta_k^2 (V_k)^\dagger V_k =  \gamma^{-1} \sum_k p_k \beta_k^2 \id_{S_1A} = \id_{S_1A}.
\end{equation}

Introducing this map in~\cref{eq:frompuretomixed1}, we get that for all $\rho_{S_1}\in\linops(\hilmaparg{S_1})$
\begin{equation} 
\label{eq:cycleandketbra}
\begin{split}
    \CJin_{S_2|S_1}\left(P\right)(\rho_{S_1})&= d_{S_1}^2 \bra{\Omega^+}_{AA'} \sum_k  E_k\left(\rho_{S_1} \otimes\ketbra{\Omega^+}_{AA'}\right)(E_k)^\dagger\ket{\Omega^+}_{AA'}\\
    &= d_{S_1}^2\gamma \sum_{i,j} \bra{i}_{A}\chanmap_{S_2A|S_1A}\left(\rho_{S_1} \otimes\ketbraa{i}{j}_{A}\right)\ket{j}_{A}\\
    &= \alpha \selfcycle_{A}\left(\chanmap_{S_2A|S_1A}\right)
    ,
\end{split}
\end{equation}
where we defined $\alpha = d_{S_1}^2\gamma \geq 0$
 We remark that this is analogous to the construction of Section 5.3 of~\cite{jean_2025}.

 Thus, for all $\vertname\in\vertsetud$, the density operator $\projarg{\vertname}\in\linops\left(\hilmaparg{\inedges{\vertname}}\otimes\hilmaparg{\outedges{\vertname}}\right)$, satisfies
 \begin{equation}
     \CJin_{\outedges{\vertname}|\inedges{\vertname}}\left(\projarg{\vertname}\right)=\alpha_\vertname\selfcycle_{A_\vertname}(\chanmaparg{\vertname}),
 \end{equation}
for a real constant $\alpha_v >0$ and a CPTP map $\chanmaparg{\vertname}:\linops(\hilmaparg{\inedges{\vertname}}\otimes\hilmaparg{A_v})\mapsto\linops(\hilmaparg{\outedges{\vertname}}\otimes\hilmaparg{A_v})$.
\end{proof}

\subsection*{Proofs of the~\cref{lemma:composition_maps,lemma: compositiongen,lemma: TNtoCMgeneral_correl}}
\compmaps*
\begin{proof}
    The tensor network obtained after applying $\mapCMtoTN$ is defined on the graph $\graphnameud$ with the same vertex set of $\graphnamedir$, $\vertsetud=\vertsetdir$, and same edges but without direction, i.e., if $\edgename=(v,w)$ is in $\edgesetdir$, then the set $\bar{\edgename}=\{v,w\}$ is in $\edgesetud$. Then, we have
    \begin{equation}
        \mapCMtoTN(\cm)=\tn = \left(\graphnameud,\{\hilmaparg{e}\}_{\edgename\in\edgesetdir}, \{\CJ_{\outedges{\vertname}|\inedges{\vertname}}(\chanmaparg{\vertname})\}_{\vertname\in\vertsetud}\right)
    \end{equation}
    where the general map $\mathcal{T}^{\vertname}$ is a CPTP map as $\overtset=\emptyset$. Then, by applying
    $\mapTNtoCM_{D_{\graphnamedir}}$ we get:
    \begin{equation}
        \begin{split}
            \mapTNtoCM_{D_{\graphnamedir}}\circ\mapCMtoTN(\cm) =& \mapTNtoCM_{D_{\graphnamedir}}\left(\graphnameud,\{\hilmaparg{e}\}_{\edgename\in\edgesetdir}, \{\CJ_{\outedges{\vertname}|\inedges{\vertname}}(\chanmaparg{\vertname})\}_{\vertname\in\vertsetud}\right)\\
            =&\left(\graphnamedir'_{D_{\graphnamedir}},\{\hilmaparg{e}\}_{\edgename\in\edgesetud}, \{\CJ^{-1}_{\outedges{\vertname}|\inedges{\vertname}}\CJ_{\outedges{\vertname}|\inedges{\vertname}}(\chanmaparg{\vertname})\}_{\vertname\in\vertsetud}\right)\\
            =&\left(\graphnamedir,\{\hilmaparg{e}\}_{\edgename\in\edgesetdir}, \{\chanmaparg{\vertname}\}_{\vertname\in\vertsetud}\right) = \cm,
        \end{split} 
    \end{equation}
    where we used that by definition $\graphnamedir'_{D_{\graphnamedir}}$ is the same as $\graphnamedir$.
    The proof follows analogously for the other direction.
\end{proof}

\compmapsgen*
\begin{proof}
Points $1.$ and $2.$ follow immediately from~\cref{def: TNtoCMgeneral,def:mapCMtoTNv2}. For proving $3.$, let us first define the set of vertices for which~\cref{eq: trace_condition} is not satisfied, i.e., $\vertname\in\vertset^\neg$ if and only if
    \begin{equation}
        \Tr_{\outedges{\vertname}}[\projarg{\vertname}] \neq \frac{\id_{\inedges{\vertname}}}{d_{\inedges{\vertname}}}.
    \end{equation}
For all vertices $\vertname\in\vertsetud\setminus\vertset^\neg$, we have $\edgesetud(v)=\emptyset$ as $\incedges{v}_{\graphnameud}=\incedges{v}_{\graphnameud^{\circlearrowleft}}$, the TP condition is satisfied by definition for $P_v$ with the chosen directions. Thus, $\CJin_{\outedges{\vertname}|\inedges{\vertname}} (\projarg{\vertname})$ is CPTP and we have:
\begin{equation}
    \projarg{v}^{\circlearrowleft} = \CJ_{\outedges{\vertname}|\inedges{\vertname}}\circ \CJin_{\outedges{\vertname}|\inedges{\vertname}} (\projarg{\vertname}) = \projarg{\vertname}.
\end{equation}

For all vertices $\vertname\in\vertset^\neg$, the set $\edgesetud(v)$ corresponds to the ancillary systems introduced such that (see the proof of~\cref{prop:fromJVS} in~\Cref{app:proofs}):
\begin{equation}
    \selfcycle_{A_\vertname}(\chanmaparg{\vertname})=\frac{1}{\alpha_\vertname}\CJin_{\outedges{\vertname}|\inedges{\vertname}}\left(\projarg{\vertname}\right),
\end{equation}
thus we have:
\begin{equation}
    \begin{split}
        &\bra{\ubellstate}_{A_vA'_v}\projarg{\vertname}^{\circlearrowleft}\ket{\ubellstate}_{A_vA'_v}\\
        =& \bra{\ubellstate}_{A_vA'_v}\CJ_{\outedges{\vertname}A_v|\inedges{\vertname}A_v}(\chanmaparg{\vertname})\ket{\ubellstate}_{A_vA'_v}\\
        =&\frac{1}{d_{A_v}} \bra{\ubellstate}_{A_vA'_v}(\chanmaparg{\vertname})(\ketbra{\ubellstate}_{A_vA'_v}\otimes \ketbra{\bellstate}_{\incedges{\vertname}_{\graphnameud}})\ket{\ubellstate}_{A_vA'_v}\\
        =&\frac{1}{d_{A_v}} \selfcycle_{A_v}(\chanmaparg{\vertname})\left(\ketbra{\bellstate}_{\incedges{\vertname}_{\graphnameud}}\right)\\
        =&\frac{1}{d_{A_v}} \CJ_{\outedges{\vertname}|\inedges{\vertname}}\left[\selfcycle_{A_v}(\chanmaparg{\vertname})\right]\\
        =&\frac{1}{d_{A_v}\alpha_v} \CJ_{\outedges{\vertname}|\inedges{\vertname}}\circ \CJin_{\outedges{\vertname}|\inedges{\vertname}}\left(\projarg{\vertname}\right)=\frac{1}{d_{A_v}\alpha_v} \projarg{\vertname},
    \end{split}
\end{equation}
where between the second and third line we used the definition of the Choi-Jamiolkowski isomorphism and that the space $\hilmaparg{\incedges{\vertname}_{\graphnameud}}$ is isomorphic to  $\hilmaparg{\inedges{\vertname}_{\graphnameud}}\otimes \hilmaparg{\outedges{\vertname}_{\graphnameud}}$, between the third and fourth line we used~\cref{lemma: cycle_PCTC}, between the fourth and fifth line we used the definition of the Choi-Jamiolkowski isomorphism and finally~\cref{eq: cycle_TNtoCMgen} to conclude the proof.
\end{proof}

\contselfcyclegen*
\begin{proof}
    Let us first consider the case where $\mapCMtoTN(\cm)=\tn$ and let us recall the spaces involved in the maps and states: the marginalised total map $\etot_{\edgeset_{\textup{out}}|\edgeset_{\textup{in}}}\in\linops(\linops(\hilmaparg{\edgeset_{\textup{in}}}),\linops(\edgeset_{\textup{out}}))$, and the total and link states $\rhop\in\linops(\hilmaparg{\edgeset_{\textup{out}}}\otimes\hilmaparg{\edgeset'_{\textup{in}}})$, $\linkstate\in\hilmaparg{\edgeset_{\textup{out}}}\otimes\hilmaparg{\edgeset'_{\textup{in}}}$. If $\tn=\mapCMtoTN(\cm)$ by~\cref{def:mapCMtoTNv2}, we have 
    \begin{equation}
        \rhop = \left(\etot_{\edgeset_{\textup{out}}|\edgeset_{\textup{in}}}\otimes \mathcal{I}_{\edgeset'_{\textup{in}}}\right)\ketbra{\bellstate}_{\edgeset_{\textup{in}},\edgeset'_{\textup{in}}}
    \end{equation}
    and the link state is
     \begin{equation}
        \ket{\bellstate}_{\edgeset_{\textup{in}},\edgeset'_{\textup{in}}} =\frac{1}{\sqrt{d_E}} \sum_{i} \ket{i}_{\edgeset_{\textup{in}}}\otimes \ket{i}_{\edgeset'_{\textup{in}}}.
    \end{equation}
    Thus, we have:
    \begin{equation}
        \begin{split}
             &\bra{L}\rhop\ket{L} \\
             &= \frac{1}{d_E} \sum_{i,j,k,l} \left(\bra{i}_{\edgeset_{\textup{out}}}\otimes\bra{i}_{E'_{\textup{in}}}\right)\left(\etot_{E_{\textup{out}}|E_{\textup{in}}}(\ketbraa{j}{k}_{E_{\textup{in}}})\otimes\ketbraa{j}{k}_{E'_{\textup{in}}}\right)\left(\ket{l}_{E_{\textup{out}}}\otimes\ket{l}_{E'_{\textup{in}}}\right)\\
             &= \frac{1}{d_E} \sum_{j,k} \bra{j}_{E_{\textup{out}}}\etot_{E_{\textup{out}}|E_{\textup{in}}}(\ketbraa{j}{k}_{E_{\textup{in}}})\ket{k}_{E_{\textup{out}}}\\
             &= \frac{1}{d_E} \selfcycle(\etot_{E_{\textup{out}}|E_{\textup{in}}}).\\
        \end{split}
    \end{equation}

    In the case of $\cm=\mapTNtoCMgen_D(\tn)$, we can define the CP map $\tilde{\chanmap}^v\equiv\selfcycle_{A_\vertname}(\chanmaparg{\vertname})$, which satisfies
\begin{equation}
\CJ_{\outedges{\vertname}|\inedges{\vertname}}\left(\tilde{\chanmap}^v\right)=\CJ_{\outedges{\vertname}|\inedges{\vertname}}\left(\selfcycle_{A_\vertname}(\chanmaparg{\vertname})\right) = \CJ_{\outedges{\vertname}|\inedges{\vertname}}\circ \CJin_{\outedges{\vertname}|\inedges{\vertname}} \left(\frac{P_v}{\alpha_v}\right)=
\frac{\projarg{\vertname}}{\alpha_\vertname}, 
\end{equation}
where we used that $\selfcycle_{A_\vertname}(\chanmaparg{\vertname})=\frac{1}{\alpha_\vertname}\CJin_{\outedges{\vertname}|\inedges{\vertname}}\left(\projarg{\vertname}\right)$ (see~\cref{eq: cycle_TNtoCMgen}).

Let us define $\tilde{\etot}=\bigotimes_{\vertname\in\vertsetud}\tilde{\chanmap}^v$, then we have
\begin{equation}
    \rhop = \tilde{\beta} \left(\tilde{\etot}_{\edgeset_{\textup{out}}|\edgeset_{\textup{in}}}\otimes \mathcal{I}_{\edgeset'_{\textup{in}}}\right)\ketbra{\bellstate}_{\edgeset_{\textup{in}},\edgeset'_{\textup{in}}},
\end{equation}
where we defined $\tilde{\beta}=\left(\prod_v \alpha_v\right)$,
and, the link state
    \begin{equation}
        \linkstate_{\edgeset_{\textup{out}},\edgeset'_{\textup{in}}} = \sum_{j} \ket{j}_{\edgeset_{\textup{out}}}\otimes\ket{j}_{\edgeset'_{\textup{in}}}.
    \end{equation}
    Thus, following the same steps as above, we have $\bra{L}\rhop\ket{L} = \frac{\tilde{\beta}}{d_E} \selfcycle(\tilde{\etot}_{E_{\textup{out}}|E_{\textup{in}}})$, thus introducing $\beta=\frac{\tilde{\beta}}{d_E}$, we have:
    \begin{equation}
        \begin{split}
             \bra{L}\rhop\ket{L} &= \beta \selfcycle(\tilde{\etot}_{E_{\textup{out}}|E_{\textup{in}}})\\
             &= \beta \selfcycle(\bigotimes_{\vertname\in\vertset}(\selfcycle_{A_\vertname}(\chanmaparg{\vertname}))\\
             &= \beta\selfcycle(\bigotimes_{\vertname\in\vertset}\chanmaparg{\vertname}) \equiv \beta\selfcycle(\etot),
        \end{split}
    \end{equation}
    where we recall that the maps $\chanmaparg{\vertname}$ are the CPTP maps associated to vertices in the causal model.
\end{proof}
\subsection*{Proof of~\cref{lem:signalling_CI}}
\ciassignalling*
\begin{proof}
Let $\cm=\mapTNtoCMgen_D(\tn)$ (the case where $\tn=\mapCMtoTN(\cm)$ follows straightforwardly from this proof and~\cref{lemma: compositiongen}).
   Consider a time-reversible intervention labelled by $a\in\outcomemaparg{A}$ and $b\in\outcomemaparg{B}$, by~\cref{def:statisticsint} we have that the observed statistic of the intervention is
    \begin{equation}
    \label{eq: ciassignalling1}
       \prob(y|a,b)_I = \frac{\selfcycle\left((\instrmapsetting{\labarg{B}}{y|b}\otimes \mathcal{U}_{\labarg{A}}^{a})\circ \bar{\etot}_{\edgeset_{\textup{out}}A_V|\edgeset_{\textup{in}}A_V}\right)}{\sum_{y'}\selfcycle\left((\instrmapsetting{\labarg{B}}{y'|b}\otimes \mathcal{U}_{\labarg{A}}^{a})\circ \bar{\etot}_{\edgeset_{\textup{out}}A_V|\edgeset_{\textup{in}}A_V}\right)},
   \end{equation}
   where we defined $\bar{\etot}:\linops(\hilmaparg{\edgeset_{\textup{in}}}\otimes\hilmaparg{ A_V})\to\linops(\hilmaparg{\edgeset_{\textup{out}}}\otimes\hilmaparg{A_V})$, as the marginalised total map of $\cm$, $\hilmaparg{\edgeset_{\textup{in}}}\cong \hilmaparg{\edgeset_{\textup{out}}}\cong\hilmaparg{\edgesetud}$ and we grouped the ancillas arising according to~\cref{def: TNtoCMgeneral} in the system $A_V$ with $\hilmaparg{A_V}=\bigotimes_{\vertname\in\vertsetud}\hilmaparg{A_\vertname}$.
   
On the other hand, the total state $\rhop\in\linops(\hilmaparg{\edgeset_{\textup{out}}}\otimes\hilmaparg{\edgeset_{\textup{in}}'})$ of the tensor network satisfies (\cref{lemma: TNtoCMgeneral,def: TNtoCMgeneral}):
\begin{equation}
\begin{split}
    \CJin_{\edgeset_{\textup{out}}|\edgeset_{\textup{in}}}\left(\rhop\right) &= \bigotimes_{\vertname\in\vertsetud}\CJin_{\outedges{\vertname}|\inedges{\vertname}}\left(\projarg{\vertname}\right)\\
    &=\bigotimes_{\vertname}\alpha_\vertname\selfcycle_{A_\vertname}(\chanmaparg{\vertname}_{\outedges{\vertname}A_v|\inedges{\vertname}A_v})\\
    &=(\prod_{v}\alpha_v)\selfcycle_{A_\vertset}(\bigotimes_{\vertname}\chanmaparg{\vertname}_{\outedges{\vertname}A_v|\inedges{\vertname}A_v})\\
    &=(\prod_{v}\alpha_v)\selfcycle_{A_\vertset}(\bar{\etot}_{\edgeset_{\textup{out}}A_V|\edgeset_{\textup{in}}A_V}) \equiv \alpha \selfcycle_{A_\vertset}(\bar{\etot}_{\edgeset_{\textup{out}}A_V|\edgeset_{\textup{in}}A_V}),
\end{split}
\end{equation}
where we introduced the constant $\alpha=(\prod_{v}\alpha_v)$ and grouped the ancillas in the system $A_V$ with Hilbert space $\hilmaparg{A_V}=\bigotimes_{\vertname\in\vertsetud}\hilmaparg{A_\vertname}$. Thus, using the fact that the Choi-Jamiołkowski mapping is an isomorphism, we have:
\begin{equation}
\label{eq:proofcorr1}
\begin{split}
    \frac{\rhop}{\alpha} &= \; \CJ_{\edgeset_{\textup{out}}|\edgeset_{\textup{in}}}\left[\selfcycle_{A_\vertset}(\bar{\etot}_{\edgeset_{\textup{out}}A_V|\edgeset_{\textup{in}}A_V})\right]\\
    &= \selfcycle_{A_\vertset}(\bar{\etot}_{\edgeset_{\textup{out}}A_V|\edgeset_{\textup{in}}A_V}) (\ketbra{\bellstate}_{\edgeset_{\textup{in}}\edgeset_{\textup{in}}'}).
\end{split}
\end{equation}
Let us now consider the correlation function where we let the unitary and Hermitian operator act on the $\edgeset_{\textup{out}}$ subsystem of $\rhop$ (see the discussion in~\Cref{app:clarification on causal influence} for further details):
\begin{equation}
    \begin{split}
        &\bra{L}(U^a_{\labarg{A}_\textup{out}}\otimes O^{y|b}_\labarg{B_\textup{out}})\rhop ({U^a_{\labarg{A_\textup{out}}}}^\dagger\otimes {O^{y|b}_\labarg{B_\textup{out}}}^\dagger)\ket{L}\\
        &\equiv \bra{L}(\mathcal{U}^a_{\labarg{A}_\textup{out}}\otimes \mathcal{O}^{y|b}_\labarg{B_\textup{out}})(\rhop) \ket{L}\\
        &= \alpha\;\bra{L}(\mathcal{U}^a_{\labarg{A}_\textup{out}}\otimes \mathcal{O}^{y|b}_\labarg{B_\textup{out}})\left[\selfcycle_{A_\vertset}(\bar{\etot}_{\edgeset_{\textup{out}}A_V|\edgeset_{\textup{in}}A_V}) (\ketbra{\bellstate}_{\edgeset_{\textup{in}}\edgeset_{\textup{in}}'})\right] \ket{L}\\
        &= \alpha\;\bra{\ubellstate}_{\edgeset_{\textup{out}}\edgeset_{\textup{in}}'}(\mathcal{U}^a_{\labarg{A}_\textup{out}}\otimes \mathcal{O}^{y|b}_\labarg{B_\textup{out}})\left[\selfcycle_{A_\vertset}(\bar{\etot}_{\edgeset_{\textup{out}}A_V|\edgeset_{\textup{in}}A_V}) (\ketbra{\bellstate}_{\edgeset_{\textup{in}}\edgeset_{\textup{in}}'})\right]\ket{\ubellstate}_{\edgeset_{\textup{out}}\edgeset_{\textup{in}}'}\\
        &= \frac{\alpha}{d_{\edgesetud}} \selfcycle_{\edgesetud}\left[(\mathcal{U}^a_{\labarg{A}_\textup{out}}\otimes \mathcal{O}^{y|b}_\labarg{B_\textup{out}})\selfcycle_{A_\vertset}(\bar{\etot}_{\edgeset_{\textup{out}}A_V|\edgeset_{\textup{in}}A_V})\right]\\
        &= \frac{\alpha}{d_{\edgesetud}} \selfcycle_{\edgesetud}\circ\selfcycle_{A_\vertset}\left[(\mathcal{U}^a_{\labarg{A}_\textup{out}}\otimes \mathcal{O}^{y|b}_\labarg{B_\textup{out}})\circ\bar{\etot}_{\edgeset_{\textup{out}}A_V|\edgeset_{\textup{in}}A_V}\right] \\
        &= \frac{\alpha}{d_{\edgesetud}} \selfcycle\left[(\mathcal{U}^a_{\labarg{A}_\textup{out}}\otimes \mathcal{O}^{y|b}_\labarg{B_\textup{out}})\circ\bar{\etot}_{\edgeset_{\textup{out}}A_V|\edgeset_{\textup{in}}A_V}\right],
    \end{split}
    \end{equation}
   where we used~\cref{eq:proofcorr1} between the second and third line, we expanded the link state according to~\cref{def:linktotalstate} between the third and fourth line, used~\cref{lemma: cycle_PCTC} between the fourth and fifth line (notice that the Hilbert spaces associated to $\edgeset_{\textup{out}}$ and $\edgeset_{\textup{in}}$ are isomorphic to $\hilmaparg{\edgesetud}$, so the cycle composition between them is well-defined and we denote it as $\edgesetud$), between the fifth and sixth line we used that the cycle composition on the ancillas $A_V$ commute with the unitary and Hermitian operations which are solely applied on subsystems of $\edgeset_{\textup{out}}$, finally we used that the two subsequent cycle compositions are equal to one global cycle composition.
    
    Thus, in~\cref{eq: ciassignalling1} we have
    \begin{equation}
    \label{eq:establishequallity}
    \begin{split}
        M(U^a:O^{y|b}|\mathcal{O}^b)&= \frac{\bra{L}(U^a_{\labarg{A}_\textup{out}}\otimes O^{y|b}_\labarg{B_\textup{out}})\rhop ({U^a_{\labarg{A_\textup{out}}}}^\dagger\otimes {O^{y|b}_\labarg{B_\textup{out}}}^\dagger)\ket{L}}{\sum_{y'}\bra{L}(U^a_{\labarg{A}_\textup{out}}\otimes O^{y'|b}_\labarg{B_\textup{out}})\rhop ({U^a_{\labarg{A_\textup{out}}}}^\dagger\otimes {O^{y'|b}_\labarg{B_\textup{out}}}^\dagger)\ket{L}}\\
        &=\frac{\alpha}{d_\edgeset}\frac{d_\edgeset}{\alpha}\frac{\selfcycle\left[(\mathcal{U}^a_{\labarg{A}_\textup{out}}\otimes \mathcal{O}^{y|b}_\labarg{B_\textup{out}})\circ\bar{\etot}_{\edgeset_{\textup{out}}A_V|\edgeset_{\textup{in}}A_V}\right]}{\sum_{y'}\selfcycle\left[(\mathcal{U}^a_{\labarg{A}_\textup{out}}\otimes \mathcal{O}^{y'|b}_\labarg{B_\textup{out}})\circ\bar{\etot}_{\edgeset_{\textup{out}}A_V|\edgeset_{\textup{in}}A_V}\right]} = \prob(y|a,b)_I
    \end{split}
    \end{equation}
    where we denoted the instrument $\mathcal{O}^b=\{O^{y'|b}\}_{y'\in\outcomemaparg{Y}}$. Notice that the constant $\alpha$ is independent of the settings and outcome, $a$, $b$ and $y$, i.e., independent of the inserted unitary and instrument, but only depends on the total map $\bar{\etot}$ of the background causal model.
    There is no-signalling from $\labarg{A}$ to $\labarg{B}$ if and only if 
    \begin{equation}
         \prob(y|a,b)_I= \prob(y|a',b)_I \quad \forall y\in\outcomemaparg{Y},\;a,a' \in\outcomemaparg{A},\; b\in\outcomemaparg{B}
    \end{equation}
    i.e., if and only if
    \begin{equation}
        M(U^a:O^{y|b}|\mathcal{O}^b)=M(U^{a'}:O^{y|b}|\mathcal{O}^b) \quad \forall y\in\outcomemaparg{Y},\;a,a' \in\outcomemaparg{A},\; b\in\outcomemaparg{B}.
    \end{equation}
    Since $a$ labels the set of unitary channels and $b$ the set of Hermitian instruments, we can conclude that there is no-signalling if and only if the causal influence is zero.
\end{proof}
\subsection*{Proof of~\cref{lem:isomchan perfect tn}}
\perfecttnch*
\begin{proof}
    Let us write the state $\ket{\Psi}_{S\bar{S}}$ compactly as
    \begin{equation}
        \ket{\Psi}_{S\bar{S}} = \sum_{i_S,i_{\bar{S}}} T_{i_S,i_{\bar{S}}} \ket{i_S}_S \otimes \ket{i_{\bar{S}}}_{\bar{S}}
    \end{equation}
    and for clarity let us relabel the Hilbert space associated to $S$ as $S'$, where $\hilmaparg{S'}\cong\hilmaparg{S}=\textup{span}\{\ket{i_S}\}_{i_S}$, i.e., the computational basis elements are labelled by the same indices in $S$. We have:
    \begin{equation}
        \begin{split}
             &\CJin_{\bar{S}|S}(\ketbra{\Psi}_{S'\bar{S}}) (\rho_S)\\
             &=d_S \Tr_{S'}\left[ \sum_{i_S,j_S} \ketbraa{i_S}{j_S}_{S'} \bra{j_S}_S\rho_S\ket{i_S}_S \ketbra{\Psi}_{S'\bar{S}}\right]\\
             &=d_S \sum_{i_S,j_S} \bra{j_S}_S\rho_S\ket{i_S}_S \Tr_{S'} \left[ \ketbraa{i_S}{j_S}_{S'} \ketbra{\Psi}_{S'\bar{S}} \right]\\
             &= d_S \sum_{i_S,j_S} \bra{j_S}\Psi\rangle_{\bar{S}} \bra{j_S}_S\rho_S\ket{i_S}_S  \bra{\Psi}i_S\rangle_{\bar{S}}\\
             &= d_S \left(\sum_{j_S}\bra{j_S}\Psi\rangle_{\bar{S}} \bra{j_S}_S\right)\rho_S \left(\sum_{i_S}\ket{i_S}_S\bra{\Psi}i_S\rangle_{\bar{S}} \right).
        \end{split}
    \end{equation}
    We have
    \begin{equation}
        \sum_{j_S}\bra{j_S}\Psi\rangle_{\bar{S}} \bra{j_S}_S = \sum_{i_S, i_{\bar{S}}, j_S} T_{i_S,i_{\bar{S}}} \bra{j_S}_S\left(\ket{i_S}_S\otimes \ket{i_{\bar{S}}}_{\bar{S}}\right) \bra{j_S}_S  = \sum_{i_S, i_{\bar{S}}} T_{i_S,i_{\bar{S}}}\ket{i_{\bar{S}}}_{\bar{S}}\bra{i_S}_S =\frac{V}{C},
    \end{equation}
    for some constant $C$ and $V$ isometry.
    Thus, 
    \begin{equation}
             \frac{C^2}{d_S}\CJin_{\bar{S}|S}(\ketbra{\Psi}_{S\bar{S}}) (\rho_S)=  V\rho_S V^\dagger=\mathcal{V}_{\bar{S}|S}(\rho),
    \end{equation}
    which concludes the proof up to relabelling the constant.
\end{proof}
\section{Further clarifications on tensor 
networks}
\paragraph{Why density operators?}
We argued that if a tensor is characterised by a collection of real positive numbers summing to one, we can trivially embed it in a density operator. Now we wish to show that an arbitrary tensor can be embedded in a density operator up to an overall constant, which does not influence the results of normalised correlation functions.

Any tensor can be trivially embedded in a pure, possibly not normalised, state as follows:
\begin{equation}
    \{T_{i_1,i_2,\dots, i_n}\}_{i_1,\dots, i_n} \mapsto \ket{\psi_T} = \sum_{i_1,\dots, i_n} T_{i_1,i_2,\dots, i_n} \ket{i_1,i_2,\dots, i_n}
\end{equation}

Tensor network contractions can then be obtained through taking the inner product of the network with a maximally entangled state on the contracted systems. For instance, let
$\{c_i\}_i \mapsto \sum_i c_i \ket{i}:= \ket{\psi_c}_1$ and $\{M_{ij}\}_{ij}\mapsto \sum_{ij}M_{ij}\ket{ij}:= \ket{\psi_M}_{23}$, where $1$, $2$ and $3$ are arbitrary system labels. The tensor network contraction obtained from contracting the systems $1$ and $2$, $\sum_{i} c_iM_{ij} :=b_j$, is then associated to the state $\ket{\psi_{cM}}_{3}=\sum_{j} b_j\ket{j}_3$, which is equivalently obtained through
\begin{equation}
    \bra{\bellstate}_{1,2}\left(\ket{\psi_c}_1\otimes\ket{\psi_M}_{23}\right) = \sum_j b_j\ket{j}_3 = \ket{\psi_{cM}}.
\end{equation}

\paragraph{On which Hilbert spaces are the operators $U_{\region{A}}$ and $O_{\region{B}}$ acting?}
\label{app:clarification on causal influence}
The doubling of Hilbert spaces necessary to define the total state and link state leads to an ambiguity on where the operators $U_{\region{A}}$ and $O_{\region{B}}$ act. Indeed, the total and link states are defined on the Hilbert space $\Htot=\bigotimes_{\vertname\in\vertsetud}\bigotimes_{\edgename\in\incedges{\vertname}}\hilmap_{(\edgename,\vertname)}\cong \hilmaparg{\edgesetud}^{\otimes 2}$ thus there are two subspaces isomorphic to each $\hilmaparg{\region{A}}$ and $\hilmaparg{\region{B}}$. \Cref{def: causal influence old} does not specify on which of these spaces the unitary and Hermitian operator act on. 

We remark that for determining whether there is causal influence or not one can arbitrarily choose on which space to apply the unitary as long as the choice is consistent in all evaluations of the quantity $M$ for different unitaries and Hermitian operators. To show that, let us denote the two isomorphic spaces associated to the edges in $\region{A}$ as $\hilmaparg{\region{A}_1}$ and $\hilmaparg{\region{A}_2}$ and for $E'\subseteq E$ we denote
\begin{equation}
    \ket{\ubellstate}_{E'_1,E'_2} = \bigotimes_{\edgename\in E'} \ket{\ubellstate}_{e_1,e_2},
\end{equation}
where $e_1$ and $e_2$ stand for the two spaces isomorphic to $e$ in $\Htot$.
Since the link state is defined by the tensor product of un-normalised Bell states, we have for any unitary $U_\region{A}\in\linops(\hilmaparg{\region{A}})$:
\begin{equation}
\begin{split}
    U_{\region{A}_1}\linkstate &\equiv (U_{\region{A}_1}\otimes \id_{\region{A}_2})\left(\ket{\ubellstate}_{\region{A}_1,\region{A}_2} \otimes \bigotimes_{\edgename\in\edgesetud\setminus\region{A}} \ket{\ubellstate}_{(\edgesetud\setminus\region{A})_1,(\edgesetud\setminus\region{A})_2}\right)\\
    &=(\id_{\region{A}_1}\otimes U^T_{\region{A}_2})\left(\ket{\ubellstate}_{\region{A}_1,\region{A}_2} \otimes \bigotimes_{\edgename\in\edgesetud\setminus\region{A}} \ket{\ubellstate}_{(\edgesetud\setminus\region{A})_1,(\edgesetud\setminus\region{A})_2}\right)= U^T_{\region{A}_2}\linkstate.
\end{split}
\end{equation}
Similarly for any Hermitian operator $O_\region{B}$, by denoting the doubling on $\region{B}$ as $\hilmaparg{\region{B}_1}$ and $\hilmaparg{\region{B}_2}$, we have
\begin{equation}
\begin{split}
    O_{\region{B}_1}\linkstate = O^T_{\region{B}_2}\linkstate.
\end{split}
\end{equation}
Remark also that if $U$ is a unitary then also $U^T$ is unitary and if $O$ is Hermitian then also $O^T$ is Hermitian and such that $O^TO^T\leq\id$. 

Let us denote with $M^{(1)}_{\textup{CHQY}}$ and $M^{(2)}_{\textup{CHQY}}$ the quantities entering~\cref{def: causal influence old} where the operators in the correlation functions are applied respectively on $\hilmaparg{\region{A}_1}\otimes\hilmaparg{\region{B}_1}$ and $\hilmaparg{\region{A}_2}\otimes\hilmaparg{\region{B}_2}$, i.e.,
\begin{equation}
\begin{split}
    M^{(1)}_{\textup{CHQY}}(U_\region{A}:O_\region{B}) = \frac{\bra{L}(U_{\region{A}_1}\otimes O_{\region{B}_1})\rhop (U_{\region{A}_1}^\dagger\otimes O_{\region{B}_1}^\dagger)\ket{L}}{\bra{L}(U_{\region{A}_1}\otimes \id_{\region{B}_1})\rhop (U_{\region{A}_1}^\dagger\otimes \id_{\region{B}_1})\ket{L}},\\
    M^{(2)}_{\textup{CHQY}}(U_\region{A}:O_\region{B}) = \frac{\bra{L}(U_{\region{A}_2}\otimes O_{\region{B}_2})\rhop (U_{\region{A}_2}^\dagger\otimes O_{\region{B}_2}^\dagger)\ket{L}}{\bra{L}(U_{\region{A}_2}\otimes \id_{\region{B}_2})\rhop (U_{\region{A}_2}^\dagger\otimes \id_{\region{B}_2})\ket{L}}.
\end{split}
\end{equation}
Given the above equalities it holds $M^{(1)}_{\textup{CHQY}}(U_\region{A}:O_\region{B}) = M^{(2)}_{\textup{CHQY}}(U^T_\region{A}:O^T_\region{B})$.

There is non-zero causal influence from $\region{A}$ to $\region{B}$ according to $M^{(1)}$ if and only if there exist $U_\region{A}$, $U'_{\region{A}}$ and $O_{\region{B}}$ such that
\begin{equation}
     M^{(1)}_{\textup{CHQY}}(U_\region{A}:O_\region{B})\neq  M^{(1)}_{\textup{CHQY}}(U'_\region{A}:O_\region{B})
\end{equation}
thus if and only if 
\begin{equation}
     M^{(2)}_{\textup{CHQY}}(U^T_\region{A}:O^T_\region{B})\neq  M^{(2)}_{\textup{CHQY}}({U'_\region{A}}^T:O^T_\region{B})
\end{equation}
thus if and only if there is non-zero causal influence according to $M^{(2)}$.
In the argument above, we only considered cases where operators are applied on $\region{A}_1$ and $\region{B}_1$ or $\region{A}_2$ and $\region{B}_2$, but the same argument holds if the operations are applied on $\region{A}_1$ and $\region{B}_2$ or $\region{A}_2$ and $\region{B}_1$. 

Thus, as long as the same choice is applied for both correlation functions in $M$ and for all unitaries and Hermitian operators the result is equivalent lifting the ambiguity.

\section{Mapping cyclic causal models to acyclic models and implications for tensor networks}
\label{app:cyctoacyc}

Our results allow to map a tensor network to a generally cyclic causal model. In a previous work~\cite{Quantum_paper} we further defined a mapping from a cyclic causal model to a family of acyclic causal models with post-selection. Combining both mappings we can map a tensor network to a family of acyclic causal models. 

In this appendix, we first review the mapping from cyclic causal models to acyclic models with post-selection, and the equivalence between conditional probabilities of such acyclic models and with the probability rule in cyclic models of~\cref{prop:probs as self cycles}. Then, we show implications of this mapping for tensor networks, specifically in how one can interpret correlation functions. Such implications are fundamental in analysing causal influence from an operational perspective (see~\Cref{sec: new ci}).

\subsection*{Cyclic causal models as acyclic models with post-selection}
\label{sec:cyclictoacyclic}

Let us recall results from~\cite{Quantum_paper}. We construct a family of acyclic causal models which are obtained through replacing directed edges of a given causal model with so-called \textit{teleportation protocols}, defined as follows:

\begin{definition}[Post-selected teleportation protocol]
\label{def:ps teleportation}
    A post-selected teleportation protocol on finite dimensional spaces $\hilmaparg A\cong \hilmaparg C$ consists of a finite-dimensional Hilbert space $\hilmaparg B$ and a pair $(\telepovm_{AB},\telestate_{BC})$ where $\telepovm_{AB} \in \linops(\hilmaparg A \otimes \hilmaparg B)$ is a POVM element and $\telestate_{BC} \in \linops(\hilmaparg B \otimes \hilmaparg C)$ is a state, such that for all $\rho_A \in \linops(\hilmaparg A)$,
    \begin{align}
    \label{eq:ps teleportation condition}
        \Tr_{AB}[\telepovm_{AB} \rho_A \telestate_{BC}] = \teleprob \rho_C,
    \end{align}
    where $\teleprob \in (0,1]$ is the success probability of the post-selected teleportation protocol. We represent these as
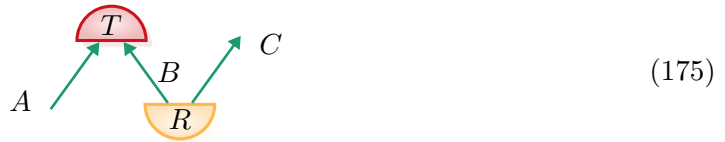
\begin{equation}
    \centertikz{
    \begin{scope}[xscale=1.2]
        \node (A) at (0,0) {};
        \node at (-0.25,0.25) {$A$};
        \node (C) at (2.25,1.25) {};
        \node at (2.5,1.) {$C$};
        \node[prenode] (pre) at (1.5,0) {$R$};
        \node[psnode] (post) at (0.75,1.25) {$T$};
        \draw[qleg] (pre) -- node[midway,right] {$B$} (post);
        \draw[qleg] (A) -- (post);
        \draw[qleg] (pre) -- (C);
    \end{scope}
    }
\end{equation} 
\end{definition}
Such protocol allow to ``simulate'' an  identity channel from $A$ to the isomorphic system $C$. In the appendix of~\cite{Quantum_paper}, the set of $(\telepovm_{AB},\telestate_{BC})$ with this property and their success probabilities are fully characterised. A well-known example of such protocols is the ``Bell'' teleportation protocol~\cite{Bennett1993} with maximally entangled states as pre- and post-selection. Such protocol has the highest success probability and has already been used to simulate closed time-like curves~\cite{Lloyd_2011_2,Lloyd_2011}.

Then, we define the following family of acyclic graphs:

\begin{definition}[Family of acyclic causal graphs $\graphfamily\graphnamedir$]
\label{def: graph_family_v3}
 Given a causal graph $\graphnamedir = (\vertsetdir,\edgesetdir)$, we define an associated family $\graphfamily\graphnamedir$ of directed acyclic causal graphs, where each element $\graphtele \in \graphfamily\graphnamedir$ is obtained from the causal graph $\graphnamedir$ as follows.
 \begin{myitem}
     \item Choose any subgraph $\graphnamedir':=(\vertsetdir',\edgesetdir')$ of $\graphnamedir=(\vertsetdir,\edgesetdir)$ with $\vertsetdir'=\vertsetdir$ and $\edgesetdir'\subseteq \edgesetdir$, such that $\graphnamedir'$ is acyclic. 
 
     \item Include in $\graphtele$ all the vertices and edges of the subgraph $\graphnamedir'$ associated with the same vertex types (observed or unobserved) as the original causal graph $\graphnamedir$. 
     \item Denoting the set of so-called split edges $\splitedges{\graphtele}:=E\backslash E'$, for each edge $\edgearg{\vertname_i}{\vertname_i'} \in \splitedges{\graphtele}$, include in $\graphtele$, two vertices $\postvertname_i$ and $\prevertname_i$ and three edges  $\edgearg{\vertname_i}{\postvertname_i}$, $\edgearg{\prevertname_i}{\postvertname_i}$ and $\edgearg{\prevertname_i}{\vertname_i'}$, such that $\postvertname_i$ is observed and $\prevertname_i$ unobserved.     
 \end{myitem}
 This makes $\graphtele$ identical to $\graphnamedir$ up to replacing each split edge $\edgearg{\vertname_i}{\vertname_i'} \in \splitedges{\graphtele}\subseteq \edgesetdir$ with the teleportation structure:
     
        \begin{equation}
        \label{eq: acyc_graph_split_edge_v3}
            \centertikz
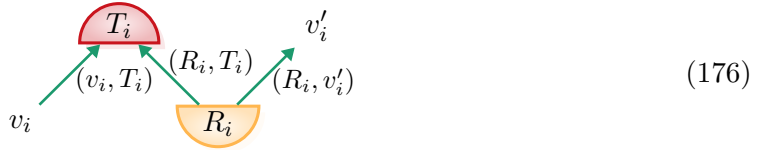
{
            \begin{scope}[xscale=2.6,yscale=1.3]
                \node (in) at (0,0) {$\vertname_i$};
                \node (out) at (1.5,1) {$\vertname_i'$};
                \node[prenode] (pre) at (1,0) {$\prevertname_i$};
                \node[psnode] (post) at (0.5,1) {$\postvertname_i$};
                \draw[qleg] (pre) -- node[pos=0.7,right] {\small$\edgearg{\prevertname_i}{\postvertname_i}$} (post);
                \draw[qleg] (in) -- node[pos=0.4,right=0pt] {\small$\edgearg{\vertname_i}{\postvertname_i}$} (post);
                \draw[qleg] (pre) -- node[pos=0.4,right] {\small$\edgearg{\prevertname_i}{\vertname_i'}$} (out);
            \end{scope}
            }
        \end{equation}
We will refer to every $\graphtele \in \graphfamily\graphnamedir$ as a \textup{teleportation graph}. We denote the set of all post-selection vertices and the set of all pre-selection vertices in $\graphtele$ as $\psvertset:=\{\postvertname_i\}_{\edgearg{\vertname_i}{\vertname_i'} \in \splitedges{\graphtele}}$ and $\prevertset:=\{\prevertname_i\}_{\edgearg{\vertname_i}{\vertname_i'} \in \splitedges{\graphtele}}$.
\end{definition}

On each of these graphs we can define a causal model, yielding the following family of acyclic causal models:

\begin{definition}[Teleportation causal models on the graph family $\graphfamily{\graphnamedir}$]
\label{def:causal model of graphfamily_v3}
 Given a causal model, $\cmG$, associated with a causal graph $\graphnamedir$, we can define a corresponding family of causal models, by associating a causal model $\cm_{\graphtele}$ to each teleportation graph $\graphtele\in \graphfamily\graphnamedir$, as follows: 
 \begin{myitem}
     \item For every edge and every vertex present in both $\graphnamedir$ and $\graphtele$, the assigned Hilbert spaces, outcome sets, CPTP maps and POVMs are the same in the two causal models $\cmG$ and $\cmtele$.
     \item  For each edge $\edgearg{\vertname_i}{\vertname_i'} \in \splitedges{\graphtele}$ that was removed from $\graphnamedir$, the causal model $\cmtele$ has the following specifications: 
     
     \begin{enumerate}
         \item \vspace{-0.5cm} The Hilbert spaces associated to the edges are $\overbrace{
            \hilmaparg{\edgearg{\vertname_i}{\postvertname_i}}
            =
            \hilmaparg{\edgearg{\prevertname_i}{\vertname_i'}}
            }^{\textup{in $\cmtele$}}
            =
            \overbrace{\hilmaparg{\edgearg{\vertname_i}{\vertname_i'}}}^{\textup{in $\cm_{\graphnamedir}$}}$.
          \item The outcome set associated to the post-selection vertex consists of $\outcomemaparg{\postvertname_i} = \{\ok,\fail\}$, the outcome taking values in this set will be denoted as $\postoutcome_i$.
        \item The POVM element $\povmelarg{\ok}{\postvertname_i}$ of the post-selection vertex and the state $\statemaparg{\prevertname_i}$ of the pre-selection vertex form a post-selected teleportation protocol.
     \end{enumerate}
 \end{myitem}
 We will refer to each such $\cmtele$ as a \textup{teleportation causal model}.
\end{definition}

The family of causal models is acyclic, thus for each of them one can use the well-known acyclic probability rule (see~\cite{Henson_2014} and~\cite{Quantum_paper}) to evaluate probabilities.

Specifically, considering a causal model $\cm_{\graphnamedir}$ on a causal graph $\graphnamedir$ with observed vertices $\overtset$ and a teleportation causal model $\cmtele$ on $\graphtele\in\graphfamily\graphnamedir$ obtained from $\cm_{\graphnamedir}$, we define the success probability of post-selection in $\cmtele$ as
 \begin{align}
      \successprob := \probacyc\left(\{\postoutcome_{i} = \ok\}_{\postvertname_i\in\psvertset}\right)_{\graphtele}
            = \sum_{\outcome} \probacyc\left(
           \outcome,
            \{\postoutcome_{i} = \ok\}_{\postvertname_i\in\psvertset}
            \right)_{\graphtele},
\end{align}
where the summation is over $\outcome:=\{\outcome_\vertname\}_{\vertname\in\overtset}$.
If $\successprob>0$, we can evaluate the conditional probability
 \begin{align}
 \label{eq: gen_prob_rule}
     \probacyc\left(\outcome \middle| \{\postoutcome_i = \ok\}_{\postvertname_i\in\psvertset}\right)_{\graphtele}=\frac{\probacyc\left(\outcome, \{\postoutcome_i = \ok\}_{\postvertname_i\in\psvertset}\right)_{\graphtele}}{\successprob}.
 \end{align}

 The following results, proven in~\cite{Quantum_paper}, establish equality of such conditional probability within the set of teleportation graphs:

 \begin{prop}[Equivalent probabilities from different teleportation graphs]
    \label{lemma: acyclic_prob_same_v3}
    Let $\cm_{\graphnamedir}$ be a causal model on a causal graph $\graphnamedir$. 
For all $\graphnamedir_{\textsc{tp},1},\graphnamedir_{\textsc{tp},2} \in \graphfamily\graphnamedir$, it holds $\successprob^{1}\neq 0 \; \iff \; \successprob^{2}\neq0$ and 
  \begin{equation}
      \probacyc\left(\outcome \middle| \{\postoutcome_i = \ok\}_{\postvertname_i\in\psvertset^1}\right)_{\graphnamedir_{\textsc{tp},1}} = \probacyc\left(\outcome \middle| \{\postoutcome_i = \ok\}_{\postvertname_i\in\psvertset^2}\right)_{\graphnamedir_{\textsc{tp},2}},
  \end{equation}
 where $\successprob^{i}$ is the success probability and $\psvertset^i$ is the set of post-selection vertices of $\cm_{\graphnamedir_{\textsc{tp},i}}$.
\end{prop}
A choice of teleportation graph, which we call maximal teleportation graph, $\graphnamedir_{\textsc{tp}}^{\textup{max}}\in\graphfamily{\graphnamedir}$ associated to a causal graph $\graphnamedir=\graphexpl$ consists in choosing $\splitedges{\graphnamedir_{\textsc{tp}}^{\textup{max}}} = \edgesetdir$. The following proposition links the conditional probability in the teleportation causal model over $\graphnamedir_{\textsc{tp}}^{\textup{max}}$ to the cycle composition of~\cref{prop:probs as self cycles}. 

\begin{prop}[General probability rule in terms of self-cycle composition]
    \label{prop:probs as self cycles_v3}
    Consider a causal model $\cm_{\graphnamedir}$ on a causal graph $\graphnamedir$ and the teleportation causal model $\cmtele$ constructed on the maximal teleportation graph $\graphnamedir_{\textsc{tp}}^{\textup{max}}\in\graphfamily{\graphnamedir}$. Considering the total maps of the causal model $\cmG$
    \begin{align}
        \etot_\outcome = \bigotimes_{\vertname\in\overtset} \measmaparg{\outcome_\vertname}{\vertname} \bigotimes_{u\in\uvertset} \chanmaparg{u},
    \end{align}
    it holds for the success probability of post-selection, $\successprob = \left(\prod_{e\in\edgesetdir} \teleprob^{(e)}\right) \sum_{x} \selfcycle(\etot_x)$, and if $\successprob>0$
    \begin{equation}
        \prob(x)_{\graphnamedir} = \frac{\selfcycle(\etot_x)}{\sum_x \selfcycle(\etot_x)} = \probacyc\left(\outcome \middle| \{\postoutcome_i = \ok\}_{\postvertname_i\in\psvertset^{\textup{max}}}\right)_{\graphnamedir_{\textsc{tp}}^{\textup{max}}},
    \end{equation}
    where $\teleprob^{(e)}$ is the teleportation probability associated to the post-selected teleportation protocol implemented instead of the edge $e\in\edgesetdir$. 
\end{prop}

As a corollary of~\cref{lemma: acyclic_prob_same_v3,prop:probs as self cycles_v3} we have 
\begin{equation}
    \prob(x)_{\graphnamedir} = \frac{\selfcycle(\etot_x)}{\sum_x \selfcycle(\etot_x)} = \probacyc\left(\outcome \middle| \{\postoutcome_i = \ok\}_{\postvertname_i\in\psvertset}\right)_{\graphnamedir_{\textsc{tp}}}
\end{equation}
for all teleportation causal model $\cmtele$ on $\graphtele\in\graphfamily\graphnamedir$ obtained from $\cm_{\graphnamedir}$.

In addition, one can prove that the distribution does not depend on which teleportation protocol is used for each edge (see~\cite{Quantum_paper}).

\subsection*{Robustness of operational causal influence in tensor networks}

Now we use the connection between tensor networks to cyclic causal model and further cyclic causal models to acyclic causal models with post-selection to further motivate and show the robustness of the operational causal influence of~\cref{def: causal influence}.

In~\cref{lem:signalling_CI}, we have shown that operational causal influence is preserved as signalling in any image causal model though~\cref{def: TNtoCMgeneral}. Specifically, if we consider causal influence in a tensor network $\tn$ on $\graphnameud$ with link state $\rhop$, because of~\cref{lem:signalling_CI} we have~\cref{eq:establishequallity} (see the proof of~\cref{lem:signalling_CI}):
\begin{equation}
    \begin{split}
        M(U^a:O^{y|b}|\mathcal{O}^b)&= \frac{\bra{L}(U^a_{\labarg{A}_\textup{out}}\otimes O^{y|b}_\labarg{B_\textup{out}})\rhop ({U^a_{\labarg{A_\textup{out}}}}^\dagger\otimes {O^{y|b}_\labarg{B_\textup{out}}}^\dagger)\ket{L}}{\sum_{y'}\bra{L}(U^a_{\labarg{A}_\textup{out}}\otimes O^{y'|b}_\labarg{B_\textup{out}})\rhop ({U^a_{\labarg{A_\textup{out}}}}^\dagger\otimes {O^{y'|b}_\labarg{B_\textup{out}}}^\dagger)\ket{L}}\\
        &=\frac{\selfcycle\left[(\mathcal{U}^a_{\labarg{A}_\textup{out}}\otimes \mathcal{O}^{y|b}_\labarg{B_\textup{out}})\circ\bar{\etot}_{\edgeset_{\textup{out}}A_V|\edgeset_{\textup{in}}A_V}\right]}{\sum_{y'}\selfcycle\left[(\mathcal{U}^a_{\labarg{A}_\textup{out}}\otimes \mathcal{O}^{y'|b}_\labarg{B_\textup{out}})\circ\bar{\etot}_{\edgeset_{\textup{out}}A_V|\edgeset_{\textup{in}}A_V}\right]} = \prob(y|a,b)_I,
    \end{split}
    \end{equation}
    where $\bar{\etot}_{\edgeset_{\textup{out}}A_V|\edgeset_{\textup{in}}A_V}$ is the marginalised total map of the image causal model $\cm = \mapTNtoCMgen_D(\tn)$ for some edge directions $D\in\{0,1\}^{|\edgesetud|}$.

    As discussed briefly in~\Cref{sec:signalling}, for each pair of labels $a$ and $b$, one can construct a causal model on a slightly modified graph, which we call intervention graph, and was already used in the examples of~\Cref{sec:signalling}. Without entering the details of such a causal model, one can think of it as being identical to the original one up to introducing vertices associated to $\labarg{A}$ and $\labarg{B}$ and respectively associating to them the unitary channel $\mathcal{U}^a$ and the instrument $\mathcal{O}^b=\{O^{y|b}\}_{y\in\mathcal{Y}}$. Then, the probability distribution of this generally cyclic causal model equals the statistics of intervention $\prob(y|a,b)_I$ given in~\cref{def:statisticsint}.

    Given the results of~\cite{Quantum_paper} reviewed before, the probability $\prob(y|a,b)_I$ can be equivalently be expressed as a distribution of an acyclic causal model, with an additional post-selection. Specifically, as outlined in~\cref{def: graph_family_v3,def:causal model of graphfamily_v3}, we can construct a family of acyclic causal models, each on a graph $\graphtele\in\graphfamily{\graphnamedir_D^{\circlearrowleft}}$, where
    \begin{equation}
        \prob(y|a,b)_I = \probacyc\left(y \middle|a,b, \{\postoutcome_i = \ok\}_{\postvertname_i\in\psvertset}\right)_{\graphtele}
    \end{equation}
    where $\psvertset$ is the set of post-selection vertices in $\graphtele$ each associated with a binary random variable $T_i$. Notice that the setting labels $a$ and $b$ are just kept to keep track of the specific intervention which we consider.

    In addition, in~\Cref{sec: new ci} we interpreted the form of correlations functions in terms of a virtual binary measurement on the link state. Indeed, we defined a binary measurement $\{\ketbra{L},\id-\ketbra{L}\}$, associated with outcomes $\{L,\neg L\}$, and reinterpreted, for instance, the tensor network contraction as a probability of observing outcome corresponding to the POVM element $\ketbra{L}$ on the state $\rhop$, i.e., $\bra{L}\rhop\ket{L}\propto P(L)_{\rhop}$. This, allowed us to understand the operational causal influence as a conditional probability (\cref{eq:virtualmmt}):
    \begin{equation}\begin{split}
     M(U_\region{A}:O^y_\region{B}|\mathcal{O}_\region{B}) &= \frac{\bra{L}(U_\region{A}\otimes O^y_\region{B})\rhop (U_\region{A}^\dagger\otimes {O^y_\region{B}}^\dagger)\ket{L}}{\sum_{y}\bra{L}(U_\region{A}\otimes O^y_\region{B})\rhop (U_\region{A}^\dagger\otimes  {O^y_\region{B}}^\dagger)\ket{L}}\\&=\frac{P(L,y|\mathcal{O}_{\region{B}})_{U_\region{A}\rhop U^\dagger_{\region{A}}}}{\sum_{y}P(L,y|\mathcal{O}_{\region{B}})_{U_\region{A}\rhop U^\dagger_{\region{A}}}}
     = P(y|L,\mathcal{O}_{\region{B}})_{U_\region{A}\rhop U^\dagger_{\region{A}}}.
\end{split}
\end{equation}

Now we can make this argument precise. Indeed, we have:
   \begin{equation}
   \begin{split}
     P(y|L,\mathcal{O}_{\region{B}})_{U_\region{A}\rhop U^\dagger_{\region{A}}}&=M(U^a_\region{A}:O^{y|b}_\region{B}|\mathcal{O}_\region{B}) \\
     &=\prob(y|a,b)_I = \probacyc\left(y \middle|a,b, \{\postoutcome_i = \ok\}_{\postvertname_i\in\psvertset}\right)_{\graphtele}.
\end{split}
\end{equation}
Thus, showing that conditioning on the outcome $L$ of the virtual measurements can be precisely understood as conditioning on the successful post-selection in the family of acyclic teleportation causal models derived from $\cm=\mapTNtoCMgen_D(\tn)$. This holds for all teleportation graphs in the family and for a more general class of pre and post selections than just the Bell ones\footnote{More precisely, all those pairs of pre and post-selections that constitute a post-selected teleportation protocol, \cref{def:ps teleportation}.} as implied by previous robustness results of \cite{Quantum_paper} reviewed in \Cref{app:cyctoacyc}. For the present case, it highlights that the causal influence quantity $M$,  in any tensor network, has a robust operational interpretation in terms of conditional probabilities across a range of related acyclic causal models. 

\section{Review of graph separation theorems}
\label{app:graphsep}

    Here, we review some well-known results in causal modelling which allow to determine no-signalling relations directly from the structure of the graph. Specifically, we these are properties of directed graphs which allow to ``read-off'' conditional independencies of probabilities of a causal model on such graph directly from the connectivity of it. 

    For completeness, let us first define conditional independence, then introduce two graph separation properties and the corresponding theorems.

\begin{definition}[Conditional independence]
\label{def:conditional independence}
    Let $V$ be a non-empty finite set, and
    let $P(x)$ be a joint probability distribution over a set $X =
    \{X_v\}_{v \in V}$ of random variables.
    Let $V_1$, $V_2$ and $V_3$ be three disjoint subsets of $V$, with $V_1$ and $V_2$ being non-empty. 
    We denote the corresponding sets of random variables as $X_i = \{ X_v \}_{v\in V_i}$ taking values corresponding values as $x_i = \{ x_v \}_{v\in V_i}$ for $i \in \{1,2,3\}$.
    We say that $X_1$ is conditionally independent of $X_2$ given $X_3$ and denote it as $(X_1\indep X_2|X_3)_{P}$ if, for all $x_1, x_2,x_3$, it holds that $P (x_1,x_2|x_3)=P(x_1|x_3)P(x_2|x_3)$.
\end{definition}

\subsection*{The acyclic case: $d$-separation}
    \begin{definition}[$Z$-$d$-open path]

Let $\graphnamedir=(\vertsetdir,\edgesetdir)$ be a directed graph and $Z\subseteq \vertsetdir$, consider a path $\pi: \vertname_1 \leftrightharpoons \dots \leftrightharpoons \vertname_n$ in $\graphnamedir$. Then, we say that $\pi$ is $Z$-$d$-open if
\begin{myitem}
    \item the endpoints are not in $Z$, i.e., $v_1,v_n\notin Z$;
    \item every triple of adjacent nodes is of one of the following forms:
    \begin{enumerate}
        \item \textup{collider:} $v_{i-1}\rightarrow v_i \leftarrow v_{i+1}$ with $v_i\in Z$ or $Z\cap \textup{Desc}(v_i)\neq \varnothing$;
        \item \textup{fork:} $v_{i-1}\leftarrow v_i \rightarrow v_{i+1}$ with $v_i\notin Z$;
        \item \textup{chain:} $v_{i-1}\rightarrow v_i \rightarrow v_{i+1}$ or $v_{i-1}\leftarrow v_i \leftarrow v_{i+1}$  with $v_i\notin Z$.
    \end{enumerate}
\end{myitem}
A path that is not $Z$-$d$-open is called $Z$-$d$-blocked.    
\end{definition}

\begin{definition}[$d$-separation]
\label{def: d-sep}
    Let $\graphnamedir=(\vertsetdir,\edgesetdir)$ be a directed graph and $V_1$, $V_2$ and $V_3$ three disjoint subsets of vertices of $\graphnamedir$ with $V_1$ and $V_2$ being non-empty, then we say that $V_1$ and $V_2$ are $d$-separated conditioned on $V_3$, and denote it as $(V_1\perp^d V_2|V_3)_{\graphnamedir}$ if and only if for all $v_1\in V_1$ and $v_2\in V_2$ all paths with endpoints $v_1$ and $v_2$ are $V_3$-$d$-blocked.
\end{definition}

Then the next theorem follows from the theory-independent $d$-separation theorem of~\cite{Henson_2014}, when restricted to the case of quantum theory.

\begin{theorem}[$d$-separation theorem for acyclic graphs]
\label{theorem: dsep theorem}
Consider a directed acyclic graph $\graphnamedir$ and let $V_1$, $V_2$ and $V_3$ be any three disjoint sets of the vertices of $\graphnamedir$ with $V_1$ and $V_2$ being non-empty. 
    Then, the following holds:
    \begin{itemize}
        \item[]\textbf{\textup{(Soundness)}} For any causal model $\cm$ on $\graphnamedir$ where the sets $V_i$ are observed, we have that $d$-separation between the vertex sets $V_i$ implies conditional independence for the corresponding sets of random variables $X_i:=\{X_{\vertname}\}_{\vertname\in V_i}$ where $i\in\{1,2,3\}$, i.e.,
        \begin{equation}
         (V_1\perp^d V_2|V_3)_{\graphname} \implies (X_1\indep X_2|X_3)_{\probfacyc_{\graphname} }.
        \end{equation} 
        \item[]\textbf{\textup{(Completeness)}} If the $d$-connection $(V_1\not\perp^d V_2|V_3)_{\graphname}$ holds in $\graphnamedir$, then there exists a causal model  $\cm$ on $\graphnamedir$ such that the sets $V_i$ are observed and ${(X_1\not \indep X_2|X_3)_{\prob{\graphname}}}$..
    \end{itemize}
    The above conditional (in)dependence statements are relative to the marginal $\probacyc(x_1,x_2,x_3)_{\graphnamedir}$.
\end{theorem}
\begin{proof}
    See~\cite{Henson_2014}.
\end{proof}

\subsection*{The cyclic case: $p$-separation}
In the cyclic case, the notion of $p$-separation is introduced as a sound and complete graph-separation property. This relies on the family of acyclic causal models defined in~\cref{def: graph_family_v3}.
\begin{definition}[$p$-separation]
\label{def: p-separation}
Let $\graphnamedir$ be a directed graph and $V_1$, $V_2$ and $V_3$ denote any three disjoint subsets of the vertices of $\graphnamedir$ with $V_1$ and $V_2$ being non-empty. Then, denoting $p$-separation as $\perp^p$ and $p$-connection as $\not\perp^p$
 \begin{equation}
 \begin{aligned}
(V_1\perp^pV_2|V_3)_{\graphnamedir} \iff \exists \graphtele \in \graphfamily{\graphnamedir} \st (V_1\perp^d V_2|V_3 \cup \psvertset)_{\graphtele}, \\     
(V_1\not\perp^pV_2|V_3)_{\graphnamedir} \iff \forall \graphtele \in \graphfamily{\graphnamedir} \st (V_1\not\perp^d V_2|V_3 \cup \psvertset)_{\graphtele}.
 \end{aligned}
 \end{equation}
\end{definition}

Then, the following theorem establishes soundness and completeness of $p$-separation for all finite dimensional cyclic causal models.
\begin{theorem}[$p$-separation theorem]
\label{theorem: psep_theorem}
Consider a directed graph $\graphnamedir$ and let $V_1$, $V_2$ and $V_3$ be any three disjoint sets of the vertices of $\graphnamedir$ with $V_1$ and $V_2$ being non-empty. 
    Then, denoting conditional independence with $\indep$ the following holds:
    \begin{itemize}
        \item[]\textbf{\textup{(Soundness)}} For any causal model $\cm$ on $\graphnamedir$ where the sets $V_i$ are observed and associated with random variables $X_i:=\{X_{\vertname}\}_{\vertname\in V_i}$, we have
        \begin{equation}
         (V_1\perp^p V_2|V_3)_{\graphnamedir} \implies (X_1\indep X_2|X_3)_{\prob_{\graphnamedir} }.
        \end{equation} 
        \item[]\textbf{\textup{(Completeness)}} If the $p$-connection $(V_1\not\perp^p V_2|V_3)_{\graphnamedir}$ holds in $\graphnamedir$, then there exists a causal model $\cm$ such that the sets $V_i$ are observed and the associated random variables satisfy ${(X_1\not \indep X_2|X_3)_{\prob_{\graphnamedir}}}$.
    \end{itemize}
\end{theorem}
\begin{proof}
    See~\cite{Quantum_paper}.
\end{proof}

\end{document}